\documentclass{amsart}

\usepackage{courier}

\usepackage{graphicx}
\usepackage{amsmath,amssymb}
\usepackage{amscd}
\usepackage{verbatim}
\usepackage{enumerate}

\newtheorem{theorem}{Theorem}[section]
\newtheorem{Proposition}[theorem]{Proposition}
\newtheorem{lemma}[theorem]{Lemma}

\newtheorem{corollary}[theorem]{Corollary}

\newtheorem{remark}[theorem]{Remark}

\newtheorem{definition}[theorem]{Definition}

\newtheorem{assumption}[theorem]{Assumption}

\newcommand{\g}{\mathfrak{g}}

\newcommand{\RR}{\mathbb{R}}

\newcommand{\ZZ}{\mathbb{Z}}
\newcommand{\cB}{\mathcal B}
\newcommand{\LL}{{\mathcal L}}
\newcommand{\M}{{\mathcal M}}
\newcommand{\MM}{\mathcal{M}}
\newcommand{\F}{{\mathcal F}}
\newcommand{\cF}{{\mathcal F}}
\newcommand{\FF}{\mathcal{F}}
\newcommand{\CL}{{\mathcal L}}

\newcommand{\Vect}{\mathrm{Vect}}
\newcommand{\gh}{\mathrm{gh}}

\newcommand{\pa}{\partial}
\newcommand{\dd}{\partial}

\newcommand{\Om}{\Omega}

\newcommand{\cA}{\mathcal A}

\newcommand{\ev}{\mathrm{ev}}
\newcommand{\bt}{\bullet}
\newcommand{\ra}{\rightarrow}
\newcommand{\mr}{\mathrm}
\newcommand{\calA}{{\mathcal{A}}}
\newcommand{\xra}{\xrightarrow}

\newcommand{\be}{\begin{equation}}
\newcommand{\ee}{\end{equation}}

\newcommand{\EL}{\mathcal{EL}}
\newcommand{\cL}{\mathcal{L}}
\newcommand{\cM}{\mathcal{M}}
\newcommand{\tX}{\tilde{X}}
\newcommand{\teta}{\tilde{\eta}}
\newcommand{\cS}{\mathcal{S}}

\newcommand{\rel}{\mathrm{rel}}
\newcommand{\p}{\mathrm{symp}}
\newcommand{\ver}{\mathrm{vert}}
\newcommand{\abs}{Q}

\DeclareMathOperator{\Map}{Map}
\DeclareMathOperator{\Mor}{Mor}

\begin{document}

\title{Classical BV theories on manifolds with boundary}
\author{Alberto S. Cattaneo}
\address{Institut f\"ur Mathematik,
Universit\"at Z\"urich,
Winterthurerstrasse 190,
CH-8057, Z\"urich, Switzerland}

\email{alberto.cattaneo@math.uzh.ch}

\author{Pavel Mnev}
\address{Institut f\"ur Mathematik,
Universit\"at Z\"urich,
Winterthurerstrasse 190,
CH-8057, Z\"urich, Switzerland}

\email{pavel.mnev@math.uzh.ch}

\author{Nicolai Reshetikhin}
\address{Department of Mathematics,
University of California, Berkeley
California 94305,
USA}

\email{reshetik@math.berkeley.edu}
\date{\today}
\maketitle

\begin{abstract} In this paper we extend the classical BV framework to gauge theories
on spacetime manifolds with
boundary. In particular, we connect the BV construction in the bulk with the BFV construction on the boundary
and we develop its extension to strata of higher codimension in the case of manifolds with corners.
We present several examples including electrodynamics, Yang-Mills theory
and topological field theories coming from the AKSZ construction, in particular,
the Chern-Simons theory, the $BF$ theory, and the Poisson sigma model. This paper is the first step towards
developing the perturbative quantization of such theories on manifolds with boundary in a way consistent
with gluing.
\end{abstract}

\tableofcontents

\section{Introduction}

One of the key features in quantum field theory is locality.
Physically, it is based on the concept of an ideal point-like particle. Mathematically, it means that the action in the corresponding classical field theory should be a local functional.

In quantum field theory locality is expected to correspond to the gluing
property of the partition function for manifolds with boundary.
For a quantum field theory on  Minkowski
cylinders this implies the usual composition law for
evolution operators. In topological and conformal field theories the locality as a gluing property
of partition functions was discussed in  \cite{At,Seg,W}.

In this paper we formulate the classical Batalin-Vilkovisky (BV) framework
for theories on manifolds with boundary. This combines the
properties of the standard BV \cite{BV} and the Batalin-Fradkin-Vilkovisky (BFV) \cite{BFV} frameworks.
Recall that standard BV theories are odd symplectic extensions of classical gauge theories by
ghosts, anti-fields and antighosts. The gauge symmetry appears as a cohomological Hamiltonian
vector field whose Hamiltonian function is an extension of the classical action. The BFV theories
give a similar cohomological description of the gauge symmetries on the boundary. Another relation between BV and BFV theories
is described in \cite{Grig}.

This paper should be considered as a first step towards the perturbative quantization of
classical BV theories on manifolds with boundary, and possibly with corners.

\vspace{1cm}

In case of topological field theories the ultimate goal is to construct perturbative topological invariants
of manifolds with boundary which satisfy the cutting-gluing principle.
All known classical topological field theories are gauge theories.
Examples are the $BF$ theory, the Chern-Simons theory, and the Poisson sigma model (see Sections \ref{sec: examples} and \ref{examples-aksz} for definitions and details).

Recall some basic facts about BV theories. These are classical
field theories with $\ZZ$-graded spaces of fields which appear naturally
in quantization of gauge theories. The BV construction is a generalization
of the Faddeev-Popov \cite{FP} and BRST \cite{BRST} methods. The BV construction
is particularly important in cases when the gauge symmetry is given by
a non-integrable distribution on the space of fields or when the gauge symmetry
is reducible with stabilizers of varying rank.
Examples of such theories include the Poisson sigma model and
the non-abelian $BF$ theory when the dimension of spacetime is greater then 3.
Examples of theories where the gauge
symmetry is an integrable distribution include Yang-Mills and Chern-Simons.
Even in cases when the distribution is integrable,
BV formalism is useful because it is compatible with Wilsonian
renormalization of Feynman diagrams, see e.g. \cite{Co_renorm}.

Topological field theories such as Chern-Simons, $BF$ theories and
the Poisson sigma model are special for several reasons.
First, in the perturbative quantization of such theories Feynman diagrams have no ultraviolet divergencies.
Another reason is that the natural choice of gauge condition, the Lorenz gauge, is
induced by a choice of metric on the spacetime manifold.
Thus the gauge independence of the Feynman diagram expansion in such theories implies that
the partition function is metric independent and therefore is a topological invariant.

All these three examples, Chern-Simons, $BF$ and the Poisson sigma model, fit into the general AKSZ construction \cite{AKSZ}. The AKSZ construction also
gives other examples of topological field theories such as the Rozansky-Witten model and others
\cite{CQZ}.

There is a consensus that perturbative
quantization of the classical Chern-Simons theory
gives the same asymptotical expansions as the combinatorial
topological field theory based on quantized universal enveloping
algebras at roots of unity \cite{RT}, or, equivalently,
on the modular category corresponding to the
Wess-Zumino-Witten conformal field theory \cite{W,MS}
with the first semiclassical computations involving torsion made in \cite{W}.
However this conjecture is still open despite a number of important results in this direction, see for example\cite{Roz,And}.

One of the reasons why the conjecture is still
open is that for manifolds with boundary the perturbative quantization of Chern-Simons theory  has not been developed yet. On the other hand,
for closed manifolds the perturbation theory involving Feynman
diagrams was developed in \cite{GM,FK,BN} and in \cite{AxS,K1,BC}. For the latest development see \cite{CM}.
Closing this gap and developing the perturbative 
quantization of Chern-Simons theory for manifolds with boundary is
one of the main motivations for the project started in this paper.


Having laid the framework of classical BV-BFV theory, in a forthcoming work we plan to discuss the formal perturbative quantization of BV theories on manifolds with boundary. The main difference with respect to the formal perturbative quantization
for closed spacetime manifolds is that in addition to a choice of gauge fixing in the bulk, one has to fix a polarization in the space of boundary fields. Given such choices, the partition function is a function on the space of
leaves of boundary polarization. Its value on a leaf
is the path integral (understood as a series of Feynman diagrams) over fluctuations lying in the gauge fixing submanifold
with boundary values being in a fixed leaf of this polarization. The resulting partition function is a state in the space of boundary states determined by the polarization. We expect that the equation (\ref{CME}) will be  replaced by a deformed quantum master equation for this state, as it happens in the example of 1d Chern-Simons theory \cite{AM}.
Gluing along a boundary component will be simply given by pairing the partition functions in the space of boundary states.

A classical BV-BFV theory is a functor from the spacetime category to the BFV category.
One should expect a similar description in the quantum case. Classically it extends to higher categories, and we anticipate
that quantum counterparts of such extensions are higher category versions of
topological quantum field theories.

In this paper spacetime manifolds are always compact oriented, possibly with boundary and corners.

In this paper we will use freely the notion of a $\ZZ$-graded manifold. See \cite{S-manif,Gr-manif}, and the Appendix \ref{graded} for details. To simplify the terminology we will call $\ZZ$-graded manifolds simply graded manifolds.
By a point in a graded manifold we understand the application of the functor of points.

\vspace{1cm}

This paper is organized in the following manner.
In section \ref{bv-cl} we recall the basic structure of classical BV theory for closed spacetime manifolds.

Section \ref{BV-bound} is the central part of the paper.
Here we formulate the BV-BFV framework for gauge theories.
In this section we propose the analog of the Classical Master Equation for spacetime manifolds with boundary. We define the relevant moduli spaces and address the problem of composing them  when
we glue two spacetime manifolds together. We also show that first order BV theories naturally extend to
lower dimensional strata.

In section \ref{bfv-category}  we introduce the notion of BFV category. We also define the classical BV theory as a
functor from the spacetime category to the BFV category.

In section \ref{sec: examples} we describe several examples of BV-BFV extensions of classical gauge theories: electrodynamics, Yang-Mills theory, scalar field and the abelian $BF$ theory. The latter is in fact topological and is an example of the AKSZ construction.



The AKSZ construction which provides a class of examples of topological BV-BFV theories
is described in section \ref{sec-aksz}. This construction is determined by the choice of a target manifold, which should be a Hamiltonian differential graded manifold. An AKSZ theory extends to strata of all  codimensions.

Further examples of BV-BFV extensions of AKSZ theories are described
in section \ref{examples-aksz}. We start with abelian and non-abelian Chern-Simons theories, then we describe the non-abelian $BF$ theory 
and the Poisson sigma model.

In Appendix \ref{co-isotr} we review some useful facts about coisotropic submanifolds and reduction. In Appendix \ref{graded} we recall some basic facts on graded manifolds. A brief discussion of smooth points on moduli spaces associated to differential graded manifolds is given in Appendix \ref{smooth}. Elements of Cartan calculus for local
forms and functionals on mapping spaces are developed in Appendix \ref{Cart-Calc} and applications to AKSZ theories are presented.

\vspace{1cm}

{\bf Acknowledgements} The authors want to thank J. Andersen, C. Arias-Abad, V. Fock,  T. Johnson-Freyd, D. Kazhdan, F. Sch\"atz, P. Teichner, C. Teleman,   M. Zambon for stimulating discussions. Authors are grateful
to the QGM center in Aarhus, where the work started, for the hospitality.
A.C. and P. M. are grateful for hospitality to the KdV institute of the University of Amsterdam and to the Department of Mathematics of the University of California at Berkeley. The work of A.C. was partially supported by SNF grant No. 200020-131813/1. P.~M. acknowledges partial support of RFBR Grants
Nos.~11-01-00570-a and 11-01-12037-ofi-m-2011 and of SNF Grant No.~200021-137595. The work of N.R. was partially supported by the NSF grant DMS-0901431, by the Danish Research Foundation through the Niels Bohr visiting professor initiative at Aarhus University and by
the Chern-Simons and Hewlett chair at UC Berkeley.

\section{Classical BV theories for closed spacetimes}\label{bv-cl}

\subsection{Non-reduced theory}

When the spacetime manifold $N$ is closed, the space of fields
$\F_N$ in a BV theory is a graded  manifold with
a symplectic form $\omega_N$ of degree $-1$ and the action functional
$S_N$ of degree zero. Usually fields are sections of or connections on a fiber bundle on $N$.
Some basic facts on graded manifolds
are summarized in Appendix \ref{graded}. The degree in the grading on the space of fields is usually called the ghost number. Thus,
$gh(\omega_N)=-1$ and $gh(S_N)=0$.
The space of fields is usually infinite dimensional.
The BV data should satisfy the following property: the Poisson bracket of the action functional with itself is zero \cite{BV,BFV} (see also \cite{AS1,Sev}).
This property is known as the classical master equation (CME).
The action functional generates the Hamiltonian vector
field $Q_N$ of degree $1$. The CME implies that
the Lie derivative of this vector field squares to zero.

Most interesting for applications both in physics and in
geometry and topology are local theories. The
notion of locality is more transparent for theories on spacetime manifolds with
boundary, see section \ref{BV-bound}.  For closed manifolds locality implies natural isomorphisms $\cF_{N_1\sqcup N_2}\simeq
\cF_{N_1}\times \cF_{N_2}$ and the additivity of the symplectic form and of the action: $\omega_{N_1\sqcup N_2}=\omega_{N_1}+
\omega_{N_2}$, $S_{N_1\sqcup N_2}=S_{N_1}+
S_{N_2}$.

We will use the following definition of classical BV theory for closed spacetime manifolds.

\begin{definition} For a closed spacetime $N$ a BV theory \cite{BV} on the space of fields
$\F_N$ is a triple $(\omega_N, Q_N, S_N)$ where $\omega_N$
is a symplectic form on $\F_N$ with $gh=-1$, $Q_N$ is
a vector field on $\F_N$ with $gh=1$, and $S_N$ is a function
on $\F_N$ with $gh=0$, subject to
\[
\iota_{Q_N} \omega_N= (-1)^n\delta S_N, \ \  L_{Q_N}^2=0
\]
where $L_Q$ is the Lie derivative for the vector field $Q$.
\end{definition}
These two equations imply the {\it classical master
equation}
\[
\{S_N, S_N\}=0 ,
\]
which can also be written as
\[
L_{Q_N}S_N=0
\]
The first equation in the definition means that the vector field  $Q_N$ is Hamiltonian
and therefore it preserves the symplectic form:
\[
L_{Q_N}\omega_N=0,
\]

Critical points of the action functional $S_N$ are solutions
to Euler-Lagrange (EL) equations
\begin{equation}\label{el}
\delta S_N=0,
\end{equation}

Because the space of fields is an odd-symplectic manifold,
solutions to the EL equations are also zeroes of the
Hamiltonian vector field generated by $S_N$, i.e. of
the cohomological vector field $Q_N$. Denote the
space of solutions to Euler-Lagrange equations (\ref{el}) by
$\EL_N$.

Because vector field $Q_N$ is Hamiltonian, and Poisson
brackets of local functionals are defined, we can write
$L_{Q_N}F=\{S_N, F\}$ for the action of the vector field $Q_N$
on the functional $F$.  Denote by $V_F$ the vector field
generated by $F$, then $L_{V_F}G=\{F,G\}$. It is clear that
vector fields $[Q_N, V_F]$ annihilate the functional $S_N$:
\[
L_{[Q_N, V_F]}S_N=\{\{S_N,F\},S_N\}=0
\]
This follows from CME and from the Jacobi identity.
Thus, any such vector field can be considered as a local (off-shell)
gauge symmetry of $S_N$.

From now on we will consider only first order BV theories.
In such theories the action functional is local and it is linear
in derivatives of fields. It is well known that any local theory
can be reformulated as a first order theory by adding corresponding ``momenta" variables in fields.

\subsection{Other degrees}\label{o-degr}
So far we have assumed the symplectic form $\omega_{N}$ to have degree $-1$. But everything can be generalized to the case when the degree is some integer $k$ (we fix however the degree of the cohomological vector field $Q_{N}$ to be $1$). In this setting the action $S_{N}$ will have degree $k+1$ \cite{Royt}. Only a few remarks are in order
\begin{enumerate}
\item Unless $k=0$, $\omega_{N}$ is automatically exact (so for $k=0$ we need this extra assumption). In the following, we will be interested in specifying a primitive $1$-form of the same degree $k$.
\item Unless $k=-1$, the condition $L_{Q_{N}}\omega=0$ implies that $Q_{N}$ is Hamiltonian and that the Hamiltonian function $S_{N}$ is uniquely determined.
\item Unless $k=-2$, the condition $Q_{N}^{2}=0$ implies $\{S_{N},S_{N}\}=0$.
\end{enumerate}
We will mainly be interested in the cases $k=-1$ (BV manifolds) and $k=0$ (BFV manifolds). For extended BV theories, we will also be interested in $k$ positive. We are not aware of any application of the cases $k<-1$.

\subsection{The $Q$-reduction}

The vector field $Q_N$ vanishes on the subspace $\EL_N\subset \cF_N$. It also defines the Lie subalgebra $\widetilde{Vect}_Q$ of the Lie algebra of vector fields on $\cF_N$ generated by vector fields which are Lie brackets with $Q_N$, i.e. by vector fields of the form $[V,Q_N]$. This Lie subalgebra in general does not determine a distribution on $\cF_N$. However, the restriction of $\widetilde{Vect}_Q$ to $\EL_N$ defines an integrable distribution $Vect_Q$.

\begin{definition} The $Q$-reduction $\EL_N/Q$ of $\EL_N$ is the space of leaves of the distribution $Vect_Q$.
\end{definition}

Assume that $X$ is a zero of $Q_N$ such that the leaf of $Vect_Q$ through $X$ is a
smooth point in $\EL_N/Q$.  Define the linear operator $\hat{Q}_X:
T_X\F_N\to T_X\F_N$ as the linearization of $Q_N$ at $X$. If we write $Q_N$ in local coordinates as $Q_N=\sum_a q^a(X)\pa_a$,
the operator $\hat{Q}_X$ acts on the basis $\pa_a$ in $T_X\F_N$ as
\[
\hat{Q}_X \pa_a=\sum_b \pa_a q^b(X)\pa_b
\]
Because $Q_N$ squares to zero we have $\hat{Q}_X^2=0$.
By definition $T_X\EL_N=\ker(\hat{Q}_X)$. Also observe that
$Vect_Q|_X=Im (\hat{Q}_X)\subset T_X\EL_N$. From now on for geometric
considerations we will focus only on smooth points as defined in the Appendix \ref{smooth}.

\begin{Proposition}\label{QisS} For a smooth point $[X]\in \EL_N/Q$,
there is a natural linear isomorphism
\[
T_{[X]}\EL_N/Q=ker(\hat{Q}_X)/Im (\hat{Q}_X)
\]
\end{Proposition}

\subsection{The symplectic reduction}\label{bv-sym-red}

The following Proposition says that if the subspace $\EL_N$
is a submanifold then it is a coisotropic submanifold.

\begin{Proposition} \label{el-coiso}
If $X\in \EL_N$ is a smooth point, then $T_X \EL_N\subset T_X\F_N$
is a coisotropic subspace.
\end{Proposition}
\begin{proof}

The invariance of $\omega_N$ with respect to the vector field
$Q_N$ means that $L_{Q_N}\omega_N=0$. If $X$ is a zero of $Q_N$
this implies $\omega_N(\hat{Q}_X \xi, \eta)=\omega_N(\xi, \hat{Q}_X \eta)$.

These properties of $\hat{Q}_X$ imply $\omega_N(\hat{Q}_X \xi, \hat{Q}_X \eta)=\omega_N(\xi, \hat{Q}_X^2 \eta)=0$. Therefore
$Im(\hat{Q}_X)$ is an isotropic subspace and therefore
its symplectic orthogonal $Im(\hat{Q}_X)^{\perp}$ is coisotropic.

On the other hand
\begin{multline}
Im(\hat{Q}_X)^{\perp}=\{\xi \;|\; \omega_N(\xi, \hat{Q}_X\eta)=0, \mbox{ for any } \eta\}\\
=\{\xi \;|\; \omega_N(\hat{Q}_X\xi, \eta)=0, \mbox{ for any } \eta\}=
\ker(\hat{Q}_X)
\end{multline}

This proves that $T_X\EL_N=\ker(\hat{Q}_X)$ is a coisotropic
subspace. Because we used only nondegeneracy of the form
$\omega_N$ on $T_X\cF$, the proof works both in finite-dimensional and in the infinite-dimensional case.

\end{proof}

Recall the definition of the symplectic reduction  of a coisotropic submanifold of a symplectic manifold.
The symplectic reduction $\underline{C}$ of the coisotropic submanifold $C$ in the symplectic manifold $S$ is the space of leaves of the characteristic foliation of $C$.
The characteristic foliation of $C$ is spanned by Hamiltonian vector fields generated by the vanishing ideal $I_C$ (the ideal in the commutative algebra of functions on $S$ generated by functions vanishing on $C$). When $\underline{C}$ is a smooth manifold, it is a symplectic manifold. Otherwise only at smooth points the tangent
space to $\underline{C}$  has a natural symplectic form. If $[X]$ is such a smooth point, the tangent space $T_{[X]} \underline{C}$ is naturally isomorphic to $T_XC/(T_XC)^{\perp}$.

\begin{remark} An alternative proof of the Proposition \ref{el-coiso} is algebraic. Observe that the vanishing ideal
for $\EL_N$ is generated by functionals of the form $\{S_N, T\}$ where $\{.,.\}$ is the Poisson bracket for the form $\omega_N$ and
$T$ is a local functional.
The simple calculations shows that
\[
\{\{S, U\}, \{S, T\}\}=\{S, \{U, \{S, T\}\}\}
\]
The bracket $\{U, \{S, T\}\}$ is a local functional if $U$ and $T$ are. Therefore the vanishing ideal is a Poisson algebra and therefore $\EL_N$ is coisotropic \emph{ in the algebraic sense}.
\end{remark}

In Proposition \ref{el-coiso} we proved that $Im(\hat{Q}_X)^\perp=\ker(\hat{Q}_X)$. This
implies $Im(\hat{Q}_X)\subset\ker(\hat{Q}_X)^\perp$, which means that the $\hat{Q}$-distribution on
$\EL_N$ is contained in the characteristic distribution. In the finite dimensional case the inclusion
becomes an equality.
\begin{assumption}
In classical field theory we set
\begin{equation}\label{assump}
Im(\hat{Q}_X)=\ker(\hat{Q}_X)^\perp ,
\end{equation}
as an assumption.
\end{assumption}

This corresponds to certain ellipticity condition for $Q_N$.
In all our examples this condition follows from the usual Hodge-de Rham decomposition.

\begin{remark} \label{rem: symp red = Q-red BV w/o bdry}As a consequence of condition (\ref{assump}) the characteristic foliation of $\EL_N$ is the
same as the foliation by $Vect_Q$, and in particular, if $[X]\in \underline{\EL}_N$ is a smooth point, then
\begin{equation}\label{tsp}
T_{[X]}\underline{\EL}_N=\ker(\hat{Q}_X)/Im(\hat{Q}_X)
\end{equation}
In other words, locally the symplectic reduction of $\EL_N$ coincides with its $Q$-reduction: $\underline{\EL}_N=\EL_N/Q$.

\end{remark}

We will call  $\underline{\EL}_N=\EL_N/Q$ the $\EL$-moduli space for $N$.
If the cohomology spaces of $\hat{Q}_X$ are
finite dimensional at a generic point $X$, the $\EL$-moduli space
is finite dimensional, but possibly singular.

\begin{remark} With the appropriate assumptions the ring of functions on the reduced space $\EL_N/Q$ is isomorphic to the  cohomology space of the ring of functions on $\F_N$ with the differential $L_{Q_N}$ (the Lie derivative with respect to $Q_N$). To be more specific, one should expect that the smooth locus
of $\EL/Q$ is isomorphic to the smooth locus of $Spec(H_Q(Fun(\cF_N)))$ with the appropriate definition of
$Spec$ and of the ring of functionals $Fun(\cF_N)$.

\end{remark}

Because $\delta S_N=0$ on the subspace $\EL_N\subset \cF_N$ we have the following Proposition.

\begin{Proposition} The action is constant on connected components of $\EL_N$.
\end{Proposition}

\section{ BV theories for spacetime manifolds with boundary}\label{BV-bound}

\subsection{Non-reduced theory}\label{bvbfv}
\subsubsection{Non-reduced BV-BFV theory}
In this section, unless otherwise specified, $N$ is an
$n$-dimensional spacetime manifold, with $\pa N$ being its $n-1$ dimensional boundary. The BV theory for spacetime manifolds with boundary consists
of two parts.

First, in the bulk we have the same data as for BV theory,
that is: the space of fields $\F_N$ which is a graded (the degree is called ghost number) symplectic manifold with the symplectic form $\omega_N$ with $gh=-1$, the cohomological vector field $Q_N$ with $gh=1$ and the
action functional $S_N$ of ghost number zero.

Second, on the boundary we have the BFV data \cite{BFV} which consist of: the space of boundary fields $\F_{\pa N}$ which is a graded (usually infinite dimensional) symplectic manifold with the symplectic
form $\omega_{\pa N}$ which is exact $\omega_{\pa N}=(-1)^{n-1}\delta\alpha_{\pa N}$ with $gh=0$, and a vector field $Q_{\pa N}$ with $gh=1$. The sign in $(-1)^{n-1}$ is here because
it appears naturally in AKSZ examples of topological gauge theories. The BFV data satisfy the following conditions:

\begin{equation}\label{BFV}
Q_{\pa N}^2=0, \ \ L_{Q_{\pa N}}\omega_{\pa N}=0
\end{equation}

These conditions imply that $Q_{\pa N}$ is Hamiltonian vector field \cite{Royt} with ghost number $1$,
with the Hamiltonian $S_{\pa N}$, which is by definition
the BFV action. In Appendix \ref{graded} we recall this construction.

A BV theory on a spacetime with boundary can be
regarded as a pairing of the BV data with the BFV data
through the natural mapping $\pi: \F_N\to \F_{\pa N}$
which is the restriction of fields to the boundary.

We will say that BV theory in the bulk and BFV theory on
the boundary agree and form the BV-BFV theory on the
manifold with boundary if the restriction mapping is a surjective submersion and if the BV data in the bulk and BFV data on the boundary satisfy the following conditions

\begin{equation}
Q_N^2=0, \ \ \delta\pi (Q_N)=Q_{\pa N}, \ \ L_{Q_{\pa N}}\omega_{\pa N}=0
\end{equation}
\begin{equation}\label{cme}
\iota_{Q_N}\omega_N=(-1)^n\delta S_N+\pi^*(\alpha_{\pa N}),
\end{equation}
The sign is purely conventional, but it fits well with the
natural definition of BV data for AKSZ theories.

Note that first two equations imply that
\begin{equation}\label{q2z}
Q_{\pa N}^2=0
\end{equation}
which means BFV axioms are satisfied.
The equation (\ref{cme}) implies
\begin{equation}\label{lo}
L_{Q_N}\omega_N=(-1)^n\pi^*(\omega_{\pa N})
\end{equation}
which means, in particular, that in BV-BFV theory the symplectic form is no longer $Q$-invariant.

The bulk action satisfies an important identity which describes
how the action changes with respect to gauge transformations.

\begin{Proposition} \label{LQS} The following identity holds in any BV-BFV theory
\be L_{Q_N} S_N = (-1)^{\dim(N)}\pi^* (2S_{\dd N} - \iota_{Q_{\dd N}} \alpha_{\dd N})  \label{CME}\ee
\end{Proposition}

\begin{remark} We can remove the signs $(-1)^n$ by redefining
vector fields $Q_N$, $Q_{\pa N}$ and the $1$-form $\alpha_{\pa N}$. Set $Q=(-1)^nQ_N$, $Q_\pa =(-1)^nQ_{\pa N}$ and $\alpha_\pa =(-1)^{n-1}\alpha_{\pa N}$. Then $\delta\pi(Q)=Q_\pa$, $\iota_Q\omega_N=S_N-\pi^*(\alpha_\pa)$ and $L_Q\omega_N=\pi^*(\omega_{\pa N})$.
\end{remark}

\begin{proof}By definition of the Lie derivative
\[
L_{Q_N}=\iota_{Q_N}\delta-\delta\iota_{Q_N}
\]
and the same formula for the Lie derivative holds for
$Q_{\pa N}$. The sign is minus in this formula because $\iota_Q$ is an even operation. We will prove the Proposition by applying
$L_{Q_N}$ to the  equation:
\[
\iota_{Q_N} \omega_N = (-1)^n \delta S_N + \pi^*(\alpha_{\pa N})
\]

By definition of $L_{Q_N}$ and because of $Q_N^2=0$ we have $[L_{Q_N}, \iota_{Q_N}]=L_{Q_N}\iota_{Q_N}-\iota_{Q_N}L_{Q_N}=0$ and
we already know that
\[
L_{Q_N} \omega_N = -\pi^*\delta\alpha_{\pa N}=(-1)^{n+1}\pi^* \omega_{\pa N},
\]

Therefore we have identities:
\begin{equation}\label{p1}
L_{Q_N}\iota_{Q_N}\omega_N=\iota_{Q_N}L_{Q_N}\omega_N=-\pi^*\iota_{Q_{\pa N}}\delta\alpha_{\pa N},
\end{equation}

Now we can apply the Lie derivative of $Q_N$ to the
classical master equation:
\begin{multline}
L_{Q_N}\iota_{Q_N}\omega_N=(-1)^nL_{Q_N}\delta S_N+L_{Q_N}\pi^*\alpha_{\pa N}=
-(-1)^n\delta\iota_{Q_N}\delta S_N+\pi^*(L_{Q_{\pa N}}\alpha_{pa N})=\\-(-1)^n\delta L_{Q_N}S+
\pi^*(\iota_{Q_{\pa N}}\delta\alpha_{\pa N}-\delta\iota_{Q_{\pa N}}\alpha_{\pa N})
\end{multline}
Here, in the last line, we used the formula for the Lie derivative $L_{Q_{\pa N}}$. Using (\ref{p1})
and the identity $\delta S_{\pa N}=\iota_{Q_\pa N}\omega_{\pa N}$ we arrive at:
\[
-\pi^*(\delta S_{\pa N})=-(-1)^n\delta L_{Q_N}S_N+\pi^*(\delta S_{\pa N})-\pi^*(\delta\iota_{Q_{\pa N}}\alpha_{\pa N})
\]

This is equivalent to
\[
\delta((-1)^nL_{Q_N}S_N-\pi^*(2S_{\pa N}-\iota_{Q_{\pa N}}\alpha_{\pa N}))=0
\]
We have $\delta (\mbox{ function of degree } 1) = 0$, which
means the function must vanish. This proves the Proposition.

\end{proof}

\begin{remark}\label{line-bun} In the definition of BV-BFV theory we made
an assumption that $\omega_\pa$ is exact. This assumption can be removed. Two examples, a charged particle in an electro-magnetic
field, and the WZW model suggest that a more natural condition
is to require that there is a line bundle ${\mathfrak L}_\pa$
over $\cF_{\pa N}$. The symplectic form is $\omega_{\pa N}= (-1)^{n-1}\delta\alpha_{\pa N}$
where $\alpha_{\pa N}$ is a connection on ${\mathfrak L}_\pa$.
The condition
\[
L_{Q_N} \omega_N = (-1)^n\pi^* \omega_{\pa N}
\]
implies that the connection $(-1)^n\frac{i}{\hbar}(-\iota_{Q_N} \omega_N +\pi^* \alpha_{\pa N})$  on the line bundle $\pi^*({\mathfrak L}_\pa)$ over $\cF_N$ is flat. In this case the action functional $S_N$
should be regarded as a horizontal section $s_N=\exp(\frac{i S_N}{\hbar})$ of $\pi^*({\mathfrak L}_\pa)$ and the equation (\ref{cme}) becomes
\[
(\delta+ (-1)^n\frac{i}{\hbar}(-\iota_{Q_N} \omega_N + \pi^* \alpha_{\pa N})) s_N=0
\]

\end{remark}

\subsubsection{Digression: gauge invariance}
Now let us discuss the gauge invariance of BV-BFV theory.
When $\pa N\neq \emptyset$, Poisson bracket of local functionals
is not defined. In this case we will call the vector field
$V$ Hamiltonian if there exists a local functional $F$
such that
\[
\iota_V\omega_N=\delta F
\]
Recall that a vector field $V$ is called {\it projectable} if for all $X\in \cF_N$
\[
\delta\pi_X (V_X)=v_{\pi(X)}
\]
for some vector field $v$ on $\cF_{\pa N}$.
The gauge invariance of the BV-BFV action
can be formulated as the following statement:
\begin{Proposition}
If $V$ is a projectable Hamiltonian vector field, then
the vector field $[Q_N, V]$ preserves $S_N$ up to a pullback from the boundary:
\[
L_{[Q_N,V]}S_N=(-1)^{\dim N} \pi^* (\iota_{v}L_{Q_{\pa N}}\alpha_{\pa N}+\iota_{Q_{\pa N}}L_{v}\alpha_{\pa N})
\]
\end{Proposition}
Indeed, we have:
\begin{multline}
L_{[Q_N,V]}S_N=L_{Q_N} L_V S +L_V L_{Q_N} S_N = \\ \iota_{Q_N} \delta \iota_V \delta S_N + (-1)^{\dim N}L_V \pi^* (2 S_{\pa N} - \iota_{Q_{\pa N}}\alpha_{\pa N})=
\\
= (-1)^{\dim N} (\iota_{Q_{\pa N}} \delta \iota_{Q_N} \iota_V \omega_N + \pi^*(-\iota_{Q_{\pa N}}\delta\iota_{v}\alpha_{\pa N}+\\2 \iota_{v}\delta S_{\pa N} - \iota_{v}\delta \iota_{Q_{\pa N}}\alpha_{\pa N}))=(-1)^{\dim N} \pi^* (\iota_{v}L_{Q_{\pa N}}\alpha_{\pa N}+\iota_{Q_{\pa N}}L_{v}\alpha_{\pa N})
\end{multline}
Here we used identities $\delta S_N=(-1)^{\dim N}(\iota_{Q_N} \omega_N-\pi^*\alpha_{\pa N})$, \ \  $\iota_V \omega_N=\delta F$,
 $L_{Q_N}L_{Q_N} F= \frac{1}{2}L_{[Q_N,Q_N]}F=0$, $\delta S_{\pa N}=(-1)^{\dim N-1}\iota_{Q_{\pa N}}\omega_{\pa N}=\iota_{Q_{\pa N}}\omega_{\pa N}=\iota_{Q_{\pa N}}\delta\alpha_{\pa N}$, and
 $F$ is the generating function for $V$ and $v$ is the projection of $V$ to the boundary.

As a corollary we also have the formula
\[
\delta L_{[Q_N,V]}S = - (-1)^{\dim N} \pi^* (L_{[Q_{\pa N},v]}\alpha_{\pa N})
\]

\subsubsection{The boundary structure from the bulk} Let us show that a first order BV theory of spacetime manifolds (on the bulk)
induces BFV theory on the boundary. Denote the pullback
of fields to the boundary by $\F_N|_{\pa N}$ (we say pullback because fields are usually either sections of a fiber bundle, or connections). The differential of
the action can be written as the bulk part plus the boundary
contribution from the use of Stokes theorem:
\[
\delta S_N[X]=\int_N A\wedge \delta X-(-1)^n\pi^*\tilde{\alpha}_{\pa N}[X]
\]
where $\tilde{\alpha}_{\pa N}$ is a one-form on the space $\F_N|_{\pa N}$
and $A$ defines the Euler-Lagrange equation.
The kernel of
$\delta \tilde{\alpha}_{\pa N}$ forms a distribution on the
space $\F_N|_{\pa N}$. Denote by $\F_{\pa N}$ the space of leaves
of this distribution. We have a natural projection $\pi: \F_N\to \F_{\pa N}$.

Denote by $\alpha_{\pa N}$ the form on $\F_{\pa N}$ corresponding to the
form $\tilde{\alpha}_{\pa N}$ on $\F_N|_{\pa N}$\footnote{Strictly speaking the form $\tilde{\alpha}_{\pa N}$ becomes a connection on a line bundle over $\cF_{\pa N}$ with the curvature $\omega_{\pa N}$. However in our examples this line bundle is trivial. This is why for the time being we will assume this and will regard $\alpha_{\pa N}$ as a $1$-form.} .

Taking into account that for closed manifolds the vector field $Q_N$ is Hamiltonian, generated by the classical action functional, we conclude that the bulk term in the
differential of the action functional is $
(-1)^n\iota_{Q_N}\omega_N$. Thus, we have equation (\ref{cme}).

Observe that $Q_N$ is projectable to $\cF_{\pa N}$ denote its
projection by $Q_{\pa N}=\delta\pi (Q_N)$.
For any form $\theta$ on boundary fields we will have $L_{Q_N}\pi^*(\theta)=\pi^*(L_{Q_{\pa N}}\theta)$. Equations
$Q_N^2=0$ and (\ref{lo}) imply that $L_{Q_N}^2\omega_N=\pi^*(L_{Q_{\pa N}}\omega_{\pa N})=0$.
Therefore $L_{Q_{\pa N}}\omega_{\pa N}=0$.

It is clear that the BFV action induced by a first order
BV theory is also of first order.

\subsubsection{Some properties of non-reduced BV-BFV theories} Associated with BV and BFV data we have the following
important subspaces in the space of bulk fields and in the spaces
of boundary fields:

\begin{itemize}

\item The space $\EL_N\subset \F_N$ of zeroes of the vector field $Q_N$.

\item The space $\EL_{\pa N}\subset \F_{\pa N}$ of zeroes of the vector field $Q_{\pa N}$.

\item The space $\cL_N=\pi(\EL_N)\subset \F_{\pa N}$ of boundary values of solutions of Euler-Lagrange equations. It is
    clear that $\cL_N\subset \EL_{\pa N}\subset \EL_N$.

\end{itemize}

\begin{Proposition}\label{paN-coisot} a) The subspace $\EL_{\pa N}$ is locally coisotropic in $\F_{\pa N}$ (its tangent space at a smooth point is coisotropic in the tangent space to the space of fields at this point).

b) $\EL_N$ is still locally coisotropic when $\pa N\neq \emptyset$.
\end{Proposition}

The proof of a) is identical to the proof of Proposition \ref{el-coiso} for closed manifolds. The same proof
carries through for b) by requiring $\eta$ to vanish on the boundary. The algebraic proof outlined after Proposition \ref{el-coiso} also works,  the only difference is
that one should take local functionals with test functions vanishing on the boundary. An example of such functional is $\int_N \alpha_a\wedge X^a$ where $X^a$ are coordinate fields and $\alpha_a$ are test
functions which vanish on $\pa N$.

\begin{Proposition} \label{sub-loc-isot} The subspace $\cL_N\subset \EL_{\pa N}\subset \F_{\pa N}$ is locally isotropic.
\end{Proposition}
\begin{proof} Assume $X$ is a smooth point of $\EL_N$ The equation (\ref{cme}) implies
that
\be \label{Q kinda self-adjoint}
\omega_X(\hat{Q}_X\xi, \eta)-\omega_X(\xi, \hat{Q}_X\eta)=
\tilde{\omega}_{\pi(X)}(\delta\pi(\xi),\delta\pi(\eta))
\ee
were $\xi,\eta$ are tangent vectors to $\F_N$ at the point $X$.
If $\xi$ and $\eta$ are tangent to $\EL_N\subset \F_N$ then
$\hat{Q}_X\xi=\hat{Q}_X\eta=0$. Therefore $\tilde{\omega}_{\pi(X)}(\delta\pi(\xi),\delta\pi(\eta))=0$ which means that
$\delta\pi(T_X\EL_N)=T_{\pi(X)}\cL_N$ is an isotropic subspace in $T_{\pi(X)}\F_{\pa N}$.
\end{proof}

We proved that $\cL_N$ is always isotropic. We will prove in Section \ref{symp-red-sub} (Corollary \ref{cor: L Lagr}) that under a natural regularity assumption (see Definition \ref{def: regular}), $\CL_N$ is in fact Lagrangian.

Another natural property of a BV-BFV classical field theory is
locality. If two spacetime manifolds $N_1$ and $N_2$ have
common boundary $\Sigma$ (say $\Sigma\subset \pa N_1$ and
it is identified by an orientation preserving diffeomorphism
with part of the boundary of $N_2$) then
\[
\F_{N_1\cup_\Sigma N_2}=\F_{N_1}\times_{\F_\Sigma}\F_{N_2}
\]
\[
S_{N_1\cup_\Sigma N_2}=S_{N_1}+S_{N_2}
\]

\subsection{The reduction of the boundary BFV theory}
\label{sec: BFV red}

The boundary manifold $\pa N$ is closed. This is why both
the $Q$-reduction  and the symplectic reduction for the
boundary BFV theory work as they do for the BV theory on closed spacetime manifolds. The additional structure in the BV-BFV theory is
the reduction of the 
submanifold $\cL_N\subset \EL_{\pa N}$ and the reduction of the fibers of the fiber bundle $\pi: \EL_N\to \cL_N$.

\subsubsection{The Q-reduction}

The following Proposition is an immediate corollary of the identity $\delta\pi\circ Q_N=Q_{\pa N}\circ \delta\pi$.

\begin{Proposition} \label{prop: Vect_Q tangent to L} The distribution $Vect_{Q_{\pa}}$ on $\F_{\pa N}$ generated by Lie brackets of vector fields with $Q_{\pa N}$ is parallel to $\cL_N$.
\end{Proposition}

\begin{definition} Define $Q$-reductions of $\cL_N$ and of $\EL_{\pa N}$ as the space of leaves of $Vect_{Q_{\pa}}$ through $\cL_N$ and $\EL_{\pa N}$ respectively.
\end{definition}

In general these spaces are singular. Because our goal here is to
develop the set up for the perturbative quantization in a vicinity of a generic classical solution, we will focus on smooth points of $\EL_{\pa N}$ and $\cL_N$ and tangent spaces at such points.
See Appendix \ref{smooth} below for the discussion of smooth points.

\subsubsection{The symplectic reduction}

Let $\underline{\EL}_{\pa N}$ be the symplectic reduction
of $\EL_{\pa N}$.
We make the following assumption, which is the analogue of (\ref{assump}) for the boundary fields:

\begin{assumption}
We assume that for any $l\in \EL_{\pa N}$,
\be \label{assump 1 for bdry}
\ker (\hat{Q}_l)^\perp = Im (\hat{Q}_l)
\ee
in $T_l \FF_{\pa N}$, where $\hat{Q}_l$ is the linearization of $Q_{\pa N}$ at $l$.
\end{assumption}

The following Proposition is the counterpart of the
Remark \ref{rem: symp red = Q-red BV w/o bdry}.

\begin{Proposition} Under the assumption (\ref{assump 1 for bdry}), locally, the symplectic reduction
of $\EL_{\pa N}$ is equal to its $Q$-reduction.
\end{Proposition}
The difference from Remark \ref{rem: symp red = Q-red BV w/o bdry} is just in the shift of gradings.

Denote by $\underline{\cL}_N\subset \underline{\EL}_{\pa N}$
the image of $\cL_N$ with respect to this reduction.
That is $\underline{\cL}_N$ is the space of leaves of the
characteristic foliation of $\EL_{\pa N}$ through $\cL_N$.
This Proposition implies that locally the  reduction $\underline{\cL}_N$
of $\cL_N$ is equal to its $Q$-reduction $\cL_N/Q$ .

Under the regularity assumption (see Definition \ref{def: regular}, Proposition \ref{prop: L Lagrangian}),
$\underline{\cL}_N\subset \underline{\EL}_{\pa N}$ is a Lagrangian submanifold.


\subsection{The reduction of the bulk BV theory}
\label{sec: Q-red fibers}

The meaning of the $Q$-reduction is passing from fields to gauge classes of fields. This is why it is natural to reduce not the whole space of solutions to Euler-Lagrange equations but a subspace with fixed boundary
conditions. Points of this reduced space correspond to
gauge classes of fields with boundary values in a given
boundary gauge class.

The space $\EL_N$ is naturally fibered over $\cL_N=\pi(\EL_N)$
where $\pi: \F_N\to \F_{\pa N}$ is the restriction mapping.

For $l\in \cL_N$ denote by $[l]$ the leaf of $Vect_{Q_{\pa N}}$ through $l$, i.e. the image of $l$ in the $Q$-reduced space $\cL_N/Q$. Denote by $\EL(N,[l])$ the fiber of $\pi$ over $[l]$,
\[
\EL(N,[l])=\pi^{-1}([l])\cap \EL_N
\]

Because the restriction mapping $\pi$ intertwines vector fields
$Q_N$ and $Q_{\pa N}$, the vector field $Q_N$ is parallel to
$\EL(N, [l])$. It also induces the projection of reduced spaces
$p: \EL_N/Q\to \cL_N/Q$ where $\EL_N/Q$ is the space of leaves of
$Vect_Q$ in $\EL_N$.

\begin{definition} The reduced space $\EL(N,[l])/Q$ is the space
of leaves of $Vect_Q$ on $\EL(N,[l])$.
\end{definition}

Because the projection $\pi$ intertwines vector fields $Q_N$ and
$Q_{\pa N}$ the reduced space $\EL(N, [l])/Q$ is the fiber of
$p$ over $[l]$: $\EL(N, [l])/Q=p^{-1}([l])$.

Assume that the image of $X\in \EL_N$ in $\EL(N, [\pi(X)])/Q$ is a smooth point. We have the natural mappings of tangent spaces:
\begin{equation}
\label{global zero-modes}
\begin{CD}
T_{[X]}\EL_N/Q @<\iota << T_{[X]}\EL(N,[l])/Q \\
@V \delta p VV \\
T_{[l]}\cL_N/Q \subset T_{[l]}\EL_{\pa N}/Q
\end{CD}
\end{equation}
where $l=\pi(X)$,  and the horizontal mapping is the natural inclusion of fibers, $Im(\iota)=\ker(\delta p)$. We know that $T_{[X]}\EL_N/Q\simeq H(T_X\F_N;\hat{Q}_X)$, and that
$T_{[l]}\EL(\pa N)/Q\simeq H(T_l\F_{\pa N};\hat{Q}_l)$\footnote{Recall that $\hat{Q}_X$ is the linearization of
$Q_N$ at $X\in \EL_N$ and $\hat{Q}_l$ is the linearization of $Q_{\pa N}$ at $l=\pi(X)\in \EL_{\pa N}$.}.

We also know that the short exact sequence
$$
\begin{CD}
T_X \FF_N @< i << \ker (\delta\pi)_X \\
@V(\delta\pi)_X VV \\
T_{l} \FF_{\dd N}
\end{CD}
$$
gives the long exact sequence
$$\cdots\rightarrow H^\bullet_{\hat{Q}_X}(\ker (\delta\pi)_X)\stackrel{[i]}{\rightarrow} H_{\hat{Q}_X}^\bullet(T_X \FF_N) \stackrel{[\delta\pi_X]}{\rightarrow} H^\bullet_{\hat{Q}_{l}} (T_{l} \FF_{\pa N})\stackrel{\phi}\rightarrow H_{\hat{Q}_X}^{\bullet+1}(\ker (\delta\pi)_X)\rightarrow \cdots $$
on the $\hat{Q}$-cohomology spaces.

Truncating this long exact sequence we obtain the short exact sequence

\begin{equation}
\begin{CD}
H_{\hat{Q}_X}(T_X\F_N) @<[i]<<   H_{\hat{Q}_X}(\ker (\delta\pi)_X)/\phi(H_{\hat{Q}_{l}}(T_{l}\FF_{\dd N})) \\
@V[\delta\pi_X]VV \\
[\delta\pi_X](H_{\hat{Q}_X}(T_X\F_N))\subset H_{\hat{Q}_{l}}(T_{l}\FF_{\dd N})
\end{CD}
\end{equation}
Because we have natural isomorphisms $H_{\hat{Q}_X}(T_X\F_N)\cong
T_{[X]}\EL_N/Q$, $H_{\hat{Q}_l}(T_l\F_{\pa N})\cong T_{[l]}\EL_{\pa N}/Q$ and $[\delta\pi_X](H_{\hat{Q}_X}(T_X\F_N))=T_{[l]}\cL_N/Q\subset T_{[l]}\EL_{\pa N}/Q$ we can identify $T_{[X]}\EL(N, [l])/Q$ with the fiber of $[\delta\pi_X]$. We proved the following:

\begin{Proposition} If $[\pi(X)]\in \EL_N/Q$ is a smooth point,
then
\[
T_{[X]}\EL(N, [l])/Q\cong
H_{\hat{Q}_X}(\ker(\delta\pi_X))/\phi(H_{\hat{Q}_{l}}(T_{l}\FF_{\dd N}))
\]
\end{Proposition}

Here is another, equivalent characterization of this space.
By definition
\[
T_{X}\EL(N, [l])=\{\xi\in \ker(\hat{Q}_X) \;|\; \delta\pi_X(\xi)\in Im(\hat{Q}_{l})\}
\]
On the other hand $(Vect_Q)_X=Im(\hat{Q}_X)$. Therefore
the tangent space to the reduced space is the quotient
\[
T_{X}\EL(N, [l])/Q=\{\xi\in \ker(\hat{Q}_X) \;|\; \delta\pi_X(\xi)\in Im(\hat{Q}_{l})\}/Im(\hat{Q}_{X})
\]
%
This implies the following.
\begin{Proposition} \label{P1}We have
$$
T_{X}\EL(N, [l])/Q\cong \{\xi\in \ker(\hat{Q}_X) \;|\; \delta\pi_X(\xi)=0\}/
\{\xi\in Im(\hat{Q}_X) \;|\; \delta\pi_X(\xi)=0\}
$$
\end{Proposition}

\subsection{Symplectic $\EL$-moduli spaces} \label{part-red-sec}
In previous sections we introduced $\EL$-moduli space as the $Q$-reduction
of the space of solutions to Euler-Lagrange equations. In case when $\pa N=\emptyset$
we proved that they carry a symplectic structure coming from symplectic reduction, but this is
no longer true in the case with boundary.

In this subsection, for the case with boundary we will introduce a different reduction which is
symplectic (see subsection \ref{symp-red-sub}), fibers over the $Q$-reduction and has simple gluing properties.
We call it the symplectic $\EL$-moduli space.

Denote the bulk and boundary $\EL$-moduli spaces by
\[
\cM_N=\EL_N/Q_N, \ \ \cM_{\pa N}=\EL_{\pa N}/Q_{\pa N}
\]
When $N$ is fixed, we will usually suppress the subscript $N$ for brevity.

\subsubsection{The main construction} Denote by $\widetilde{\Vect}_Q^\rel$ the space of vector fields of the form $[Q_N, v]$ where $v$ is a \textit{vertical} vector field, i.e. a vector field tangent to the fibers of $\pi$.
This space is also a Lie subalgebra of $\widetilde{\Vect}_Q$.
The restriction of $\widetilde{\Vect}_Q^\rel$ to 
$\EL$ gives a distribution which we denote by  $\Vect_Q^\rel$ .

Define the \textit{symplectic} $\EL$-moduli space $\MM^\p$
as the space of leaves of $\Vect_Q^\rel$ on $\EL$:
\[
\cM^\p={\EL}/\Vect_Q^\rel
\]
It is clear that
\[
T_{[X]}\cM^\p=\ker(\hat{Q}_X)/\hat{Q}_X(\ker(\delta\pi_X))
\]

The restriction to the boundary $\pi: \cF\to \cF_{\pa}$
induces the projection $\pi_*: \cM^\p\to \EL_{\pa}$.

The distribution $\Vect_Q$ on $\EL$ induces the distribution $b$ on $\cM^\p$ which projects to
the distribution $\Vect_{Q^\pa}$ on $\EL_{\pa}$ by $\delta\pi_*$.
It is easy to see that $b$ is involutive
and that
\[
b_{[X]}\simeq Im(\hat{Q}_X)/\hat{Q}_X(\ker(\delta\pi_X))
\]
where $[X]$ is the leaf of $\Vect_Q^\rel$ through $X$.
The image of this projection can be thought of as a quotient distribution $\Vect_Q/\Vect_Q^\rel$.

Notice that there is a natural isomorphism $T_X\cF/\ker{\delta\pi_X}\simeq T_{l}\cF_{\pa}$ where $l=\pi(X)$.
We also have a natural degree 1 surjective map
\[
\beta: T\cF_{\pa}\twoheadrightarrow b,  \ \ \beta_X: T_{l}\cF_{\pa}\twoheadrightarrow b_{[X]}
\]
defined as follows: for $\xi\in T_X\cF$ and  $\xi+\ker(\delta \pi)$ being identified with an element in $T_{l}\cF_{\pa}$, we set
\[
\beta_X(\xi+\ker(\delta \pi))=\hat{Q}_X(\xi)+\hat{Q}_X(\ker(\delta \pi))
\]

The following statement follows immediately from $\hat{Q}_X^2=0$.
\begin{lemma} $\beta$ vanishes on $Im(\hat{Q}_l)$.
\end{lemma}

The usual $\EL$-moduli space, i.e. space of leaves of
$\Vect_Q$ on $\EL$, is naturally isomorphic to
the space of leaves of $b$ on $\cM^\p$. This
follows immediately from the definition of $\cM^\p$.

Finally, from the definitions above, it is clear that the following diagram is commutative:

\be
\begin{CD}
b_{[X]} @= b_{[X]} @>\subset>> T_{[X]} \MM^\p \\
@A\beta_{[X]} AA @V(\delta\pi_*)_{[X]}VV @V(\delta\pi_*)_{[X]}VV \\
T_{l}[-1]\FF_\dd @>\hat{Q}_{l}>> (\Vect_{Q_\dd})_{l} @>\subset>> T_{l} \EL_\dd
\end{CD}
\label{beta diagram}\ee

\subsubsection{More on tangent spaces}
Define the vertical component $T^\ver_{[X]}\cM^\p$ of the tangent space $T_{[X]}\cM^\p$ as
\[
T^\ver_{[X]}\cM^\p=\ker((\delta\pi_*)_X)\subset T_{[X]}\cM^\p
\]
and denote the quotient map by $\chi: T_{[X]}^\ver\cM^\p \to T_{[X]^\abs}\cM $
where $\cM$ is the usual $\EL$-moduli space for $N$, i.e. the space of leaves of $\Vect_Q$ on $\EL$, and $[X]^\abs$ is the leaf of $\Vect_Q$ through $X$.

The projection $\pi: \cF\to \cF_{\pa}$ restricted to $\EL$
induces the natural projection
$\pi^{\abs}_*: \cM\to \cM_{\pa}$. Recall that the fiber of $\pi^\abs_*$ over $[l]\in \cM_\pa$ is the space $\EL(N,[l])/Q$ discussed in Section \ref{sec: Q-red fibers}; the image of $\pi^\abs_*$ is the $Q$-reduction of the space $\cL\subset \EL$ discussed in Section \ref{sec: BFV red}.
Denote by $\psi=\delta\pi^{\abs}_*$ the corresponding mapping of
tangent bundles, $\psi_{[X]^\abs}: T_{[X]^\abs}\cM\to T_{[l]}\cM_{\pa}$.

The restriction of the mapping $\beta: T\cF_{\pa}\to b$
defined earlier to $\ker(\hat{Q}_{l})$
vanishes on $Im(\hat{Q}_{l})$
and therefore induces a linear mapping $(\beta_*)_{[X]}:
T_{l}[-1]\cM_{\pa}\to T^\mr{vert}_{[X]}\cM^\p$.

We have a sequence of linear maps:
\be \cdots \ra T_{[l]}[-1]\MM_{\dd} \xra{\beta_*} T^\ver_{[X]}\MM^\p \xra{\chi} T_{[X]^\abs}\MM \xra{\psi} T_{[l]}\MM_{\dd} \ra\cdots \label{LES of tangent spaces to moduli spaces}\ee

\begin{Proposition} This is an exact sequence.
\end{Proposition}

\begin{proof}
Sequence (\ref{LES of tangent spaces to moduli spaces}) can be written as
$$\cdots \ra H^{\bt-1}_{\hat{Q}_{l}}(T_{l}\FF_{\dd}) \ra H^\bt_{\hat{Q}_X} (T^\ver_X\FF)\ra H^\bt_{\hat{Q}_X}(T_X\FF) \ra H^\bt_{\hat{Q}_{l}}(T_l\FF_{\dd})\ra \cdots$$
which is induced from the short exact sequence in the non-reduced picture
$$0\ra T^\ver_X \FF \ra T_X \FF \ra T_l\FF_{\dd}\ra 0$$
by passing to $\hat{Q}$-cohomology and so is exact by snake lemma.
\end{proof}

\subsubsection{Symplectic structure, Lefschetz duality}\label{symp-red-sub}

Denote by $\hat{Q}_X^\ver$ the restriction of $\hat{Q}_X$ to $T_X^\ver \FF$.

We know (cf. the proof of Proposition \ref{paN-coisot}) that in $T_X\FF$ we have
$$Im(\hat{Q}_X^\ver)^\perp= \ker (\hat{Q}_X),\quad Im(\hat{Q}_X)^\perp= \ker (\hat{Q}_X^\ver)$$
\begin{assumption}
We will assume that also the opposite holds:
\be \ker (\hat{Q}_X)^\perp = Im(\hat{Q}_X^\ver),\quad \ker (\hat{Q}_X^\ver)^\perp = Im(\hat{Q}_X) \label{assump bdry} \ee
\end{assumption}
This is a stronger version of assumption (\ref{assump}); in our examples it follows from Hodge-Morrey decomposition theorem for manifolds with boundary \cite{CGMT}.
It immediately implies that $\Vect_Q^\rel$ coincides with the characteristic distribution on $\EL$, thus we have the following.

\begin{Proposition}
Assuming (\ref{assump bdry}), $\MM^\p$ is the symplectic reduction of $\EL$ and carries a degree $-1$ symplectic structure $\underline{\omega}$ coming from $\FF$.
\end{Proposition}

\begin{definition}\label{def: regular} We call a BV-BFV theory \textbf{regular} if the assumption (\ref{assump bdry}) holds for any $X\in \EL$, together with the assumption (\ref{assump 1 for bdry}) for any $l\in \EL_\pa$.
\end{definition}


Symplectic structure $\omega$ on $\FF$ induces a bilinear pairing
$H_{\hat{Q}^\ver_X}\otimes H_{\hat{Q}_X}[1] \ra \RR$ or equivalently
$T^\ver_{[X]}\cM^\p \otimes T_{[X]^\abs}[1]\cM  \ra \RR$
which is well-defined due to (\ref{Q kinda self-adjoint}) and non-degenerate due to (\ref{assump bdry}). Together with the symplectic structure on $\MM_{\pa}=\underline{\EL}_{\pa}$ it gives the Lefschetz duality for the long exact sequence (\ref{LES of tangent spaces to moduli spaces}), i.e.
a non-degenerate pairing on $T^\ver_{[X]}\cM^\p\oplus T_{[X]^\abs}[1]\cM \oplus T_{[l]}\cM_{\pa}$:

\be \langle \bt,\bt\rangle:\quad
\left\{\begin{array}{lll} T^\ver_{[X]}\cM^\p \otimes T_{[X]^\abs}[1]\cM & \ra &\RR \\
 T_{[X]^\abs}[1]\cM \otimes T^\ver_{[X]}\cM^\p & \ra & \RR \\
 T_{[l]}\cM_{\pa} \otimes T_{[l]}\cM_{\pa} & \ra & \RR
\end{array} \right. \label{Lefshetz pairing}\ee
Observe that (\ref{Q kinda self-adjoint}) implies that the map $\chi$ is self-adjoint with respect to $\langle\bt,\bt\rangle$, and $\psi$ and $\beta_*$ are adjoint to each other, therefore Lefschetz duality can be stated as an injective chain map between the complex (\ref{LES of tangent spaces to moduli spaces}) and the dual one:
{\tiny
$$ 
\begin{CD}
\cdots @>>> T_{[l]}[-1]\MM_{\dd} @>\beta_*>> T_{[X]}^\ver\MM^\p @>\chi>> T_{[X]^\abs} \MM @>\psi>> T_{[l]}\MM_{\dd} @>>> \cdots \\
@. @VVV @VVV @VVV @VVV @. \\
\cdots @>>> T_{[l]}^*[-1]\MM_{\dd}@>\psi^*>> T_{[X]^\abs}^*[-1] \MM @>\chi^*>>  (T_{[X]}^\ver)^*[-1]\MM^\p @>(\beta_*)^*>> T_{[l]}^* \MM_{\dd}@>>> \cdots
\end{CD}
$$
}
In case of a topological field theory, the complex (\ref{LES of tangent spaces to moduli spaces}) consists of finite dimensional vector spaces and vertical arrows in the diagram above become isomorphisms.

The form $\underline{\omega}$ induces presymplectic structures $\omega^{\p,\ver}$ on the fibers of $\pi_*: \cM^{\p} \ra \EL_{\pa}$. The form $\omega^{\p,\ver}$ can be written in terms of Lefschetz duality as:
$$ \omega^{\p,\ver}_{[X]}= \langle \bt , \chi (\bt) \rangle : \qquad  T^\ver_{[X]}\MM^\p\otimes  T^\ver_{[X]}[1]\MM^\p \ra \RR $$

\begin{Proposition}\label{prop: red fibers symp} Fibers of $\pi_*^\abs: \cM\ra \cM_{\pa}$ carry a natural degree $-1$ symplectic structure coming from $\FF$.
\end{Proposition}

\begin{proof}
By non-degeneracy of (\ref{Lefshetz pairing}), the kernel of $\omega^{\p,\ver}$ is exactly the kernel of $\chi$. Thus $\omega^{\p,\ver}$ induces a non-degenerate degree $-1$ symplectic structure on $Im(\chi)=\ker(\psi)\subset T_{[X]^\abs}\cM$ and hence on fibers of $\pi_*^\abs: \cM\ra \cM_{\pa}$.
\end{proof}

Thus, the $Q$-reduced fibers $\EL(N,[l])/Q$ have natural symplectic
structure. We call these reduced fibers \emph{the moduli spaces of vacua}.

\begin{Proposition}\label{prop: L Lagrangian}
The image of $\pi_*^\abs: \cM\ra \cM_{\pa}$ is locally Lagrangian in $\cM_{\pa}$.
\end{Proposition}
\begin{proof} Indeed, for a smooth point $[l]=[\pi(X)]\in \cM_{\pa}$,
tangent space to the image of $\pi_*^\abs$ is
$Im(\psi: H^\bt_{\hat{Q}_X}(T_X\FF)\ra H^\bt_{\hat{Q}_l}(T_l\FF_\pa))$.
Lagrangianity is proven as follows:
\begin{eqnarray*} Im(\psi)^\perp &=& \{a\in H_{\hat{Q}_l} (T_l\FF_\pa) \;|\; \langle a,\psi (b) \rangle=0\;\forall b\in H_{\hat{Q}_X}(T_X\FF) \} \\
&=& \{a\in H_{\hat{Q}_l} \;|\; \langle \beta_*(a), b \rangle=0\;\forall b\in H_{\hat{Q}_X} \}\\
&=& \ker(\beta_*)\\
&=& Im(\psi)
\end{eqnarray*}
where in the third line we used the non-degeneracy of Lefschetz pairing.
\end{proof}

The following is a corollary of the above, using also Propositions \ref{lagr} and \ref{prop: Vect_Q tangent to L}:
\begin{corollary} \label{cor: L Lagr}
The space $\cL=\pi(\EL)\subset \FF_\dd$ is locally Lagrangian.
\end{corollary}


\subsubsection{Digression: fibers of $\pi_*^\abs$ via symplectic reduction}

Let $\Lambda\subset \F_{\pa}$ be a Lagrangian submanifold
which is transversal to $\cL=\pi(\EL)$. Assume that $\Lambda$ intersects $\cL$ at the point $l$. The subspace $\pi^{-1}(\Lambda)\subset \F$ is the space of
fields with boundary values on $\Lambda$. We will assume that the subspace $\pi^{-1}(\Lambda)$ is symplectic.

Let us prove that 
the subspace $\pi^{-1}(\Lambda)\cap \EL$ is locally coisotropic in $\pi^{-1}(\Lambda)$. For this we need
to prove that at a smooth point $X$ the tangent space $\cS_X=T_X\pi^{-1}(\Lambda)\cap \EL$ is a coisotropic subspace in $T_X\F$.

Because Lagrangian subspaces $\cL$ and $\Lambda$ are transversal,
the intersection of their tangent spaces is trivial,
$T_X\Lambda\cap T_X\cL=\{0\}$, and therefore for each $\xi\in
T_X\pi^{-1}(\Lambda)\cap \EL$ we have $\delta\pi_X(\xi)=0$. Thus
\[
\cS_X=\{\xi\in \ker(\hat{Q}_X) \;|\; \delta\pi(\xi)=0\}
\]
Denote
$$I_X=\{\xi\in Im(\hat{Q}_X) \;|\; \delta\pi(\xi)=0\}$$
Observe that $I_X\subset \cS_X$ and that
\be \label{S/I} \cS_X/I_X \cong T_{X} \EL(N,[l])/Q
\ee
by Proposition \ref{P1}.


\begin{lemma}\label{lm1 symp red fib} For a smooth point $X$, the symplectic orthogonal subspace to  $I_X$ in $T_X\pi^{-1}(\Lambda)$ is $\cS_X$.
\end{lemma}

\begin{proof}
The symplectic orthogonal of $I_X$ in $T_X\pi^{-1}(\Lambda)$ is:
\[
I_X^\perp=\{\eta\in  T_X\pi^{-1}(\Lambda) \;|\; \omega_N(\eta,\xi)=0, \mbox{ for any } \xi\in I_X\}
\]
That is for any $\xi=\hat{Q}_X\lambda$ such that $\delta\pi(\lambda)\in \ker(\hat{Q}_{\pa})$. The orthogonality of $\xi$ and $\eta$ is equivalent to
\[
\omega(\hat{Q}_X\eta, \lambda)+\omega_{\pa}(\delta\pi(\eta), \delta\pi(\lambda))=0
\]
for any $\lambda$ such that $\delta\pi(\lambda)\in \ker(\hat{Q}_{l})$. Because the first term should
vanish for all $\lambda$, we have $\eta\in \ker(\hat{Q}_X)$.
Thus the last term should vanish separately for
any such $\lambda$. This implies that $\delta\pi(\eta)$ is symplectic orthogonal to $\ker(\hat{Q}_{l})$. 
By condition (\ref{assump 1 for bdry}) this implies that $\delta\pi(\eta)\in Im(\hat{Q}_{l})$. Because $Im(\hat{Q}_{l})\subset T_{l}\cL$ and because $\cL$ and $\Lambda$ are transversal, we have $\delta\pi(\eta)=0$.
Therefore $I_X^\perp=\cS_X$.

\end{proof}

\begin{corollary} The subspace $\cS_X\subset T_X\pi^{-1}(\Lambda)$
is coisotropic. Thus $\pi^{-1}(\Lambda)\cap \EL$ is locally coisotropic in  $\pi^{-1}(\Lambda)$.
\end{corollary}

\begin{lemma} \label{lm2 symp red fib}
The symplectic orthogonal subspace to $\cS_X$ in $T_X\pi^{-1}(\Lambda)$ is $I_X$.
\end{lemma}

\begin{proof} By Lemma \ref{lm1 symp red fib}, $I_X^\perp=\cS_X$, which implies $I_X\subseteq \cS_X^\perp$. On the other hand, due to (\ref{S/I}) and to Proposition \ref{prop: red fibers symp}, $\cS_X/I_X$ is symplectic and thus $I_X$ has to coincide with $\cS_X^\perp$.
\end{proof}

The following is the immediate corollary of the Lemma above.
\begin{Proposition}\label{C1} The symplectic reduction of $\cS_X$ is
\[
\underline{\cS}_X=\cS_X/I_X
\]
Comparing with the $Q$-reduction of $T_X\EL(N,[l])$ we
see that two reductions are naturally isomorphic:
\[
\underline{T_X\pi^{-1}(\Lambda)\cap \EL}=T_X\EL(N,[l])/Q
\]
Thus the fiber of $\pi_*^\abs:\cM\ra \cM_\pa$ over $[l]$ coincides with the symplectic reduction of $\pi^{-1}(\Lambda)\cap \EL$ in $\pi^{-1}(\Lambda)$.
\end{Proposition}


\subsubsection{Example: symplectic $\EL$-moduli space for the abelian Chern-Simons theory} In the abelian Chern-Simons theory, see section \ref{ab-cs} for details, the following occurs:
\begin{itemize}
\item The symplectic $\EL$-moduli space is
\begin{multline*}
\MM^\p = \EL/\Vect^\rel_{Q}=\\
=\{A\in \Omega^{\bt+1} (N)  \,|\, dA=0 \}/(A\sim A+d\alpha\quad \mbox{for any}\; \alpha\in \Omega^\bt(N)\; \mbox{s.t.}\; \alpha|_{\dd N}=0)
\end{multline*}
\item
The restriction map $\pi_*: \MM^\p \ra \EL_{\dd} = \Omega^{\bt+1}_\mr{closed}(\dd N)$ sends the class $[A]\in \MM^\p$ to $A|_{\dd N}\in \EL_{\dd}$ (which is well defined).
\item Fibers of $\pi_*$ are isomorphic to the relative cohomology $H^{\bt+1}(N,\dd N)$.
\item The map $\beta_{[A]}: T_{A|_{\dd N}} \FF_{\dd N} \ra T_{[A]}\MM^\p$ sends $\alpha_\dd\in \Omega^\bt(\dd N)$ to $[d\tilde{\alpha}]$ for arbitrary $\tilde{\alpha}\in \Omega^{\bt}(N)$ such that $\tilde{\alpha}|_{\dd N} = \alpha_\dd$.
\item Taking the quotient of $\MM^\p$ by the distribution $b= \mr{im}(\beta)$ gives the usual $\EL$-moduli space
$$\MM^\p/b = \MM= H^{\bt+1}(N)$$
which is the absolute de Rham cohomology of $N$.
\item Exact sequence (\ref{LES of tangent spaces to moduli spaces}) is the usual long exact sequence of relative cohomology
$$ \cdots\ra H^\bt(\dd N)\ra H^{\bt+1}(N,\dd N) \ra H^{\bt+1}(N)\ra H^{\bt+1}(\dd N)\ra \cdots $$
\item The symplectic form $\omega$ on $\Omega^\bullet(N)$ descends to
$$\underline{\omega}([A], [B])= \int_N A \wedge B$$
on $\cM^\p$. Lefschetz duality (\ref{Lefshetz pairing}) is the usual Lefschetz duality between absolute and relative cohomology, plus Poincar\'e duality for the cohomology of the boundary.

\end{itemize}

\subsection{The gluing (cutting) of symplectic $\EL$-moduli  spaces}

\subsubsection{The non-reduced case}
Assume that a manifold $N$ is cut into two pieces $N_1$, $N_2$ along a codimension 1 submanifold $N^\dd_2$ such that $\pa N_1=N_1^{\pa}\sqcup N_2^\pa$ and $\pa N_2=N_2^{\pa}\sqcup N_3^\pa$.
We have natural commutative diagram:

\begin{equation}
\label{FF fiber product}
\begin{CD}
\FF_N @> >> \FF_{N_1} \\
@VV V @VV V \\
\FF_{N_2} @> >> \FF_{N_2^\pa}
\end{CD}
\end{equation}

The space of fields $\cF_N$ is the subspace in the fiber product of spaces $\cF_{N_1}$ and $\cF_{N_2}$ over $\cF_{ N_{\pa}}$ consisting of fields which are smooth at $N_\pa$.

The Euler-Lagrange space for
$N$ is the fiber product of Euler-Lagrange spaces
for $N_1$ and for $N_2$ :
\[
\EL_N=\EL_{N_1}\times_{\EL_{N_2^\pa}}\EL_{N_2}
\]

\subsubsection{The gluing}

The gluing for symplectic $\EL$-moduli spaces goes as follows. Let $N_1, N_2$ be two spacetime manifolds with boundaries $\pa N_1= N^\pa_1 \sqcup N^\pa_2$ and $\pa N_2= (N^\pa_2)' \sqcup N^\pa_3$ respectively. Assume that $N_2^\pa$ diffeomorphic to $(N^\pa_2)'$, and denote by $N$ the result of
gluing $N_1$ and $N_2$ along the common boundary component:
$$N=(N_1\sqcup N_2)/(N^\pa_2\sim (N^\pa_2)')$$

The space $\cM^\p_N$ can be constructed intrinsically in terms of $\cM^\p_{N_1}$ and $\cM^\p_{N_2}$ as follows:
\begin{enumerate}[(i)]
\item First consider the fiber product
\be \widetilde{\cM}=\cM_{N_1}^\p\times_{\EL_{N^\dd_2}} \cM_{N_2}^\p \label{M tilde}\ee
For any point $\tilde{X}\in\widetilde\MM$, we have a map
\be \tilde\beta_{\tilde{X}} 
: T_{\pi_1(\tilde{X})}\cF_{N^\pa_1} \times T_{\pi_2(\tilde{X})}\cF_{N^\pa_2} \times T_{\pi_3(\tilde{X})}\cF_{N^\pa_3} \ra T_{\tilde{X}}\widetilde{\cM}
\label{beta tilde}\ee
induced from the two maps
\begin{eqnarray*}
(\beta_1\times \beta_2)_{[X_1]}: T_{\pi_1(X_1)}\cF_{N^\pa_1}\times T_{\pi_2(X_1)}\cF_{N^\pa_2} &\ra& T_{[X_1]}\cM^\p_{N_1} \\
(\beta_{2'}\times \beta_3)_{[X_2]}: T_{\pi_{2'}(X_2)}\cF_{(N^\pa_2)'}\times T_{\pi_3(X_2)}\cF_{N^\pa_3} &\ra& T_{[X_2]}\cM^\p_{N_2}
\end{eqnarray*}
\item  
The space $\cM_N^\p$ now can be identified with the leaf space of the distribution $\tilde{\beta}(0\times  T\cF_{N^\pa_2}\times 0)$ on $\widetilde{\cM}$:
\[
\cM_N^\p=\widetilde{\cM}/\tilde{\beta}(0\times  T\cF_{N^\pa_2}\times 0)
\]
It inherits the quotient distribution
\[
b_N=b_N^1\times b_N^3= \mr{im}(\tilde\beta)/\tilde{\beta}(0\times T\cF_{N^\pa_2}\times 0)
\]
parameterized by
$T\cF_{N^\pa_1}\times T\cF_{N^\pa_3}$.
\end{enumerate}

The important point here is that the gluing of symplectic $\EL$-moduli spaces is done in intrinsic terms of
 the symplectically reduced picture, i.e. in terms of the ingredients of diagram (\ref{beta diagram}): $(\cM^\p, \cF_\pa, Q_\pa, \pi_*, \beta)$.

\subsubsection{Gluing tangent spaces} The construction of gluing of symplectic $\EL$-moduli
spaces described above implies  the following Mayer-Vietoris type long exact sequence for $T_{[X]}\cM^\p$:
\begin{multline} \cdots \ra T_{[\pi_2(X)]}[-1] \cM_{N_2^\pa} \xra{(\beta_2-\beta_{2'})_*} T_{[X]}\cM_N^\p \ra \\
\ra T_{[X_1]} \cM^{\p,\abs_2}_{N_1}\oplus T_{[X_2]} \cM^{\p,\abs_{2'}}_{N_2} \ra T_{[\pi_2(X)]}\cM_{N_2^\pa}\ra \cdots \label{Meyer-Vietoris for part reduced mod spaces}\end{multline}
Here $\cM^{\p,\abs_2}_{N_1}$ is the quotient of $\cM^\p_{N_1}$ by the distribution $b_2$.
Similarly,  $\cM^{\p,\abs_{2'}}_{N_2}$ is the quotient of $\cM^\p_{N_2}$ by the distribution $b_{2'}$. \footnote{An equivalent description: $\cM^{\p,\abs_2}_{N_1}=\EL_{N_1}/Vect_Q^{\rel_1}$ where $Vect_Q^{\rel_1}$ is the distribution generated by vector fields of the form $[Q_{N_1},v]$ with $v$ tangent to the fibers of $\pi_1$. Likewise, $\cM^{\p,\abs_{2'}}_{N_2}=\EL_{N_2}/Vect_Q^{\rel_3}$.}

We have a similar long exact sequence for the tangent space to the usual $\EL$-moduli space:
\be \cdots \ra T_{[\pi_2(X)]}[-1] \cM_{N^\dd_2} \ra T_{[X]} \cM_N \ra T_{[X_1]} \cM_{N_1}\oplus T_{[X_2]} \cM_{N_2} \ra T_{[\pi_2(X)]}\cM_{N^\dd_2}\ra \cdots \label{Meyer-Vietoris for absolute mod spaces}\ee
Here $[X]$ is the space of leaves of $\Vect_Q$ through $X$.

\subsubsection{Example: gluing in abelian Chern-Simons theory}

Here we will give an example of gluing symplectic $\EL$-moduli spaces. Let $N=N_1\cup_{N_2^\dd} N_2$ as before. Then
\[
\widetilde{\MM}= \frac{\{A\in \Omega^{\bt+1}(N) \,|\, dA=0\}}{A\sim A+ d\alpha\quad \mbox{for any}\; \alpha\in \Omega^\bt(N)\;\mbox{s.t.}\; \alpha|_{N^\dd_1}=\alpha|_{N^\dd_2}=\alpha|_{N^\dd_3}=0}
\]
Mapping (\ref{beta tilde}) acts as

\[
\tilde{\beta}_{[A]}: (\alpha^\dd_1,\alpha^\dd_2,\alpha^\dd_3) \mapsto [d\tilde\alpha]
\]

Here $(\alpha^\dd_1,\alpha^\dd_2,\alpha^\dd_3)\in T_{A|_{N^\dd_1}}\FF_{N^\dd_1}\times T_{A|_{N^\dd_2}}\FF_{N^\dd_2}\times T_{A|_{N^\dd_3}}\FF_{N^\dd_3}$ and $\tilde\alpha\in \Omega^\bt(N)$ is any form such that $\tilde\alpha|_{N^\dd_k}=\alpha^\dd_k$ for $k=1,2,3$. The class $[d\tilde\alpha]\in T_{[A]}\widetilde{\MM}$
does not depend on the choice of $\tilde\alpha$. To pass
from from $\widetilde\MM$ to $\MM_N^\p$ we should mod out differentials of forms on the interface $N_2^\dd$ extended to $N$.
Now we can write the symplectic $\EL$-moduli space as the quotient
\begin{multline*}
\MM_N^\p= \widetilde\MM/\tilde\beta(0\times \Omega^\bt(N_2^\dd)\times 0)\simeq \\
\simeq \frac{\{A\in \Omega^{\bt+1}(N) \,|\, dA=0\}}{A\sim A+ d\alpha\quad \mbox{for any}\; \alpha\in \Omega^\bt(N)\;\mbox{s.t.}\; \alpha|_{N^\dd_1}=\alpha|_{N^\dd_3}=0}
\end{multline*}
It is equipped with two commuting distributions $b_N^1, b_N^3$ parameterized by $\Omega^\bt(N^\dd_1), \Omega^\bt(N^\dd_3)$ respectively.

Mayer-Vietoris sequence (\ref{Meyer-Vietoris for part reduced mod spaces}) becomes the following:
\begin{multline*}
\cdots\ra H^\bt(N_2^\dd) \ra H^{\bt+1}(N, N_1^\dd\sqcup N_3^\dd) \ra \\
\ra H^{\bt+1}(N_1, N_1^\dd) \oplus H^{\bt+1}(N_2, N_3^\dd) \ra H^{\bt+1}(N_2^\dd) \ra\cdots
\end{multline*}
whereas the version for usual $\EL$-moduli spaces (\ref{Meyer-Vietoris for absolute mod spaces}) reads
$$ \cdots \ra H^\bt(N_2^\dd)\ra H^{\bt+1}(N) \ra H^{\bt+1}(N_1)\oplus H^{\bt+1}(N_2) \ra H^{\bt+1}(N_2^\dd)\ra\cdots $$

\subsection{Higher codimensions}

The BV-BFV theory is the extension of the BV theory to manifolds with boundary.
In a similar way the BV theory extends to higher codimension submanifolds.

\begin{definition} The following collection of data is called a {\bf $k$-extended}
BV theory in dimension $n$.

For each $i=0,\dots,k-1$,
a $k$-extended BV theory assigns to every $(n-i)$-dimensional manifold $N_i$ with boundary $N_{i+1}$:
\begin{enumerate}
\item a space of fields $\F_{N_i}$ which is a graded manifold with exact symplectic form $\omega_{N_i}=(-1)^{n-i}\delta\alpha_{N_i}$ with $gh(\omega_{N_i})=gh(\alpha_{N_i})=i-1$, and with cohomological vector field $Q_{N_i}$,
\item a projection $\pi_{i}: \F_{N_{i}}\to \F_{N_{i+1}}$,
\item an action functional $S_{N_i}$ on $\F_{N_i}$ with $gh(S_{N_i})=i$.
\end{enumerate}

These data should satisfy the following axioms:

\begin{itemize}
\item $\delta\pi_{i-1}(Q_{N_{i-1}})=Q_{N_i}$,
\item $\iota_{Q_{N_{i-1}}}\omega_{N_{i-1}}=
    (-1)^{n-i+1}\delta S_{N_{i-1}}+\pi_{i}^*(\alpha_{N_{i}})$
\end{itemize}
We call {\bf $k$-extension} of a given BV theory in dimension $n$ a $k$-extended
theory with the original data for $i=0$.
\end{definition}



\begin{definition} A BV field theory has {\bf length} $k$ if $k$ is the maximal number such that
in its $k$-extension   for $i=1, \dots, k$,  all $\pi_{i-1}(\EL_{N_{i-1}})=\cL_{N_{i-1}}\subset \EL_{N_{i}}\subset \cF_{N_{i}}$ are Lagrangian
and all $Q$-reduced fibers are finite dimensional. Here $\EL_{N_i}\subset \cF_{N_i}$ is the set of zeroes of the
vector field $Q_{N_i}$. If $k$ is equal to $n$, we say that the theory is {\bf maximally extended}.
\end{definition}

Usually BV field theories have length $1$. For example, scalar field theory in dimension
greater than one or Yang-Mills
theory in dimension greater than two. We never consider theories of length zero.
In the rest of the paper we will show that scalar field theory in dimension $1$,
Yang-Mills theory in dimension $2$ and
all AKSZ theories are maximally extended (this includes $BF$, Chern-Simons theories and the Poisson sigma model).
Notice that scalar theory in dimension one is just quantum mechanics and
the Yang-Mills theory is dimension $2$ is known to be almost topological
\cite{W-QCD} (meaning that it only depends on the topology of the
spacetime manifold and on its volume).




\begin{remark}[Extended BV theories on manifolds with corners]
A $k$-extended BV theory in dimension $n$ naturally leads to an associated theory on $n$-dimensional manifolds with corners up to codimension $k$. Namely, to such a manifold $N$ we associate the data for $N_0=N$ and $N_1$ the union of the codimension one\footnote{Here and in the following, the codimension of a boundary stratum is  computed in terms of the bulk manifold.} strata in $\pa N$. To each such stratum $N'\subset\pa N$ we associate the data for  $N_1=N'$ and $N_2$ the union of the codimension one strata in
$\pa N'$ (notice that this is the union of only some codimension two strata in $\pa N$). To each
such stratum $N''\subset\pa N'$ we associate the data for  $N_2=N''$ and $N_3$ the union of the codimension one strata in $\pa N''$
(notice that this is the union of only some codimension two strata in $\pa N'$ and in turn of some codimension three strata in $\pa N$),
and so on.
\end{remark}

\subsection{Boundary conditions}

In the current paper we insist on having free boundary conditions. An alternative approach, first explored in \cite{CF-AKSZ,JP}, fixes boundary conditions in
a way compatible with the BFV structure on the boundary. This corresponds to choosing a Lagrangian submanifold $\cL$  of the BFV space of boundary fields to which the boundary
cohomological vector field is tangent and
on which the boundary $1$-form vanishes. Equivalently,  the boundary action should vanish on $\cL$. We call such Lagrangian submanifolds {\it adapted}.

At first sight, it would also seem natural to impose $\cL$ to be transversal to the Lagrangian submanifold
$L_M$. However, this condition is too restrictive and rules out a lot of interesting boundary conditions. Instead, a better assumption would be that
the intersection of $\cL$ with $L_N$ should be finite dimensional after reduction.


As shown in \cite{CF-AKSZ},
in the case of AKSZ theories,  one can obtain an adapted Lagrangian submanifold $\cL$ by choosing an adapted Lagrangian submanifold $L_i$ of the target manifold $\cM$ for each boundary component $\partial_i N$ of the spacetime $N$. The submanifold  $\cL=\prod_i\Map(T[1]\partial_iN,L_i)$ in this case is an adapted
Lagranian in $\cF_{\pa N}$. Notice that usually one does not require the decomposition
$\partial M=\cup_i\partial_iM$ to be disjoint, rather the pairwise intersections are assumed to be of lower dimension. The adapted Lagrangian submanifolds
of the target are usually referred to as {\it branes}.

In the case when the target is $T^*[1]P$, with $P$ a Poisson manifold, one can easily show that a brane is necessarily of the form
$N^*[1]C$, where $C$ is a coisotropic submanifold of $P$ and $N^*C$ denotes its conormal bundle.
When the target is a differential graded symplectic manifold (with symplectic form of degree $2$) associated to a Courant algebroid over a manifold $N$, then a brane covering $N$ is the same as a Dirac structure.

An intermediate type of boundary condition consists in splitting the boundary into two components and in choosing an adapted Lagrangian submanifold of the BFV space of fields on the first
component while keeping free boundary conditions for the second
component. The analysis performed in this paper still holds for  boundary fields of the second component. This mixed approach might have several applications. For example \cite{CC}, the study of the Poisson Sigma Model on the disk whose boundary is split into an even number of ordered intervals $I_i$ with boundary
conditions on $I_{2s}$ determined by the $C=P$ $\forall s$ leads to the construction of what is known as the relative symplectic groupoid integrating $P$.

\subsection{$gh=0$ part of the BV-BFV theory}

\subsubsection{Non reduced theory}
The $gh=0$ part of the BV-BFV theory is a first order
classical field theory with the space of fields $F_N=\cF_N^{(0)}$
(the $gh=0$ part of $\cF_N$),
the classical action $S^{(0)}_N=S_N|_{\cF_N^{(0)}}$.

The $gh=0$ part of the space of boundary fields $F_{\pa N}=\cF_{\pa N}^{(0)}$
is a symplectic manifold with the symplectic form $\omega^{(0)}_{\pa N}=\delta\alpha^{(0)}_{\pa N}$ which is
the $gh=0$ part of the symplectic form $\omega_{\pa N}$ on
boundary BFV fields. The one-form $\delta\alpha^{(0)}_{\pa N}$
is the $gh=0$ part of the form $\alpha_{\pa N}$ and it determined by the $gh=0$ part of the classical action.

Gauge transformations are $gh=0$ part of the BV-BFV gauge
transformations generated by vector fields $[V,Q]$ and they are Hamiltonian on the boundary.

The $gh=0$ part of the coisotropic submanifold $\EL_{\pa N}\subset \cF_N$ is a coisotropic submanifold $C_{\pa N}=\EL_{\pa N}^{(0)}\subset F_N$. It consists of boundary fields which can be extended to a solution to Euler-Lagrange equations in the vicinity of the boundary in $N$.

In regular BV theories the space of solutions to the Euler-Lagrange equations $EL_N=\EL^{(0)}_N$ projects to the Lagrangian subspace $L_N=\cL_N^{(0)}=\pi(EL_N)\subset C_{\pa N}\subset F_{\pa N}$.

\subsubsection{The reduction}
The classical moduli space $EL_N/G_N$ is the $gh=0$ part of the $\EL$-moduli space $\mathcal{EL}_N/Q_N$ and maps to the reduced phase space $C_{\pa N}/G_{\pa N}=\underline{C_{\pa N}}$, which is in turn the $gh=0$ part of the boundary $\EL$-moduli space $\underline{\EL_{\pa N}}$. Here $G_N$ denotes the distribution on $EL_N$ induced by $\mathrm{Vect}_{Q_N}$, likewise for $G_{\pa N}$; also $G_{\pa N}$ can be seen as the coisotropic distribution on $C_{\pa N}$.

In the regular case, the image of
\begin{equation}\pi_*: EL_N/G_N \rightarrow \underline{C_{\pa N}}\label{pi_*}\end{equation}
is the reduced Lagrangian $L_N/G_{\pa N}=\underline{L_N}$.

\begin{remark}
In regular case, there is also the following relation between the smooth loci (cf. Appendix \ref{smooth}) of moduli spaces:
\begin{equation}
\begin{CD}
T^*_\mathrm{vert}[-1](EL_N/G_N)^\mathrm{smooth} @>\subseteq>> (\EL_N/Q_N)^\mathrm{smooth} \\
@V\pi_*VV @V\pi_*VV \\
(\underline{C_{\pa N}})^\mathrm{smooth} @>\subseteq>> (\underline{\EL_{\pa N}})^\mathrm{smooth}
\end{CD}
\label{relation between BV mod spaces and gh=0 part}
\end{equation}
Here $T^*_\mathrm{vert}$ denotes the dual of the vertical tangent bundle of the fibration (\ref{pi_*}).
Horizontal arrows in (\ref{relation between BV mod spaces and gh=0 part}) may be strict inclusions (e.g. abelian Chern-Simons theory where $\EL$-moduli spaces contain additional smooth pieces: $H^0(N)$ and $H^3(N)$, cf. section \ref{ab CS reduced}), or may be equalities (e.g. non-abelian Chern-Simons theory with a simple gauge group $G$).
\end{remark}


\section{The BFV category}\label{bfv-category}

\subsection{Spacetime categories}
Recall that an $n$-dimensional spacetime category is the category of $n$-dimensional cobordisms
which may have additional structure (smooth, Riemannian etc.). See \cite{At,Seg,ST}
for the discussion of various examples.

In most general terms {\it objects} of a {\it $d$-dimensional
spacetime category} \index{spacetime category} are
$(d-1)$-dimensional manifolds (space manifolds). In specific
examples of spacetime categories, space manifolds are equipped with
a structure (orientation, symplectic structure, Riemannian metric, etc.).

A {\it morphism} between two space manifolds $\Sigma_1$ and
$\Sigma_2$ is a  $d$-dimensional manifold $M$, possibly with a
structure (orientation, symplectic, Riemannian metric, etc.), together with the identification of $\Sigma_1\sqcup \overline{\Sigma_2}$ with the boundary of $M$. Here $\overline{\Sigma}$ is the manifold
$\Sigma$ with reversed orientation.

{\it Composition} of morphisms is the gluing along the common boundary. Here are examples of spacetime categories.

\vspace{0.5cm}

{\bf The $d$-dimensional topological category}. \index{topological
category} Objects are smooth, compact, oriented $(d-1)$-dimensional
manifolds. A morphism between $\Sigma_1$ and $\Sigma_2$ is a $d$-dimensional smooth compact oriented manifold
with $\pa M=\Sigma_1\sqcup \overline{\Sigma_2}$. The orientation on $M$
should agree with the orientations of $\Sigma_i$ in a natural way.
The composition consists of gluing two morphisms along the common boundary.

\vspace{0.5cm}

{\bf The $d$-dimensional Riemannian category}.  Objects are oriented $(d-1)$ Riemannian manifolds with collars.  Morphisms
between two objects  $N_1$
and $N_2$ are  oriented $d$-dimensional Riemannian manifolds $M$ ,
such that $\pa M =N_1\sqcup \overline{N_2}$. The orientation on all
three manifolds should naturally agree, and the metric on $M$ agrees
with the metric on $N_1$ and $N_2$ on a collar of the boundary. The
composition is the gluing of such Riemannian cobordisms. For
details see \cite{ST}.

\vspace{0.5cm}

{\bf The $d$-dimensional metrized cell complexes}. Objects are $(d-1)$-dimensional oriented metrized cell complexes (edges have length,
2-cells have area, etc.). A morphism between two such complexes
$C_1$ and $C_2$ is an oriented metrized $d$-dimensional cell complex $C$ together with two
embeddings of metrized cell complexes $i: C_1\hookrightarrow C$,
$j: \overline{C_2}\hookrightarrow C$ where $i$ is orientation reversing and $j$ is orientation preserving. The composition is the gluing of
such triples along the common $(d-1)$-dimensional subcomplex.

This is the underlying category for all lattice models
in statistical mechanics.

\vspace{0.5cm}

{\bf The Pseudo-Riemannian category} The difference between this category and the Riemannian category is that morphisms
are pseudo-Riemannian with the signature $(d-1,1)$.
This is the most interesting category for physics.
When $d=4$ it represents the spacetime structure of our universe.

\subsection{The BFV category}\label{bfv-cat}

The category $\mathsf{BFV}$
has the following objects, morphisms and compositions of morphisms.

{\it Objects} of $\mathsf{BFV}$ are triples $(\cF, \alpha, Q)$ where $\cF$ is an exact graded symplectic manifolds with the symplectic form $\omega=d\alpha$ with ghost number $0$, and $Q$ is a cohomological vector field (i.e. its Lie derivative squares to zero) with ghost number $1$. The symplectic form should be preserved by $Q$:
\[
L_Q\omega=0
\]
Let $\EL$ be the space of zeroes of the vector field $Q$.
The restriction of the Lie subalgebra $Vect_Q=[Vect(\cF), Q]\subset Vect(\cF)$ defines an involutive\footnote{ We will
always assume that this distribution is actually integrable, which is not automatic in
the infinite dimensional case.} distribution on $\EL$.

{\it Morphisms} between $(\cF_1, \alpha_1, Q_1)$ and $(\cF_2,\alpha_2, Q_2)$ are differential graded manifolds $\cF$ with symplectic
form $\omega^\cF$ with ghost number $-1$, with cohomological vector
field $Q^\cF$ with ghost number $1$, with the function $S^\cF$ (action function) on $\cF$ with $gh(S^\cF)=0$, and with two projection mappings $\pi_i: \cF\to \cF_i$. These data should satisfy the following conditions:

\begin{itemize}

\item Projections $\pi_i$ are mappings of differential graded manifolds, i.e. $\delta\pi_i(Q^\cF)=Q_i$.

\item  The following identity should hold
\[
\iota_{Q^\cF}\omega^\cF=(-1)^n\delta S^\cF-\pi^*_1(\alpha_1)+\pi^*_2(\alpha_2),
\]

\item Let $\EL^\cF$ be the zero-locus of the vector field $Q^\cF$, then $\cL^\cF=(\pi_1\times\pi_2)(\EL^\cF)\subset \overline{\cF_1}\times \cF_2$ should be Lagrangian.
    Here $\overline{\cF}_1$ is the dg manifold $\cF_1$ with symplectic form $-\omega_1$.

\item For each Lagrangian submanifold $\cL$ which is generic, relative to $\cL^\cF$ (the intersection is transversal), the preimage $(\pi_1\times \pi_2)^{-1}(\cL)\subset \cF$ should be symplectic.

\end{itemize}

{\it Composition of morphisms} Let $\cF: \cF_1\to \cF_2$ and  $\cF': \cF_2\to  \cF_3$ be two morphisms in $\mathsf{BFV}$.
The composition $\cF'\cdot \cF$ is the fiber product $\cF\times_{\cF_2}\cF'=\{(x,x')\in \cF\times \cF'| \pi_2(x)=\pi_2(x')\}$. It is a submanifold in $\cF'\times\cF$.

The symplectic form on  $\cF'\cdot \cF$ is the pullback of the symplectic form on $\cF'\times \F$. The vector field $Q+Q'$ on $\cF'\times \F$ is tangent to $\cF'\cdot \cF\subset\cF'\times \F$
and it induced the vector field $Q^{comp}$ on the fiber product.
The action function is additive: $S^{\cF'\cdot\cF}(x,x')=S^{\cF}(x)+S^{\cF'}(x')$.

Let $f$ be a mapping which assigns to each object $\cF$ of $\mathsf{BFV}$ the  functional $f^\cF$ on $\cF$.
Define the mapping $F_f : \mathsf{BFV}\to \mathsf{BFV}$  as as follows. It acts trivially on $\cF$. It acts on one forms $\alpha_\Sigma$ as
\[
\alpha^\cF\to \alpha^\cF+\delta f^\cF
\]
and does not change $Q^\cF$. On the morphism $(\cF, \omega^\cF, Q^\cF, S^\cF):\cF_1\to \cF_2$ this mapping acts it as follows. It acts
trivially on $\omega^\cF$, on $Q^\cF$, and on $\cF$ while
on the action $S^\cF$ it acts as
\[
S^\cF\to S^\cF+\pi_1^*(f^{\cF_1})-\pi_2^*(f^{\cF_2})
\]

It is easy to see that $F$ is a covariant endofunctor for $\mathsf{BFV}$.

\begin{remark} The notion of BV-BFV category introduced above
has a natural generalization where the objects are quadruples
$(\cF, {\mathfrak L}, \alpha, Q)$. Here $\cF$ is a $\ZZ$-graded manifold, ${\mathfrak L}$ is a line bundle over it, $\alpha$ is
a connection on ${\mathfrak L}$ and $Q$ is a cohomological degree $1$ vector field on $\cF$.

A morphism between $(\cF_1, {\mathfrak L}_1, \alpha_1, Q_1)$ and $(\cF_2, {\mathfrak L}_2, \alpha_2, Q_2)$ is a $\ZZ$-graded manifold with the same data as before but instead of $S^\cF$
being a function on $\cF$ we have a section of the line bundle ${\mathfrak L}^\cF=\pi^*({\mathfrak L}_1\times {\mathfrak L}_2)$ over $\cF$  which is horizontal for the flat connection $(-1)^{n-1}\frac{i}{\hbar}(\iota_{Q^\cF}\omega^\cF+\pi_1^*(\alpha_1)-\pi_2^*(\alpha_2))$
on ${\mathfrak L}^\cF$, cf. Remark \ref{line-bun}.
\end{remark}

\subsection{The BV functor}
Now classical BV-BFV theories can be regarded as functors from
spacetime categories to BFV category.

Fix a spacetime category.
A classical BV-BFV field theory
for a spacetime from this category defines the covariant  functor,
the {\it BV functor}, from the spacetime category to the BFV category.

{\it On objects}: The BV functor assigns the space of fields
$\cF_\Sigma$ to the object $\Sigma$ of the spacetime category
with the BFV data $\alpha_\Sigma$ and $Q_\Sigma$.

{\it On morphisms}: The BV functor assigns the space of
fields $\cF_N$ for a morphism $N:\Sigma_1\to \Sigma_2$
with the BV data $\omega_N$, $Q_N$, $S_N$, and mappings
$\pi_i: \cF_N\to \cF_{\Sigma_i}$.

The properties of the BV-BFV field theories guarantee
that this mapping is a covariant functor.

\begin{remark} The BFV category is $1$-category. The corresponding BV functor is a functor from $1$-category of cobordisms to the
BFV category. This can be extended to higher categories.
The natural target structure is a $k$-extended BFV category.
It is a $k$-category. The $k$-extended classical BV field theory is the $k$-functor from the $k$-category of $k$-cobordisms
to the $k$-extended BFV category, similar to $k$-extended
topological quantum field theories \cite{BD,Lurie}. We will
discuss extended theories in another publication.
\end{remark}

\section{Examples of BV-BFV theories}
\label{sec: examples}

\subsection{Electrodynamics}\label{bv-electro}

Here we will consider the BV-extended classical Euclidean electrodynamics
in the trivial $U(1)$-bundle. Spacetime manifolds in Euclidean electrodynamic are smooth oriented $n$-dimensional Riemannian manifolds. By $*:\Omega^i(N)\to \Omega^{n-i}(N)$ we denote the
Hodge operation induced by the metric on $N$.

\subsubsection{The BV-BFV structure for classical electrodynamics}
The space of fields in the BV-extended classical
electrodynamics on the spacetime manifold $N$ is $T^*[-1]E_N$.
Here $E_N=\Omega^1(N)\oplus \Omega^{n-2}(N)\oplus \Omega^0(N)[1]$
where the first summand is the space of connections $A$ in the trivial
$U(1)$ bundle over $N$, the second summand is the space of fields $B$, the Hamiltonian counterpart (``momentum") of $A$, the third summand is the space of ghost fields. The total space of fields
in BV classical electrodynamics is
\[
\FF_N=\Omega^1(N)\oplus \Omega^{n-2}(N)\oplus \Omega^0(N)[1]\oplus
\Omega^{n-1}(N)[-1]\oplus \Omega^{2}(N)[-1]\oplus \Omega^n(N)[-2]
\]
We will use notations $A, B, c, A^\dag, B^\dag, c^\dag$ for
the fields from corresponding summands\footnote{Here we discuss the \emph{minimal} BV extension of Hamiltonian classical electrodynamics.}. Here we regard $\Omega^k(N)$ as a vector space concentrated in degree zero.

The BV-symplectic form on the space of fields is the canonical
symplectic form on $T^*[-1]E_N$:
\[
\omega_N=\int_N (\delta A\wedge \delta A^\dag+ \delta B\wedge \delta B^\dag +\delta c\wedge \delta c^\dag)
\]
It has degree $-1$.

The BV-extended action of the classical electrodynamics on $N$
has degree zero:
\[
S_N=\int_N \left(B\wedge F(A)+\frac{1}{2} B\wedge *B+ A^\dag \wedge dc\right),
\]
where $F(A)=dA$ is the curvature of the connection $A$.
The vector field $Q_N$ is
\[
Q_N=\int_N \left(dc\wedge \frac{\delta}{\delta A} + dB\wedge \frac{\delta}{\delta A^\dag}+ (*B+dA)\wedge \frac{\delta}{\delta B^\dag}+dA^\dag\wedge \frac{\delta}{\delta c^\dag}\right)
\]
and it has $gh=1$.

It acts on coordinate fields as
\[
Q_NA=dc, \ \ Q_N A^\dag =dB, \ \ Q_N B^\dag= *B+dA, \ \ Q_Nc^\dag =dA^\dag
\]
On other coordinate fields $Q_N$ acts trivially. Here and below we are using the same notation for the vector field $Q_N$ an for its Lie derivative.

The boundary BFV theory has the space of fields
\[
\F_{\pa N}=\Omega^1(\pa N)\oplus
\Omega^{n-2}(\pa N)\oplus\Omega^0(\pa N)[1]\oplus\Omega^{n-1}(\pa N)[-1]
\]
We will denote corresponding fields by $A, B, c, A^\dag$ respectively. The projection $\pi: \F_N\to \F_{\pa N}$ acts as
\[
\pi(A)=i^*(A), \ \  \pi(B)=i^*(B), \ \ \pi(c)=i^*(c), \ \ \pi(A^\dag)=i^*(A^\dag), \ \  \pi(B^\dag)=0, \ \ \pi(c^\dag)=0
\]

The boundary symplectic form is the differential
of the form
\[
\alpha_{\pa N}=\int_{\pa N} (B\wedge \delta A + A^\dag \wedge \delta c)
\]
\[
\omega_{\pa N}=\delta\alpha_{\pa N}=\int_{\pa N}(\delta B\wedge \delta A+\delta A^\dag\wedge \delta c)
\]

The boundary vector field $Q_{\pa N}=\delta\pi Q_N$ is
\begin{equation}\label{qpa}
Q_{\pa N}=\int_{\pa N} \left(dB\wedge \frac{\delta}{\delta A^\dag}+dc\wedge \frac{\delta}{\delta A}\right)
\end{equation}
The boundary action is
\[
S_{\pa N}=\int_{\pa N} c\wedge dB
\]

\begin{Proposition}
The data described above satisfy the BV-BFV axioms.
\end{Proposition}

\begin{proof}
The only non-trivial computation is to check the classical
master equation (\ref{cme}) when $\pa N\neq \emptyset$. Contracting
the vector field $Q_N$ with the symplectic form $\omega_N$
we obtain
\[
\iota_{Q_N}\omega_N=\int_N(dc\wedge \delta A^\dag+\delta A\wedge dB+\delta B\wedge (*B+dA)+dA^\dag\wedge \delta c)
\]
The differential of the action is easy to compute:
\[
\delta S_N=\int_N(\delta B\wedge dA+B\wedge d\delta A+\delta B\wedge *B+ \delta A^\dag\wedge dc +
A^\dag\wedge d\delta c)
\]
Comparing these two formulae and using the Stokes formula we obtain the classical master equation .
\end{proof}

\begin{remark}The connection field $A$ in electrodynamics is called the vector potential. When $n=4$ and $N=[t_1,t_2]\times M$ is equipped with Minkowsky metric choose a basis $e_0, e_1, e_2,e_3$ in the tangent space where $e_0$ is a time direction. On $EL_N$ the components $B_{0i}$, $i=1,2,3$ give the magnetic
field on $M$ and the components $B_{ij}$  give the electric field.
\end{remark}

\subsubsection{The $Q$-reduction of $\EL_N$} The Euler-Lagrange equations in the bulk are:
\[
dB=0, \ \ B=-*dA, \ \ dA^\dag=0, \ \ dc=0
\]
Note that this implies the usual Maxwell's equation for the vector potential $d^*dA=0$.

This defines the subspace $\EL_N\subset \F_N$:
\begin{multline}\label{el-maxw}
\EL_N
=\{(A,B)\in \Omega^1(N)\oplus \Omega^{n-2}(N)\; |\; d*dA=0, B=-*dA\} \oplus \\ \oplus \Omega^0_{closed}(N)[1]
\oplus \Omega^{n-1}_{closed}(N)[-1]\oplus \Omega^2(N)[-1]\oplus\Omega^n(N)[-2]
\end{multline}
The summands correspond to fields $A, B, c, A^\dag, B^\dag, c^\dag$ respectively.

\begin{Proposition} The $Q$-reduced space of solutions to the Euler-Lagrange
equations is
\begin{equation}\label{Q-red-b}
\EL_N/Q\simeq \Omega^1_{Maxw}(N)/\Omega^1_{exact}(N)\oplus H^0(N)[1]\oplus H^{n-1}(N)[-1]\oplus H^n(N)[-2]
\end{equation}
Here $\Omega^1_{Maxw}(N)$ is the space of solutions to
Maxwell's equations, i.e. $1$-forms $A$ , such that $d*dA=0$. The summands represent quotient spaces in fields
$A, c, A^\dag, c^\dag$ respectively.

\end{Proposition}
\begin{proof}
Let us first find the $Q$-reduced tangent space to $\EL_N$:
\[
T_X\EL_N/Q=\ker(\hat{Q}_X)/Im(\hat{Q}_X)
\]
Because $\EL_N$ is a vector space its tangent space, which is isomorphic to $\ker(\hat{Q}_X)$, is given by
(\ref{el-maxw}). The image of $\hat{Q}_X$ is easy to compute:
\begin{multline}
Im(\hat{Q}_X)=\Omega^1_{exact}(N)\oplus
\{ (A^\dag, B^\dag)\in \Omega^{n-1}(N)[-1]\oplus \\ \Omega^2(N)[-1] \; | \; A^\dag=d\beta, B^\dag=*\beta+ \mbox{exact}, \beta\in \Omega^{n-2}(N)\}\oplus \Omega^n_{exact}(N)[-2]
\end{multline}
Here the components correspond to fields $A,A^\dag, B^\dag, c^\dag$ respectively.
Let us prove now that
\begin{multline}\label{imqb}
\{  (A^\dag, B^\dag)\in \Omega^{n-1}(N)\oplus\Omega^2(N) \; | \; A^\dag=d\beta, \ \ \\ B^\dag=*\beta+ \mbox{exact}, \beta\in \Omega^{n-2}(N)\}=
 \Omega^{n-1}_{exact}(N)\oplus \Omega^2(N)
\end{multline}
By the Hodge-Morrey decomposition (see for example \cite{CGMT} and references therein)
we can write $\Omega(N)=\Omega_{coclosed}(N)+\Omega_{exact}(N)$. Using this decomposition,
for an exact $(n-1)$-form $A^\dag=d\gamma$
and an arbitrary $2$-form $B^\dag$ we can write $B^\dag-*\gamma=*\theta+d\eta$ with $\theta$ closed.
Then $A^\dag=d\beta$, $B^\dag=*\beta+d\eta$ for $\beta=\gamma+\theta$ and therefore the r.h.s. of
(\ref{imqb}) is the subspace of l.h.s. and since to opposite inclusion is obvious we proved
(\ref{imqb}).

Now we can write
\[
Im(\hat{Q}_X)=\Omega^1_{exact}(N)\oplus \Omega^{n-1}_{exact}(N)[-1]\oplus \Omega^2(N)[-1]\oplus \Omega^n_{exact}(N)[-2]
\]
Together with the formula (\ref{el-maxw}) this proves the Proposition.
\end{proof}

\begin{Proposition} When $\pa N=\emptyset $
\[
\EL_N/Q=H^1(N)\oplus H^0(N)[1]\oplus H^{n-1}(N)[-1]\oplus H^n(N)[-2]
\]
The summands correspond to fields $A,c,A^\dag, c^\dag$ respectively.
\end{Proposition}

\begin{proof} A $1$-form $A$ satisfies the Maxwell's equation
$d*dA=0$ if and only if $dA\in ker(d^*)$. The metric on $N$
gives the Hodge decomposition:
\[
\Omega(N)=H(N)\oplus \Omega_{d-exact}(N)\oplus \Omega_{d^*-exact}(N)
\]
where $H(N)$ are harmonic forms representing cohomology classes
of $N$. It is clear from this decomposition that the first summand
in  (\ref{Q-red-b}) is $H^1(N)$.
\end{proof}

When $N$ does not have a boundary 
the space $\EL_N/Q$ described above has a natural symplectic
structure given by the Poincar\'e pairing between $H^1(N)$
an $H^{n-1}(N)$ and between $H^0(N)$ and $H^n(N)$.

\subsubsection{The reduction of boundary structures} Recall that the space of boundary fields is
\[
\F_{\pa N}=\Omega^1(\pa N)\oplus
\Omega^{n-2}(\pa N)\oplus\Omega^0(\pa N)[1]\oplus\Omega^{n-1}(\pa N)[-1]
\]
Where the summands correspond to pullbacks of fields
$A, B, c, A^\dag$ respectively.
The Euler-Lagrange equations on the boundary
(equations for zeroes of the vector field $Q_{\pa N}$) are
\[
dB=0, \ \ dc=0
\]
Thus, the space of solutions to boundary Euler-Lagrange equations
is
\[
\EL_{\pa N}=\Omega^1(\pa N)\oplus \Omega^{n-2}_{closed}(\pa N)\oplus \Omega_{closed}^0(\pa N)[1]\oplus\Omega^{n-1}(\pa N)[-1]
\]
Because it is a vector space, it is isomorphic to its tangent space at every point.

Similarly to the discussion for the bulk, for $l\in \EL_{\pa N}$ we get
\begin{equation}\label{imq}
Im(\hat{Q}_l)=\Omega^1_{exact}(\pa N)
\oplus \Omega^{n-1}_{exact}(\pa N)[-1]\subset T_l \EL_{\pa N}
\end{equation}

Taking the quotient space $ker(\hat{Q}_{\pa N})/Im(\hat{Q}_{\pa N})$ and identifying the tangent space with the space itself we
prove the following.

\begin{Proposition} The $Q$-reduced space of boundary fields is
\[
\EL(\pa N)/Q=\Omega^1(\pa N)/\{exact\}\oplus \Omega^{n-2}_{closed}(\pa N)\oplus H^0(\pa N)[1]\oplus H^{n-1}(\pa N)[-1]
\]
Here summands correspond to fields $A, B, c, A^\dag$ respectively.
\end{Proposition}

This space is clearly infinite-dimensional. It coincides with
the symplectic reduction of $\EL_{\pa N}$ with symplectic structure given by the natural pairing between the first and the second and
between the third and the fourth
summands.

\begin{Proposition} BV-extended classical electrodynamics is a regular theory (in the sense of Definition \ref{def: regular}).
\end{Proposition}

\begin{proof} By direct calculation, we have
\begin{eqnarray}
\label{ED ker Q}\ker(\hat{Q}_X) &=& \underbrace{\{(a,b)\in \Omega^1\oplus\Omega^{n-2}\,|\, b=-*da,\, db=0 \}}_{\cong \Omega^1_{cl}\oplus \Omega^{n-2}_{coex,cl}}\oplus \\
\nonumber && \oplus \Omega^0_{cl}[1] \oplus \Omega^{n-1}_{cl}[-1] \oplus \Omega^2[-1] \oplus \Omega^n[-2] , \\
\label{ED im Q}Im(\hat{Q}_X) &=& \Omega^1_{ex}\oplus \Omega^{n-1}_{ex}[-1]\oplus \Omega^2[-1]\oplus \Omega^n_{ex}[-2] , \\
\label{ED ker Q^ver}\ker(\hat{Q}^\ver_X) &=& \Omega^1_{cl,D}\oplus \Omega^0_{cl,D}[1] \oplus \Omega^{n-1}_{cl,D}[-1] \oplus \Omega^2[-1] \oplus \Omega^n[-2] , \\
\label{ED im Q^ver}Im(\hat{Q}^\ver_X) &=& \Omega^1_{ex,D}\oplus \Omega^n_{ex,D}[-2] \oplus \\
\nonumber &&\oplus \underbrace{\{ (db, da+ *b) \in \Omega^{n-1}\oplus\Omega^2 \,|\, a\in \Omega^1_D,\, b\in \Omega^{n-2}_D\}}_{\cong \Omega^{n-1}_{ex,D}\oplus
\left(\Omega^2_{ex,D}\oplus \Omega^2_{cocl,N}\right)}[-1]
\end{eqnarray}
Here we use a shorthand notation with $ex, cl, coex, cocl, D,N$ standing for exact, closed, coexact, coclosed, Dirichlet, Neumann respectively (the last two indicating the imposed boundary condition; $\Omega_{ex,D}$ means exact, with the \textit{primitive} being subject to Dirichlet condition). Hodge-Morrey decomposition theorem implies that subspaces (\ref{ED ker Q}) and (\ref{ED im Q^ver}) are mutually orthogonal in $T_X\FF_N$, and (\ref{ED ker Q^ver}) and (\ref{ED im Q}) are mutually orthogonal too. Thus the assumption (\ref{assump bdry}) holds.

Also,
$$\ker(\hat{Q}_l)= \Omega^1_\pa\oplus \Omega^{n-2}_{cl,\pa} \oplus \Omega^0_{cl,\pa}[1] \oplus\Omega^{n-1}_\pa[-1]$$
and
$$Im(\hat{Q}_l)= \Omega^1_{ex,\pa}\oplus \Omega^{n-1}_{ex,\pa}[-1]$$
(where $\pa$ stands for forms on the boundary) are mutually orthogonal in $T_l\FF_{\dd N}$ due to Hodge decomposition on $\pa N$, thus the assumption (\ref{assump 1 for bdry}) also holds. Therefore electrodynamics is a regular theory.
\end{proof}

\subsubsection{The Lagrangian subspace $\cL_N$ and its reduction}

Now let us describe the evolution relation $\cL_N=\pi(\EL_N)\subset \F_{\pa N}$ and its reduction. Due to regularity, $\cL_N$ is Lagrangian. This subspace consists
of pullbacks of fields $A, B, c, A^\dag$ satisfying Euler-Lagrange equation in $N$:
\begin{multline}
\cL_N=\{(i^*(A), i^*(-*dA))\in \Omega^1(\pa N)\oplus \Omega^{n-2}(\pa N) \;|\; d*dA=0, A\in \Omega^1(N)\}\\
\oplus i^*(\Omega_{closed}^0(N))[1]\oplus
i^*(\Omega_{closed}^{n-1})[-1]
\end{multline}

The $Q$-reduction of the Lagrangian subspace $\cL_N$, or equivalently, its
symplectic reduction is:
\begin{multline}\label{red-evol-qcd}
\cL_N/Q=\{(i^*(A), i^*(-*dA))\in \Omega^1(\pa N)\oplus \Omega^{n-2}(\pa N) \;|\\ d*dA=0, A\in \Omega^1(N)\}/\{\mbox{ exact } A\} \oplus \tilde{H}^0(\pa N)[1]\oplus \tilde{H}^{n-1}(\pa N)[-1]
\end{multline}
Here $\tilde{H}^i(\pa N)$ are cohomology classes of $\pa N$ which
are pullbacks of cohomology classes on $N$.

\subsubsection{Gauge classes of solutions to Euler-Lagrange equations
with fixed gauge classes of boundary values}

It is clear that the tangent space $T_X \EL(N,[l])$ depends neither on $X$ nor on $l$ and is
\begin{multline}
\{(A, -*dA)\in \Omega^1(\pa N)\oplus \Omega^{n-2}(\pa N) \;|\; d*dA=0, i^*(A) \mbox{ exact }, A\in \Omega^1(N)\}\oplus \\ \oplus \Omega_{closed}^0(N, \pa N)[1]\oplus \{A^\dag\in \Omega_{closed}^{n-1}(N)[-1] \;|\; i^*(A^\dag) \mbox{ exact }\}\oplus \\ \oplus \Omega^2(N)[-1]\oplus \Omega^n(N)[-2]
\end{multline}

For the quotient space we have:
\begin{multline}\label{red-qcd}
T_X \EL(N,[l])/Im(\hat{Q}_X)=\\
=\frac{H^1(N,\pa N)}{H^0(\pa N)}\oplus \frac{H^{n-1}(N,\pa N)}{H^{n-2}(\pa N)}[-1]\oplus H^0(N, \pa N)[1] \oplus H^n(N)[-2]
\end{multline}
It is the $Q$-reduction of the tangent space $T_X \EL(N,[l])$. It is a finite dimensional
symplectic space. Since we are in a linear case, $\EL(N, [l])/Q$ is isomorphic to (\ref{red-qcd}).

\subsubsection{Codimension $2$ BV structure}
Here we will describe the extension of the BV electrodynamics
to codimension $2$ strata.
Let $\Sigma$ be an $(n-1)$-dimensional manifold with the
boundary $\pa \Sigma$. The boundary $\pa \Sigma$ is closed
and the space of boundary fields is
\[
\cF_{\pa \Sigma}=\Omega^{n-2}(\pa \Sigma)\oplus \Omega^0(\pa \Sigma)[1]
\]
where the summands correspond to pullbacks of
fields $B$ and $c$ from $\Sigma$ to the boundary. We will
denote the pullbacks of $B$ and $c$ by the same letters.

The restriction mapping $\pi: \F_\Sigma\to \F_{\pa \Sigma}$ acts as $\pi(A)=\pi(A^\dag)=0$ and $\pi(B)=B$, $\pi(c)=c$.

The symplectic structure on this space is exact
\[
\omega_{\pa \Sigma}=\int_{\pa \Sigma} \delta B\wedge \delta c=\delta\alpha_{\pa \Sigma}
\]
where
\[
\alpha_{\pa \Sigma}=\int_{\pa \Sigma}\delta B\wedge c
\]

The action and the vector field $Q$ on $\F_{\pa \Sigma}$ are
trivial:
\[
S_{\pa \Sigma}=0, \ \ Q_{\pa \Sigma}=0
\]

\begin{Proposition} The action $S_\Sigma=
\int_\Sigma c\wedge dB$ satisfies the equation (\ref{cme}).
\end{Proposition}
The proof is a straightforward computation.

It is clear that
\[
\EL_{\pa \Sigma}=\F_{\pa \Sigma},
\]
The evolution relation $\cL_{\Sigma}=\pi(\EL_{\Sigma})$
\[
\cL_{\Sigma}=i^*(\Omega^{n-2}_{closed}(\Sigma))\oplus i^*(\Omega^0_{closed}(\Sigma))[1]
\]
When there is only one connected component of the
boundary
\[
\cL_{\Sigma}=\Omega^{n-2}_{exact}(\pa \Sigma)\oplus
\Omega^0_{closed}(\pa \Sigma)[1]
\]

Because $Q_{\pa \Sigma}=0$ the reduction is trivial: the reduced
structures are the same as non-reduced ones. The reduced fiber over
$l=(B,c)$ is
\[
\EL(\Sigma, [l])/Q=\Omega^1(\Sigma)/\{exact\}\oplus \Omega^{n-2}_{closed}(\Sigma)\oplus H^0(\Sigma, \pa \Sigma)[1]\oplus H^{n-1}(\Sigma)[-1]
\]
This space is infinite dimensional for $n>2$ and it is
finite dimensional when $n=2$.

\begin{remark} The reduced space $\cL_{\Sigma}/Q$ and fibers of
$\pi: \EL_{\Sigma}/Q\to \cL_{\Sigma}/Q$ are infinite dimensional when
$n>2$. When $n=2$ they are finite dimensional. This corresponds to
two dimensional electrodynamics which is an almost topological field theory \cite{W-QCD}.
 \end{remark}

\subsection{Yang-Mills theory}\label{bv-ymills}

As in the classical Euclidean electrodynamics spacetime manifolds in the Yang-Mills theory are oriented compact smooth, possibly with boundary (and with corners),  Riemannian manifolds. Let $\g$ be the Lie algebra of a  finite dimensional simply connected Lie group $G$ with $\g$-invariant scalar product.
To simplify notations we assume that $\g$ is a matrix algebra and
that the scalar product is given by the trace: $<a,b>=tr(ab)$.

\subsubsection{The non-reduced theory}
In the first order formulation of Yang-Mills theory fields are
connections $A$ in a principal $G$-bundle $P$ over $N$ and
$(n-2)$-forms $B$ with coefficients in the associated adjoint bundle.
For simplicity, we assume that the principal bundle is trivial and consider connections as $1$-forms with coefficients in $\g$ and $B$ fields as $(n-2)$-forms with coefficients in $\g$.  The ghost fields $c$
are $0$-forms with coefficients in $\g$. The BV extension includes anti-fields $A^\dag, B^\dag, c^\dag$. The total space of BV extended Yang-Mills theory is\footnote{Here, as in the case of classical
electrodynamics we only discuss the minimal BV extension. When $n=4$ the BV extension of Yang-Mills theory can also be presented in a different way using the decomposition of $2$-forms into self-dual and anti-self-dual parts, see \cite{CCFMRTZ} and \cite{Co}. 
}
\begin{eqnarray*}
\FF_N &=& \underbrace{\g\otimes \Omega^1(N)}_{A}\oplus \underbrace{\g\otimes \Omega^{n-2}(N)}_{B}\oplus \underbrace{\g\otimes \Omega^0(N)[1]}_{c}\oplus \\
& & \oplus\underbrace{\g\otimes
\Omega^{n-1}(N)[-1]}_{A^\dag}\oplus \underbrace{\g\otimes \Omega^2(N)[-1]}_{B^\dag}\oplus \underbrace{\g\otimes \Omega^{n}(N)[-2]}_{c^\dag}
\end{eqnarray*}
This is a graded infinite dimensional vector space with the
symplectic form
\begin{eqnarray*}
\omega_N &=& \int_N tr \left(\delta A\wedge \delta A^\dag + \delta B\wedge \delta B^\dag + \delta c\wedge \delta c^\dag\right)
\end{eqnarray*}
The action functional is
\begin{eqnarray*}
S_N &=& \int_N tr\left( B\wedge F_A + \frac{1}{2} B\wedge *B + A^\dag\wedge d_A c + B^\dag\wedge [B,c]+ \frac{1}{2}c^\dag\wedge [c,c]\right)
\end{eqnarray*}
where $F_A=dA+\frac{1}{2}[A, A]$, and the cohomological vector field is
\begin{eqnarray*}
Q_N &=& \int_N tr\left( d_A c\wedge\frac{\delta}{\delta A}  + [B,c]\wedge \frac{\delta}{\delta B}  + \frac{1}{2}[c,c]\wedge\frac{\delta}{\delta c}  +  \left(d_A B + [A^\dag,c]\right)\wedge \frac{\delta}{\delta A^\dag}  +\right.\\
&& \left.+ \left(F_A + *B+ [B^\dag,c]\right)\wedge \frac{\delta}{\delta B^\dag}  +  \left(d_A A^\dag + [B,B^\dag] + [c,c^\dag]\right)\wedge \frac{\delta}{\delta c^\dag} \right)
\end{eqnarray*}

\textbf{The boundary structure.} We will denote the pullback to the boundary of forms $A, B, A^\dag, c$ by the same letters. The space of boundary fields
is the quotient space of the pullback of $\cF_N$ to the boundary
over the kernel of the form $\delta\tilde{\alpha}_{\pa N}$, as
it is explained in section \ref{bvbfv}:
\begin{eqnarray*}
\FF_{\dd N} &=& \underbrace{\g\otimes\Omega^1(\dd N)[1]}_{A,\; \gh=0} \oplus \underbrace{\g\otimes\Omega^{n-2}(\dd N)[n-2]}_{B,\; \gh=0} \oplus \\
&& \oplus\underbrace{\g\otimes\Omega^0(\dd N)[1]}_{c,\; \gh=1} \oplus \underbrace{\g\otimes\Omega^{n-1}(\dd N)[n-2]}_{A^\dag,\; \gh=-1}
\end{eqnarray*}
The structure of an exact symplectic manifold on $\cF_{\pa N}$ is
given by
\begin{eqnarray*}
\alpha_{\dd N} &=&  \int_{\dd N} tr\left(B\wedge \delta A + A^\dag\wedge \delta c\right), \\
\omega_{\dd N} &=&  \int_{\dd N} tr\left(\delta B\wedge \delta A + \delta A^\dag\wedge \delta c\right)
\end{eqnarray*}
The vector field $Q_{\pa N}$ and the action $S_{\pa N}$ for the
boundary BFV theory are
\begin{eqnarray*}
Q_{\dd N} &=& \int_{\dd N} tr\left( d_{A}c\wedge \frac{\delta}{\delta A}  +  [B,c]\wedge\frac{\delta}{\delta B}   +\right.\\
& & \left.+ (d_{A}B+[A^\dag,c])\wedge\frac{\delta}{\delta A^\dag} +  \frac{1}{2}[c,c]\wedge\frac{\delta}{\delta c}\right),\\
S_{\dd N} &=& \int_{\dd N} tr\left( B \wedge d_{A} c + \frac{1}{2}A^\dag\wedge[c,c]\right)
\end{eqnarray*}

\textbf{The codimension $2$ structure.} Let $\Sigma$ be a stratum of codimension 2. The BV-BFV theory on $N$ and on $\pa N$ induce
the following data associated on $\Sigma$ (see section \ref{bvbfv}. The space of fields:

\[
\FF_\Sigma = \underbrace{\g\otimes\Omega^{n-2}(\Sigma)[n-2]}_{B,\; \gh=0} \oplus \underbrace{\g\otimes\Omega^{0}(\Sigma)[1]}_{c,\; \gh=1}, \]
Here we denoted the pullback of $B$ and of $c$ to $\Sigma$ by the same letters. The rest of the BFV data, the exact symplectic form,
the vector field $Q$ and the action $S$ (which can be obtained from $Q$ by the Roytenberg's construction) are:
\begin{eqnarray*}
\alpha_\Sigma &=& \int_\Sigma tr (B\wedge \delta c),\\
\omega_\Sigma &=& \int_\Sigma tr (\delta B\wedge \delta c), \\
Q_\Sigma &=& \int_\Sigma tr\left( [B,c]\wedge \frac{\delta}{\delta B} + \frac{1}{2}[c,c]\wedge\frac{\delta}{\delta c}\right), \\
S_\Sigma &=& \int_\Sigma tr \left(\frac{1}{2} B\wedge [c,c]\right)
\end{eqnarray*}


\subsubsection{The non-reduced $\gh=0$ part}
\textbf{\\The bulk.} The fields $A$ and it Hamiltonian counterpart remain,
so the space of fields is
\[
F_N = \underbrace{\g\otimes\Omega^1(\dd N)}_{A} \oplus \underbrace{\g\otimes \Omega^{n-2}(\dd N)}_{B}
\]
The classical action is
\[
S^{cl}_N = \int_N tr\left(B\wedge F_A + \frac{1}{2} B\wedge *B\right)
\]
Its critical points are solutions to  Euler-Lagrange equations
\[
d_A B = 0, \ \ F_A+*B=0
\]
or, equivalently
\[
d_A * F_A = 0, \ \  B=-*F_A
\]

Infinitesimal gauge transformations are infinitesimal
automorphisms of the trivial $G$-bundle over $N$, i.e.
elements of the Lie algebra $\Map(N, \g)$, They act as
\[
A\mapsto A+d_A\alpha, \ \  B\mapsto B+[B,\alpha]
\]
where $\alpha\in \Omega^0(N,\g)$.

\textbf{The boundary.} Boundary fields in the Yang-Mills theory are pullbacks of fields $A$ and $B$. We will denote them by the same letter:
\[
F_{\dd N} =\underbrace{\g\otimes \Omega^1(\dd N)}_{A} \oplus \underbrace{\g\otimes \Omega^{n-2}(\dd N)}_B
\]
The exact symplectic structure on $\cF_{\pa N}$ is given by
\begin{eqnarray*}
\alpha_{\dd N}^{cl} &=& \int_{\dd N} tr(B \wedge \delta A) \\
\omega_{\dd N}^{cl} &=& \int_{\dd N} tr ( \delta B \wedge \delta A) =\delta \alpha_{\dd N}^{cl}
\end{eqnarray*}
Euler-Lagrange subspace, $gh=0$ part of $\EL_{\pa N}$,
is defined by the constraint
\[
d_{A} B=0
\]
Boundary gauge transformations are:
\[
A\mapsto A+d_{A}\alpha, \ \ B\mapsto B+[B,\alpha]
\]
The action of the Lie algebra of gauge transformations is Hamiltonian with the moment map $\mu: F_{\dd N}  \longrightarrow \g\otimes\Omega^{n-1}(\dd N)$:
\[
(A,B)  \mapsto  d_{A}B
\]

\textbf{Codimension $2$ part.} The $\gh=0$ part of the codimension $2$ structure is
given by the pullback of $B$ to the boundary:
\[
F_\Sigma = \g\otimes \Omega^{n-2}(\Sigma)
\]
with no constraints and the gauge transformations given by
\[
B\mapsto B+[B,\alpha]
\]
where $\alpha\in \Omega^0(\Sigma, \g)$.

\subsubsection{The reduction in the $\gh=0$ part}
The space $EL_N$ is naturally isomorphic to the space of solutions to the Yang-Mills equation $d_A*F(A)=0$.

The classical moduli space is naturally isomorphic to
the space of gauge classes of solutions to the YM equation:
\begin{multline*}EL_N/G_N= \\
=\{(A,B)\; |\; A\in \g\otimes \Omega^1(N),\; d_A *F_A=0,\; B=-*F_A\}/\{(A,B)\sim (A+d_A\alpha, B+[B,\alpha])\} \end{multline*}

The reduced phase space is naturally isomorphic to the cotangent bundle of the space of gauge classes
of all connections:
\begin{multline*}
EL_{\dd N}/G_{\dd N}=\\ \{\{(A,B)\;|\; A\in \g\otimes \Omega^1(\dd N),\; B\in \g\otimes \Omega^{n-2}(\dd N),\; d_{A}B=0\}\}/\{(A,B)\sim (A+d_{A}\alpha, B+[B,\alpha])\}
\end{multline*}

The $gh=0$ part of the $\EL$-moduli space for a codimension 2 stratum
is simply
\[
EL_{\Sigma}/G_\Sigma \cong (\g\otimes \Omega^{n-2}(\Sigma))/G
\]
where $\Sigma$ is $(n-2)$-dimensional and the quotient is by the adjoint action of $G$.

The restriction map
\[
\pi_*: EL_N/G_N \ra EL_{\dd N}/G_{\dd N}
\]
is the pullback of the connection $A$ and of the form $B$
to the boundary. The fibers are finite-dimensional. One way
to see this is by homological perturbation theory around electrodynamics.

The reduced phase space $EL_{\dd N}/G_{\dd N}$ is infinite-dimensional for $n\geq 3$. The case $n=2$ is special as for $n=2$ we have
\[
EL_{\dd N}/G_{\dd N}\cong T^* \MM^G_{\dd N}
\]
where $\MM^G_{\pa N}$ denotes the moduli space of flat $G$-connections on $\pa N$ which is
finite-dimensional.

Under the regularity assumption, the image of $\pi_*$ is Lagrangian.
Smooth loci of $\EL$-moduli spaces for the bulk and the boundary are described by diagram (\ref{relation between BV mod spaces and gh=0 part}) with horizontal arrows being equalities (we assume the group $G$ to be simple for this to hold).

\subsection{Scalar field}\label{sc-field}
\subsubsection{The non-reduced picture}
For the massless free $n$-dimensional scalar field the bulk BV data is:
\begin{eqnarray*}
\F_N &=& \underbrace{\Omega^0(N)}_{\phi} \oplus \underbrace{\Omega^{n-1}(N)}_{p}
\oplus \underbrace{\Omega^n(N)[-1]}_{\phi^\dag} \oplus \underbrace{\Omega^1(N)[-1]}_{p^\dag} \\
\omega_N &=& \int_N (\delta \phi \wedge \delta \phi^\dag + \delta p \wedge \delta p^\dag) \\
S_N &=& \int_N (p\wedge d\phi +\frac{1}{2} p\wedge *p )\\
Q_N &=& \int_N \left(dp \wedge \frac{\delta}{\delta \phi^\dag} + (d\phi+ *p)\wedge \frac{\delta}{\delta p^\dag}\right)
\end{eqnarray*}
The boundary BFV data in this theory is:
\begin{eqnarray*}
\F_{\pa N} &=& \underbrace{\Omega^0(\pa N)}_{\phi} \oplus \underbrace{\Omega^{n-1}(\pa N)}_{p} \\
\alpha_{\pa N} &=& \int_{\pa N} p\wedge \delta\phi \\
\omega_{\pa N} &=& \delta\alpha_{\pa N}= \int_{\pa N} \delta p\wedge \delta \phi \\
Q_{\pa N} &=& 0 \\
S_{\pa N} &=& 0
\end{eqnarray*}
It is easy to derive them from the BV data in the bulk as it is explained in section \ref{bvbfv}. It is clear that
the theory has length one.

For the Euler-Lagrange space we have
\be \EL_N=\{(\phi,p,\phi^\dag,p^\dag)\;| \; dp=0,\; p=-*d\phi\} \ee
At the boundary $\EL_{\pa N}=\cF_{\pa N}$.

\subsubsection{The reduction}


Similarly to the case of electrodynamics we have,
\[
\ker(\hat{Q})/Im(\hat{Q})\simeq \Omega^0_{Harm}(N) \oplus H^n(N)[-1]
\]
where $\Omega^0_{Harm}(N)$ are harmonic zero-forms on $N$;
the first summand correspond to the field $\phi$ and the second to $\phi^\dag$.
Because of the linearity, the moduli space $\EL_N/Q_N$ is given by the same formula.

\begin{remark} If $\pa N=\emptyset$ then $\EL_N/Q_N \simeq H^0(N)\oplus H^n(N)[-1]$ is a finite-dimensional odd-symplectic space.
\end{remark}

Because $Q_{\pa N}=0$, the reduced boundary Euler-Lagrange space is the space of boundary fields:
\be  \EL_{\pa N}/Q_{\pa N} =\F_{\pa N} = \Omega^0(\pa N)\oplus \Omega^{n-1}(\pa N)\ee
The reduced evolution relation is the same as the non-reduced one:
\begin{multline} \underline{\LL}_N= \\
=\{(\phi, p)\in \F_{\pa N}\; | \; p=-(*d \tilde{\phi})|_{\pa N} \;\mbox{ for }\; \tilde{\phi}\;\mbox{a harmonic continuation of}\; \phi \;\mbox{into}\; N\}\simeq\\
\simeq \Omega^0(\pa N) \end{multline}
Because the harmonic continuation $\tilde{\phi}$ exists and is unique, the subspace $\underline{\LL}_N$ is Lagrangian.

The restriction (the pullback to the boundary) mapping is surjective over $\underline{\LL}_N$
\be
\begin{CD}
\EL_N/Q_N  \\
@VVV \\
\underline{\cL}_{ N}\subset \EL_{\pa N}/Q_{\pa N}
\end{CD}
\ee
Its fibers are  finite-dimensional odd-symplectic spaces
canonically isomorphic to
\be  H^0(N,\pa N) \oplus H^n(N)[-1] \simeq T^*[-1]\mathbb{R}^k \ee
Where $k\geq 0$ is the number of closed connected components of $N$.

\subsubsection{Massive scalar field} The BV extension of the massive scalar field is similar to the massless scalar filed.
The space of fields and the symplectic structure are the same.
The action and the vector field $Q_N$ are
\begin{eqnarray}
S_N &=& \int_N \left(p\wedge d\phi + \frac{1}{2} p\wedge *p - \frac{m^2}{2} \phi\wedge *\phi\right) \\
Q_N &=& \int_N \left( (dp - m^2 *\phi)\wedge \frac{\pa}{\pa \phi^\dag}+ (d\phi+ *p)\wedge \frac{\pa}{\pa p^\dag} \right)
\end{eqnarray}
The boundary data are unchanged. The reduced bulk Euler-Lagrange
space is
\be
\EL_N/Q_N \simeq  \Omega^0_{Klein-Gordon}(N)
\ee
Here $\Omega_{Klein-Gordon}^0(N)$ is the space of functions satisfying the equation $\Delta\phi-m^2\phi=0$.
It fibers over $\underline{\LL}\simeq \Omega^0(\pa N)$ with zero fibers. That is, when $m\neq 0$ the fibers of the $\EL$-moduli spaces over given boundary values are trivial.

It is also easy to construct BV extensions of scalar fields interacting with the Yang-Mills theory, and to generalize it to Dirac and Majorana fermions.

\subsection{Abelian $BF$ theory}\label{BFBV}

\subsubsection{The BV-BFV structure of $BF$ theory}
Here we will focus on the BV-BFV extension of the
classical abelian $BF$ gauge theory. Let us start with the
description of the corresponding BV theory.

The space of fields in this theory is:
\[
\cF_N=\Om^\bullet(N)[1]\oplus \Om^\bullet(N)[n-2]
\]
For coordinate fields we will write $\cA\in \Om^\bullet(N)[1]$
and $\cB\in \Om^\bullet(N)[n-2]$. Here $\Omega^\bullet(N)$ is regarded as a
$\ZZ$-graded vector space with $\Omega^k(N)$ being its component of degree $-k$.

\begin{remark} The BV-$BF$ theory is an example of the AKSZ theory ( cf. section \ref{sec-aksz}) with the target manifold $T^*[1]\RR$ and with $\Theta=0$.
\end{remark}
\begin{remark} Dimensions $n=2,3$ are special. When $n=2$ the
theory is equivalent to the topological sector of electrodynamics.
When $n=3$ it is equivalent to two copies of abelian Chern-Simons theories.
\end{remark}

The BV symplectic form
\[
\omega_N=
\int_N \delta\cA\wedge \delta \cB
\]
The total degree (ghost number) is $gh(\omega_N)=-1$.

The action functional is:
\[
S_N=\int_N \cB\wedge d\cA
\]

The vector field $Q_N$ is
\[
Q_N=\int_N\left(d\cA\wedge\frac{\delta}{\delta \cA}+d\cB\wedge \frac{\delta}{\delta \cB}\right)
\]
On ``coordinate functions" it acts as
\begin{equation}\label{bfbvgauge}
Q_N\cA= d\cA , \ \ Q_N\cB= d\cB
\end{equation}

The space $\EL_N$ of solutions to Euler-Lagrange equations, i.e. the space of zeroes of the vector
field $Q_N$, is the space of closed forms:
\[
\EL_N=\Omega^\bullet_{closed}(N)[1]\oplus \Omega^\bullet_{closed}(N)[n-2]
\]
Clearly this space is coisotropic in $\F_N$.

The space of boundary fields is:
\[
\cF_{\pa N}=\Om^\bullet (\pa N)[1]\oplus \Om^\bullet(\pa N)[n-2]
\]
and the natural restriction map $\pi:\cF_{ N}\to \cF_{\pa N}$
is the pullback of forms to the boundary.

The space of boundary fields is a symplectic space
with the exact symplectic form of degree $0$:
\[
\omega_{\pa N}=\delta\alpha_{\pa N}, \ \ \alpha_{\pa N}=\int_{\pa N}  \cA\wedge \delta\cB
\]
The boundary cohomological vector field $Q_{\pa N}=\delta\pi Q_N$ is
\[
Q_{\pa N}\cA=d\cA, \ \ Q_{\pa N}\cB=d\cB
\]

The Euler-Lagrange subspace $\EL_{\pa N}\subset \cF_{\pa N}$ is the subspace of closed forms on $\pa N$.

\begin{Proposition} The data described above is an example
of the BV-BFV theory.
\end{Proposition}

The proof is a straightforward computation which proves
identities from section \ref{bvbfv}.

The evolution relation $\cL_N\subset \EL_{\pa N}$ consists of closed
forms on $\pa N$ which extend to closed forms on $N$.

\begin{Proposition} \label{LagrLN}The subspace
$\cL_N\subset \EL_{\pa N}\subset\cF_{\pa N}$
is Lagrangian.
\end{Proposition}

To prove this Proposition we need the following lemma.

\begin{lemma}\label{lnlagr} Denote  by $\iota_*: H_\bullet(\pa N)\to H_\bullet(N)$ the pushforward by $\iota: \pa N \hookrightarrow N$  in homology. The subspace $ker(\iota_*)$ is isotropic in
$H_\bullet(\pa N)$ with respect to
the bilinear form given by the intersection pairing.
\end{lemma}
\begin{proof} Let $U$ and $V$ be two cycles in $N$ relative to the boundary, with $dim(U)+dim(V)=dim(N)+1$,
and let $u=\pa U $, $v=\pa V$ be their boundaries in $\pa N$. By general position argument we can assume
that $U\cap V$ is a one dimensional chain. This chain gives a cobordism
between $u\cap v$ and the empty set. Hence the sum of coefficients in $u\cap v$, which is the
intersection pairing, vanishes.
\end{proof}

\begin{proof}[Proof of Proposition \ref{LagrLN}]

By the natural identification of $H_\bullet(\pa N)^*$
with $H^\bullet(\pa N)$, the annihilator of $\ker(\iota_*)$ is  $Im(\iota^*)$.
Therefore, by Lemma \ref{lnlagr},  the subspace $Im(\iota^*)$ is coisotropic in
$H^\bullet(\pa N)$.


Because $\cL_N\subset \EL_N$ is isotropic by Proposition \ref{sub-loc-isot}, the reduced space
$\underline{\cL}_N$ is also isotropic. Since we just proved that
it is coisotropic, we conclude that it is Lagrangian. By Proposition \ref{lagr} the preimage of a Lagrangian subspace
with respect to the symplectic reduction is Lagrangian if it contains
the kernel of the presymplectic form. In our case this kernel consists
of exact forms on $\pa N$, and exact forms are clearly in $\cL_N$.


\end{proof}

\begin{Proposition} \label{prop: ab BF is regular}
Abelian $BF$ theory is regular.
\end{Proposition}

\begin{proof}
We have
\begin{eqnarray}
\label{ab BF ker Q}\ker (\hat{Q}_X) &=& \Omega^\bt_{closed}(N)[1] \oplus \Omega^\bt_{closed}(N)[n-2], \\
\label{ab BF im Q}Im (\hat{Q}_X) &=& \Omega^\bt_{exact}(N)[1] \oplus \Omega^\bt_{exact}(N)[n-2], \\
\label{ab BF ker Q^ver}\ker (\hat{Q}_X^\ver) &=& \Omega^\bt_{closed}(N,\pa N)[1] \oplus \Omega^\bt_{closed}(N,\pa N)[n-2], \\
\label{ab BF im Q^ver}Im (\hat{Q}_X^\ver) &=& d(\Omega^\bt(N,\pa N))[0] \oplus d(\Omega^\bt(N,\pa N))[n-3], \\
\label{ab BF ker Q bdry}\ker (\hat{Q}_l) &=& \Omega^\bt_{closed}(\pa N)[1]\oplus \Omega^\bt_{closed}(\pa N)[n-2], \\
\label{ab BF im Q bdry}Im (\hat{Q}_l) &=& \Omega^\bt_{exact}(\pa N)[1] \oplus \Omega^\bt_{exact}(\pa N)[n-1]
\end{eqnarray}
Due to Hodge-Morrey decomposition for forms on $N$, pairs of subspaces (\ref{ab BF ker Q}), (\ref{ab BF im Q^ver}) and (\ref{ab BF ker Q^ver}), (\ref{ab BF im Q}) are mutually orthogonal. Due to Hodge decomposition on $\pa N$, subspaces (\ref{ab BF ker Q bdry}), (\ref{ab BF im Q bdry}) are also mutually orthogonal. Thus the theory is regular.
\end{proof}

\subsubsection{The reduction of boundary structures}

From section \ref{bv-sym-red} we know that the $Q$-reduced Euler-Lagrange space
coincides with its symplectic reduction and is
\[
\EL_{\pa N}/Q=\underline{\EL}_{\pa N}= H^\bullet (\pa N)[1]\oplus H^\bullet(\pa N)[n-2]
\]
The reduced space $\cL_N/Q$ is the space of cohomology classes of
closed forms on $\pa N$ which continue to closed forms on $N$, and we already proved
that it is Lagrangian.

\subsubsection{The reduction of of bulk fields}
\label{sec: ab BF red fibers}

The $Q$-reduced space of solutions to the EL equations is:
\[
\EL_N/Q=
H^\bullet( N)[1]\oplus H^\bullet(N)[n-2]
\]
To compute the space $\EL(N, [l])/Q$, for $l\in \cL_N$ consider the natural exact sequence:
\[
0\to \Om^\bullet(N,\pa N)\to \Om^\bullet(N) \stackrel{\iota^*}{\to} \Om^\bullet(\pa N)\to 0
\]
where $\iota^*$ is the pullback to the boundary corresponding to the
inclusion mapping $\iota:\pa N\hookrightarrow N$, and $\Omega^\bullet(N, \pa N)$ are
forms vanishing on the boundary (their pullback to the boundary is zero). It induces
the standard long exact sequence for cohomology spaces:
\begin{eqnarray*}
\dots & \to & H^\bullet(N,\pa N)\stackrel{\chi}{\to} H^\bullet(N) \stackrel{\iota^*}{\to} H^\bullet(\pa N)\stackrel{\beta}{\to} \\
& \to & H^{\bullet+1}(N,\pa N)\stackrel{\chi}{\to} H^{\bullet+1}(N) \stackrel{\iota^*}{\to} H^{\bullet+1}(\pa N)\stackrel{\beta}{\to} \dots
\end{eqnarray*}

For $l\in \cL_N$ the $Q$-reduction of the space $\EL(N, [l]))=\pi^{-1}([l])\cap \EL_N$ is the $Q$-reduced space
$\EL(N, [l])/Q=(\pi^{-1}([l])\cap \EL_N)/Q$
\begin{multline}\label{elqbf}
\EL(N, [l])/Q=Im(\chi)[1]\oplus Im(\chi)[n-2]\simeq \\
\simeq H^\bullet (N, \pa N)/Im(\beta)[1]\oplus H^\bullet(N, \pa N)/Im(\beta)[n-2]
\end{multline}
Recall that $[l]$ is the leaf of $Vect_Q$ through $l$.
Due to regularity, (\ref{elqbf})
is symplectic, with the symplectic structure coming from the Lefschetz duality between $H^\bullet(N)$ and $H^\bullet(N,\pa N)$.

In terms of the long exact sequence of the pair $(N, \pa N)$, the reduced space $\cL_N$ is:
\[
\cL_N/Q=Im(\iota^*)[1]\oplus Im(\iota^*)[n-2]\quad  \subset\; H^\bullet (\pa N)[1]\oplus H^\bullet(\pa N)[n-2]
\]

\begin{remark} The symplectic reduction of $\EL_N$ is infinite dimensional and, as a vector space,
\[
\underline{\EL}_N=H^\bullet(N, \pa N)[1]\oplus H^\bullet(N,\pa N)[n-2]\oplus
\Omega^\bullet_{closed}(\pa N)[1]\oplus
\Omega^\bullet_{closed}(\pa N)[n-2]
\]
Indeed, it is easy to check that the symplectic orthogonal  subspace
$\EL_N^{\perp, symp}$ to $\EL_N$ is
\[
\Omega_{exact}(N,\pa N)[1]\oplus \Omega_{exact}(N,\pa N)[n-2]
\]
where $\Omega_{exact}(N,\pa N)$ is the space of exact
forms with the pullback to the
boundary being zero. The  symplectic reduction is the quotient space
$\EL_N/\EL_N^{\perp, symp}$. It is the symplectic $\EL$-moduli space
discussed in section \ref{part-red-sec} and comes with residual gauge symmetry
data, which in particular allow for a simple gluing formula.
\end{remark}

\subsubsection{BV extensions to strata of codimension $k$}
The $BF$ theory can be maximally extended.
On an $n-k$ dimensional stratum $N_k$ the space of fields
is
\[
\F_{N_k}=\Omega^\bullet(N_k)[1]\oplus\Omega^\bullet(N_k)[n-2]
\]
The symplectic form is
\[
\omega_{N_k}=\int_{N_k}\delta\cA\wedge \delta\cB
\]

The cohomological vector field $Q_{N_k}$ is
\[
Q_{N_k}=\int_{N_k} \left(d\cA\wedge \frac{\delta}{\delta \cA}+
d\cB\wedge \frac{\delta}{\delta \cB}\right)
\]

The action is
\[
S_{N_k}=\int_{N_k} \cB\wedge d\cA
\]

Formulae for $\alpha, Q, S$ are structurally the same for all $k$.
This is a general feature of AKSZ theories of which abelian $BF$ theory is an
example.

\subsubsection{The $gh=0$ part of the abelian $BF$ theory}

In this section we will consider the restriction of
the abelian $BF$ theory to the fields with $gh=0$.

The $gh=0$ part of the space of fields is $F_N=\Omega^1(N)\oplus \Om^{n-2}(N)$.
The action is the restriction of the BV action
to the $gh=0$ sector:
\[
S_N=\int_N B\wedge dA
\]
where $A\in \Om^1(N)$ and $B\in \Om^{n-2}(N)$.

Euler-Lagrange equations are:
\[
dA=0, \ \ dB=0,
\]
Solutions are closed forms
\[
EL_N=\Omega^1_{closed}(N)\oplus \Omega^{n-2}_{closed}(N)
\]
This is the degree zero part of the space $\EL_N$ in the
$BF$-BV theory.

The gauge group $G_N=\Om^0(N)\oplus \Om^{n-3}(N)$ acts on fields as
\[
A\mapsto A+d\alpha, \ \ B\mapsto B+d\beta
\]
Vector fields generated by these transformations are degree zero parts of vector fields $Vect_Q$ . The action of the abelian $BF$ theory is invariant with respect to these transformations.

The set of gauge classes of solutions is the degree zero part of the $Q$-reduced Euler-Lagrange space:
\[
EL_N/G_N=H^1(N)\oplus H^{n-2}(N)
\]

\textbf{Non-reduced boundary theory.} The boundary fields are:
\[
F_{\pa N}=\Omega^1(\pa N)\oplus \Om^{n-2}(\pa N)
\]

This is a symplectic manifold with the symplectic
form induced by the intersection pairing:
\[
\omega_{\pa N}=\int_N \delta A\wedge \delta B
\]
This form is restriction of the symplectic form
in the $BF$ theory to the $gh=0$ subspace in the space of fields.

The boundary Euler-Lagrange subspace consists of closed $1$- and $(n-2)$-forms on $\pa N$:
\[
EL_{\pa N}=\Omega^1_{closed}(\pa N)\oplus \Om^{n-2}_{closed}(\pa N)
\]
It is a coisotropic subspace of $F_{\pa N}$.

The evolution relation $L_N=\pi(EL_N)$ is the space of closed $1$- and $(n-2)$-forms on $\pa N$ which continue to closed forms on $N$. As follows from Lemma \ref{lnlagr}, $L_N$ is a
Lagrangian subspace in $EL_{\pa N}$.

\textbf{Reduced boundary theory.}
Boundary gauge group $G_{\pa N}=\Om^0(\pa N)\oplus \Om^{n-3}(\pa N)$ is the pullback of the group of gauge transformations on $N$ to the boundary.

The symplectically reduced Euler-Lagrange subspace $\underline{EL}_{\pa N}$ is
the set of leaves of the characteristic foliation of
$EL_{\pa N}$ (the foliation of $EL_{\pa N}$ by
Hamiltonian vector fields generated by the ideal $I_{EL_{\pa N}}$  of functions on $F_{\pa N}$ vanishing on $EL_{\pa N}$.

\begin{remark} Because the action of the gauge group is Hamiltonian
with
\[
H_\alpha=\int_{\pa N} \alpha\wedge dB, \ \ H_\beta=\int_{\pa N}\beta\wedge dA
\]
The reduced space $\underline{EL}_{\pa N}$ is also the result of Hamiltonian reduction: $\overline{EL}_{\pa N}=J^{-1}(0)/G_{\pa N}$, where $J: F_{\pa N}\to [\Om^0(\pa N)\oplus \Om^{n-3}(\pa N)]^*$ is the moment map
$(A, B)\mapsto (H_\alpha, H_\beta)$. It is also clear that $EL_{\pa N}=J^{-1}(0)$.
\end{remark}

The $Q$-reduced Euler-Lagrange space
\[
EL_{\pa N}/G_{\pa N}=H^1(\pa N)\oplus H^{n-2}(\pa N)
\]
has the natural symplectic structure given by the Poincar\'e pairing.

The reduced evolution relation $L_N/G_{\pa N}$ consists of
cohomology classes of closed forms on $\pa N$ which continue to closed forms on $N$. In other words, this is the image of the restriction mapping:
\[
H^1(N)\oplus H^{n-2}(N)\stackrel{\iota^*\oplus \iota^*}{\to}H^1(\pa N)\oplus H^{n-2}(\pa N)
\]
where $\iota: \pa N\hookrightarrow N$ is the inclusion of the boundary
mapping. As follows from the Proposition \ref{LagrLN},
it is a Lagrangian subspace.

The kernel of the restriction map is isomorphic to

\[
H^1(N)/Im (\beta_0)\oplus H^{n-2}( N)/Im (\beta_{n-3})
\]

Here $\beta_i: H^i(\pa N)\to H^{i+1}(N, \pa N)$ is the mapping in the long exact sequence of the pair $(N, \pa N)$.

\section{The AKSZ construction of classical topological gauge theories}\label{sec-aksz}

In this section we will recall the construction of classical topological field theories for closed manifolds known
as the AKSZ construction  \cite{AKSZ} and we will extend this construction to manifolds with boundary. We will show that
this extension gives an example of a BV theory for spacetimes with boundary. The AKSZ construction generalizes the BV extension of
the $BF$ gauge theory.

\subsection{The target manifold}
\subsubsection{Hamiltonian dg manifolds} The AKSZ construction  requires the choice of a Hamiltonian
differential graded manifold as the target space.

Recall that a {\it differential graded (dg) manifold } is a pair $(\M, Q)$
where $\M$ is a graded manifold and $Q$ is a cohomological vector field of
degree one. A vector field is cohomological if its Lie derivative
squares to zero.

A dg manifold is a {\it dg symplectic manifold} of degree $m$ if it has a symplectic form $\omega$ of
degree $m$ which is $Q$-invariant (i.e. $L_Q\omega=0$ where
$L_Q$ is the Lie derivative with respect to $Q$).  We denote the degree by $deg$.

\begin{definition} A  dg symplectic manifold $(\M, \omega, Q)$
of degree $m$ is Hamiltonian if there exists
an element $\Theta\in Fun(\M)$ with $deg(\Theta)=m+1$ such that

\begin{equation}\label{cme-theta}
\{\Theta, \Theta\}=0
\end{equation}
and $Q$ is the Hamiltonian vector field of $\Theta$.
\end{definition}

\begin{remark} A graded symplectic manifold is always
exact when $deg(\omega)\neq 0$: $\omega=d\alpha$.
If $deg(\omega)=0$ we will require it to be exact. A dg symplectic manifold is automatically
Hamiltonian when $deg(\omega)\neq -1,-2$. See section \ref{o-degr} and  \cite{Royt}.
\end{remark}

Notice that a Hamiltonian dg manifold is actually defined by the symplectic
form and by the function $\Theta$ satisfying (\ref{cme-theta}). The vector
field $Q$ is defined as the Hamiltonian vector field of $\Theta$ and
acts of functions as
on functions on $\M$ as:
\[
Qf=\{\Theta, f\}
\]
In local coordinates
\[
Qf=\sum_{ab} \Theta\frac{\overleftarrow \pa }{\overleftarrow{\pa} x^a}\omega^{ab}
\frac{\overrightarrow{\pa}}{\overrightarrow{\pa} x^b} f
\]
For polynomial functions $f$ we have
\[
\frac{d f}{dt}(x+t\epsilon)=\sum_a \epsilon^a \frac{\overrightarrow{\pa}}{\overrightarrow{\pa} x^a} f=
\sum_a f \frac{\overleftarrow{\pa} }{\overleftarrow{\pa} x^a}
\epsilon^a
\]
were $deg(t)=0$.

The condition (\ref{cme-theta}) implies
\[
Q^2=0
\]
Because  $deg(\omega)=m$, we have $deg(\omega^{-1})=-m$ and
because $deg(\Theta)=m+1$, we have $deg(Q)=1$.

\subsubsection{Examples of Hamiltonian dg manifolds}\label{aksz-ex}

Here we will give few examples of Hamiltonian dg manifolds.

{\bf Example 1}. Let $m=2$, $\g$ be a finite dimensional Lie algebra with an invariant inner product,
and $x^a$ be coordinates on $\g$ in an orthonormal linear basis $e_a$\footnote{ $\g$ can be also a Lie superalgebra} . Choose $\M=\g[1]$ with $deg(x^a)=1$ and define
\[
\omega=\frac{1}{2}\sum_a d x^a\wedge d x^a
\]
and
\[
\Theta=\frac{1}{6} \sum_{abc} f_{abc}x^ax^bx^c
\]
where $f_{abc}$ are structure constants of $\g$ in the basis $e_a$. Clearly $deg(\omega)=2$, and $deg(\Theta)=3$, which agrees with $m=2$.

The differential $Q$ in this example is given by
\[
Q=\frac{1}{2} \sum_{abc} f_{ab}^cx^ax^b\frac{\pa}{\pa x^c}
\]
which corresponds to the Chevalley-Eilenberg differential
for the Lie algebra $\g$.

This example describes the target space for the  Chern-Simons theory in the BV formalism.

{\bf Example 2} Let $n$ be any integer and $\g$ be
a finite dimensional Lie algebra. Define
\[
\M=\g[1]\oplus \g^*[n-2]=T^*[n-1](\g[1])
\]
Let $x^a$ be coordinates on $\g$, and $p_a$ be corresponding
coordinates on the dual space $\g^*$, $gh(x^a)=1, \ \ gh(p_a)=n-2$. Define
\[
\omega=\sum_a d p_a \wedge d x^a, \ \  \alpha=\sum_a p_adx^a
\]
\[
\Theta=\frac{1}{2} \sum_{abc} f^a_{bc} p_ax^bx^c
\]
It is is clear that $deg(\omega)=n-1$ and $deg(\Theta)=n$.

This is the target space for $BF$ models. When $n=2$ this space
is the same as in the next example (Poisson sigma model) for $M=g^*$
with the Kirillov-Kostant Poisson structure.

{\bf Example 3} Let $m=1$ and $\M=T^*[1]M$, where $M$ is a
finite dimensional Poisson
manifold with the Poisson tensor $\pi$. Let $x^i$ be local coordinates on $M$, and $p_i$ be corresponding coordinates
on the cotangent space $T^*_xM$. The grading is such that
$deg(x^i)=0$ and $deg(p_i)=1$

Define
\[
\omega=\sum_i d p_i \wedge d x^i, \ \ \alpha=\sum_i p_i dx^i
\]
and
\[
\Theta=\frac{1}{2} \sum_{ij}\pi^{ij}(x)p_ip_j
\]
Condition (\ref{cme-theta}) is equivalent to the fact that $\pi$
is Poisson.
Clearly $deg(\omega)=1$, and $deg(\Theta)=2$ which
agrees with $m=1$.

This example describes the target space in the Poisson
sigma model.

{\bf Example 4} When $m=2$ and the Lie algebra has an invariant bilinear form,
the Example 2 can be modified
by adding a cubic term in $p$ to $\Theta$. For this we identify $\g$ and
$\g^*$ using the Killing form and assume $x^a$, $p_a$ are coordinates in an orthonormal basis.
The new potential is
\[
\Theta= \frac{1}{2}\sum_{abc} f^a_{bc} p_ax^bx^c\pm\frac{1}{6}f^{abc}p_ap_bp_c
\]

This Hamiltonian dg manifold is isomorphic to two copies of
the Hamiltonian dg manifold from the first example.

{\bf Example 5} When $m=3$ and $\g$ has an invariant bilinear form, we can add a quadratic
term to $\Theta$. The new potential is
\[
\Theta=\frac{1}{2}\sum_{abc} f^a_{bc} p_ax^bx^c+\frac{1}{2}\sum_ap_a^2
\]

\subsection{The space of fields}Fix a Hamiltonian dg manifold $(\M, \omega,
\Theta)$ of degree $n-1$. The classical $n$-dimensional AKSZ field theory
with the target manifold $\M$ on the spacetime manifold $N$
has  $\F_N=\Map(T[1]N, \M)$  as the space of fields\footnote{The AKSZ construction does not have to have $T[1]N$ as the source graded manifold. However we will consider only these cases, as
they seem to be more important in field theory}. We assume that
the spacetime is a compact oriented smooth manifold.
See Appendix \ref{graded} for basic facts on graded manifolds and their mapping spaces.
As in the previous sections we will say that the grading on the space of fields
is given by ghost numbers $gh$.

Let $f$ be a smooth function of fixed degree on $\M$
and $X\in \F_N$,
then the composition $X_f=f\circ X$ is a
smooth function on $T[1]N$. Let $\xi^i$ be local coordinates on
the fiber $T_u[1]N$, $u\in N$. Then $X_f$ at $(\xi,u)$ can be written
as
\[
X_f(\xi, u)=\sum_{k=0}^n X_f(u)_{i_1,\dots, i_k} \xi^{i_1}\dots\xi^{i_k}
\]
with $gh(X_f(u)_{i_1,\dots, i_k})=deg(f)-k$.

We have natural identification
$\phi: C^\infty(T[1]N)\simeq \Omega(N)$. In coordinates
this mapping brings the function $X_f$ to the form
\[
\phi(X_f)( u)=\sum_{k=0}^n X_f(u)_{i_1,\dots, i_k} du^{i_1}\wedge \dots \wedge du^{i_k}
\]

Let $x^a$, $a=1,\dots n$ be homogeneous local coordinates on $\M$.  Denote by $X^a$
the composition of the field $X\in \cF_N$ and the
coordinate function $x^a$.  Component fields $X^a$ can be regarded
as forms
$\phi(X^a)(u)$ on $N$, $a=1,\dots, n$:

\begin{equation}\label{fields}
\phi(X^a)(u)=\sum_{k=0}^d X^a(u)_{i_1,\dots, i_k} du^{i_1}\wedge\dots\wedge  du^{i_k}
\end{equation}
We will denote the component of degree $k$ by $X^a_k(u)$.
We will also call forms $X^a(u)$ coordinate fields, or superfields, when it will
not cause a confusion.

As it follows from the definition of $X^a(u)_{i_1,\dots, i_k}$
the ghost number of this field is  $gh(X^a_k)=deg(x^a)-k$.

\subsection{The non-reduced AKSZ theory for spacetime manifolds with boundary}

\subsubsection{The AKSZ theory in terms of coordinate fields}

The AKSZ action consists of ``kinetic" and ``interaction" parts:
\[
S_N[X]=S^{kin}_N[X]+S^{int}_N[X]
\]

Let $x^a$ be local coordinates on $\M$ and $\alpha(x)=\sum_a \alpha_a(x)dx^a$. In local coordinates the ``kinetic" part of the
AKSZ action is
\[
S^{kin}_N[X]=\int_N \sum_a \alpha_a(X(u))\wedge dX^a(u)
\]

Here  and below $X^a(u)$
are coordinate components of fields as in (\ref{fields}).
It is easy to see that $gh(S^{kin}_N)=0$.

The ``interaction" part of the AKSZ action is the functional
\[
S^{int}_N[X]=\int_N\Theta(X)
\]
where $\Theta$ is the potential function for the target manifold.
Because $deg(\Theta)=n$, the action has the
zero grading, i.e. $gh(S^{int}_N)=0$. Here and below
the expression $\int_N \Theta(X)$ for a homogeneous polynomial $\Theta(x)=
\sum_{\{a\}} \Theta_{\{a\}} x^{a_1}\dots x^{a_k}$ of degree
$k$
on $\M$ means
$\sum_{\{a\}} \Theta_{\{a\}} \int_N X^{a_1}\wedge\dots\wedge X^{a_k}$.

Assume that in local coordinates $\omega_\M=\frac{1}{2}\sum_{ab} \omega_{ab} dx^a\wedge dx^b$, then
\[
\omega_N=\frac{1}{2}\int_N \sum_{ab} \omega_{ab}(X(u))\wedge
\delta X^a(u)\wedge \delta X^b(u)
\]
where $\delta X^a(x)$ is the de Rham differential on the space of fields.

In local coordinates, the action of $Q_N$ on local functionals is:
\[
Q_{N}F=\int_N\left(dX^a+\omega^{ab}(X)\wedge\frac{\pa \Theta}{\pa x^b}(X)\right)\wedge \frac{\delta F}{\delta X^a}
\]
Here $\omega^{ab}$ are components of the Poisson bivector
field corresponding to the symplectic form $\omega$.
Because $deg(\omega)=n-1$ and $deg(\Theta)=n$, we have $gh(Q_N)=1$.

It acts on coordinate fields as:
\begin{equation}\label{qvf-coord}
Q_NX^a=dX^a+\omega^{ab}(X)\wedge\frac{\pa \Theta}{\pa x^b}(X)
\end{equation}

The variation of $S_{N}$:
\[
\delta S_{N}[X]= \int_N( \delta X^a(u)\wedge \pa_a\alpha_b(X(u))\wedge dX^b(u)
+\alpha_a(X(u))\wedge\delta dX^a(u)+ \delta \Theta(X))
\]
After integration by parts this expression becomes
\[
\int_{\pa N} \alpha_a(X(u))\wedge \delta X^a(u)+ \int_N\left(\omega_{ab}(X(u))\wedge dX^a(u)+\frac{\pa \Theta}{\pa X^b}(X(u))\right)\wedge \delta X^b(u)
\]
In this section we assume that $\pa N=\emptyset $, which means
the first term is absent. The Euler-Lagrange equations
are:
\begin{equation}\label{BV-EL}
dX^a(u)+ \sum_b\omega^{ab}(X(u))\wedge\frac{\pa \Theta}{\pa x^a}(X(u))=0
\end{equation}

The same formulae hold for the boundary BFV action $S_{\pa N}$, for
the boundary symplectic form $\omega_{\pa N}$, and for the boundary cohomological vector field $Q_{\pa N}$.
One should simply substitute $\pa N$ instead of $N$. The only difference is the degree change.
Because $dim(\pa N)=n-1$, $gh(S_{\pa N})=1, \ \ gh(\omega_{\pa N})=-1$ and $gh(Q_{\pa N}=-1$ as
in the bulk $N$.

The BV and BFV data described above have the right grading
and give an example of a BV-BFV theory:

\begin{Proposition}\label{aksz=bvbfv} The cohomological vector field $Q_N$ is Hamiltonian, up to a boundary term:
\be
\iota_{Q_N} \omega_N = (-1)^{\dim (N)} \delta S_N +\pi^* \alpha_{\dd N} \label{Q_N Ham}
\ee
\end{Proposition}

We defer the proof to Appendix \ref{Cart-Calc}.

\begin{remark} When $N$ has a non-empty boundary, the AKSZ action depends of the choice of the form $\alpha$, not only on its
cohomology class as in case of closed manifolds. Let $\tilde{S_N}$ be the
action corresponding to $\tilde{\alpha}=\alpha+df$, then
\[
\tilde{S}_N[X]=S_N[X]+\int_N d f(X)=S_N+\int_{\pa N} f(X)
\]

\end{remark}

Any AKSZ theory can be maximally extended.
On an $n-k$ dimensional stratum $\Sigma$ the space of fields is the space of maps $\F_\Sigma=\Map( T[1]\Sigma, \M)$.
The symplectic form $\omega_\Sigma$, the form $\alpha_\Sigma$, the action functional $S_\Sigma$ and the  vector field $Q_\Sigma$ are all given by the same formulae as above.
The difference from the $n$-dimensional stratum is only in
the grading:
\[
gh(\omega_\Sigma)=k-1, \ \ gh(Q_\Sigma)=1, \ \ gh(S_\Sigma)=k
\]

\section{Examples of AKSZ theories}\label{examples-aksz}

\subsection{Abelian Chern-Simons theory}\label{ab-cs}
The target space for this AKSZ theory is $\RR[1]$ with
symplectic structure $\omega=da\wedge da$, where $a$ is the coordinate
on $\RR[1]$ and $\Theta=0$.

\subsubsection{The bulk BV theory} The space of fields on
the $3$-dimensional spacetime manifold $N$ is :
\[
\cF_N=\Omega^\bullet(N)[1]
\]
The fields corresponding to forms of degree $0,1,2,3$ will be denoted by $c, A, A^\dag, c^\dag$ respectively. The ghost numbers are $1, 0, -1, -2$. We will write $\cA=c+A+A^\dag+c^\dag$ for the
BV superfield.

The symplectic form,  the vector field
$Q$, and the classical action are:
\[
\omega_N=\frac{1}{2}\int_N \delta\cA\wedge \delta\cA=\int_N(\delta c\wedge \delta c^\dag+\delta A\wedge \delta A^\dag)
\]
\[
Q_N=\int_Nd\cA\wedge \frac{\delta}{\delta \cA}=\int_N(dc\wedge \frac{\delta}{\delta A}+ dA\wedge \frac{\delta}{\delta A^\dag}+
dA^\dag\wedge \frac{\delta}{\delta c^\dag})
\]
\[
S_N=\frac{1}{2}\int_N \cA\wedge d\cA=\frac{1}{2}\int_N(A\wedge dA+A^\dag\wedge dc+c\wedge dA^\dag)
\]

The Euler-Lagrange space is
\[
\EL_N=\Omega_{closed}^\bullet(N)[1]
\]
\subsubsection{The boundary BFV theory}
Boundary fields are pullbacks of the bulk fields to the
boundary.
\[
\cF_{\pa N}=\Omega^\bullet(N)[1]
\]
We will use the same notation for pullbacks as for bulk fields.
This means $0,1,2$ forms will be denoted by $c, A, A^\dag$ respectively. They have ghost numbers $1,0,-1$.

The one form $\alpha_{\pa N}$, the symplectic structure, the
vector field $Q$ and the action for the boundary BFV theory are:
\[
\alpha_{\pa N}=\frac{1}{2}\int_{\pa N} \cA\wedge \delta\cA=\frac{1}{2}
\int_{\pa N}(A\wedge \delta A+c\wedge \delta A^\dag+A^\dag\wedge \delta c)
\]

\[
\omega_{\pa N}=\frac{1}{2}\int_{\pa N} \delta\cA\wedge \delta\cA=\int_{\pa N}\left(\frac{1}{2}\,\delta A\wedge \delta A+\delta c\wedge \delta A^\dag\right)
\]
\[
Q_{\pa N}=\int_{\pa N}d\cA\wedge \frac{\delta}{\delta \cA}=\int_{\pa N}\left(dc\wedge \frac{\delta}{\delta A}+ dA\wedge \frac{\delta}{\delta A^\dag}\right)
\]
\[
S_{\pa N}=\frac{1}{2}\int_{\pa N} \cA\wedge \delta\cA=\int_{\pa N}c\wedge dA
\]
The boundary Euler-Lagrange space is
\[
\EL_{\pa N}=\Omega^\bullet_{closed}(\pa N)[1]
\]
The evolution relation $\cL_N\subset \EL_{\pa N}$ is the subspace of forms in $\Omega^\bullet_{closed}(\pa N)[1]$ which continue to closed forms on $N$.

Abelian Chern-Simons theory is regular, which is proven similarly to Proposition \ref{prop: ab BF is regular}.

\subsubsection{Reduced BV-BFV theory}
\label{ab CS reduced}
Reduced bulk and boundary Euler-Lagrange spaces
are
\[
\EL_N/Q\simeq H^\bullet(N)[1], \ \ \EL_{\pa N}/Q\simeq H^\bullet(\pa N)[1]
\]
respectively. The symplectic form on $\EL_{\pa N}/Q$ is
given by the Poincar\'e duality.

The pullback mapping $\pi: \cF_N\to \cF_{\pa N}$ induces
the mapping of reduced spaces:
\[
\pi_*: \EL_N/Q\to \EL_{\pa N}/Q
\]
The reduced evolution relation $\cL_N/Q=Im(\pi_*)\subset \EL_{\pa N}/Q $ is a Lagrangian subspace.

Because $\pi_*$ is a linear mapping, its fiber over any point
of $\EL_{\pa N}/Q$ is simply $\ker(\pi_*)$. This space
\[
\ker(\pi_*)\simeq H^\bullet(N,\pa N)/H^{\bullet -1}(\pa N)[1]
\]
has a natural symplectic structure of degree $-1$ coming from the Lefschetz duality.

\subsubsection{The $gh=0$ part of the theory}
The $gh=0$ part of the space of fields is $F_N=\Omega^1(N)$. The $gh=0$ part of
$Vect_Q$ gives the gauge action of $\Omega^0(N)$ on the
space of fields: $A\to A+d\beta$ where $A\in \cF_N$ and $\beta\in \Omega^0(N)$. The abelian Chern-Simons action
\[
S^{cl}_N=\frac{1}{2}\int_N A\wedge dA
\]
is gauge invariant when $\pa N=\emptyset$. When $N$ has non-empty boundary the gauge transformation generated by $\beta$ changes the action by $\int_{\pa N} \beta\wedge A$.


The space $EL_N$ is the space of closed $1$-forms on $N$.
The space of boundary fields $F_{\pa N}=\Omega^1(\pa N)$ is
exact symplectic with
\[
\alpha_{\pa N}=\frac{1}{2}\int_{\pa N} A\wedge \delta A, \ \ \omega_{\pa N}=\delta\alpha_{\pa N}=\frac{1}{2}\int_{\pa N}\delta A\wedge \delta A
\]
The $gh=0$ part of the boundary Euler-Lagrange space, the  space $C_{\pa N}$, is the space of
closed $1$-forms on $\pa N$. It is clearly coisotropic.
The gauge action is Hamiltonian with the momentum map $\mu=d: \Omega^1(\pa N)\to \Omega^2(\pa N)$, where we consider $\Omega^2(N)$ as the dual space to the abelian Lie algebra of $0$-forms.

The $gh=0$ part of the moduli space $\EL_N/Q$ is the space of gauge orbits $EL_N/G_N$, and we have the natural isomorphism $EL_N/G_N\simeq H^1(N)$. The $gh=0$ part of
the moduli space $\EL_{\pa N}/Q$ is isomorphic to the space $C_{\pa N}/G_{\pa N}$ of gauge orbits through $C_{\pa N}$, or equivalently, since the action is Hamiltonian, this space is the symplectic reduction of $C_{\pa N}$. It is clear that we have the natural isomorphism
$C_{\pa N}/G_{\pa N}\simeq H^1(\pa N)$. This space is symplectic with
the symplectic structure given by the Poincar\'e duality.

The pullback to the boundary $\pi: F_N\to F_{\pa N}$ induces the
restriction mapping $\pi_*: H^1(N)\to H^1(\pa N)$.
The subspace
\[
L_N=Im(\pi_*)\subset  H^1(\pa N)
\]
is Lagrangian. The fiber over any point of $C_{\pa N}/G_{\pa N}$ is
$\ker(\pi_*)$.

\subsection{Non-abelian Chern-Simons theory}
In this case the target manifold is constructed from a Lie algebra $\g$ with
an invariant scalar product (for example a simple Lie algebra).
The target manifold is described in details in section \ref{aksz-ex}.
We assume that the Lie algebra and the corresponding
simply connected Lie group are matrix groups
and will write $tr(ab)$ for the Killing form evaluated on
two Lie algebra elements.

\subsubsection{The bulk BV theory}

Fields are graded connections in a principal $G$-bundle
over the spacetime $N$. We assume the bundle is trivial,
so the space of fields is
\[
\cF_N=\Omega^\bullet(N, \g)[1] 
\]
Here $\Omega^\bullet(N, \g)[1]=\Omega^\bullet(N)\otimes \g[1]$.
The fields corresponding to forms of degree $0,1,2,3$ will be denoted by $c, A, A^\dag, c^\dag$ respectively. The ghost numbers
are $1, 0, -1, -2$. We will write $\cA=c+A+A^\dag+c^\dag$ for the
BV superfield.

The symplectic form,  the vector field
$Q$, and the classical action are:
\[
\omega_N=\frac{1}{2}\int_N tr(\delta\cA\wedge \delta\cA)=\int_Ntr(\delta c\wedge \delta c^\dag+\delta A\wedge \delta A^\dag),
\]
\begin{multline}
Q_N=\int_N tr\left((d\cA+\frac{1}{2}[\cA, \cA])\wedge \frac{\delta}{\delta \cA}\right)=\\ =
\int_N tr\left(d_Ac\wedge \frac{\delta}{\delta A}+ (F(A)+[c,A^\dag])\wedge \frac{\delta}{\delta A^\dag}+
(d_AA^\dag+[c,c^\dag])\wedge \frac{\delta}{\delta c^\dag}+
\frac{1}{2}[c,c]\wedge \frac{\delta}{\delta c}\right),
\end{multline}
\begin{multline}
S_N=\int_N tr\left(\frac{1}{2}\cA\wedge d\cA+\frac{1}{6}\cA\wedge[\cA, \cA]\right)=\\ =
\int_N tr\left(\frac{1}{2}A\wedge dA+\frac{1}{6}A\wedge[A, A]+\frac{1}{2}A^\dag\wedge d_Ac+\frac{1}{2}c\wedge d_AA^\dag+
\frac{1}{2}c^\dag\wedge[c,c]\right)
\end{multline}

The Euler-Lagrange space is the space of flat graded connections
\[
\EL_N=\{\cA\in \cF_N| \ \ d\cA+\frac{1}{2}[\cA, \cA]=0\}
\]

In coordinate fields the Euler-Lagrange space consists of
$c, A, A^\dag, c^\dag$ which satisfy
\[
[c,c]=0, \ \ d_Ac=0, \ \ F(A)+[c,A^\dag]=0, \ \ d_AA^\dag+[c,c^\dag]=0
\]
The tangent space at $\cA\in \EL_N$ is kernel of $\hat{Q}_\cA$.
It consists of elements $(\gamma,\alpha,\alpha^\dag, \gamma^\dag)
\in \Omega^\bullet(N)[1]$ such that
\begin{multline}\label{tspace-blk}
[\gamma,c]=0,
\ \ d_A\gamma+[\alpha,c]=0, \ \ d_A\alpha+[c,\alpha^\dag]+[\gamma, A^\dag]=0, \\
d_A\alpha^\dag+[c,\gamma^\dag]+[A^\dag, \alpha]+[c^\dag ,\gamma]=0,
\end{multline}

The point $\cA$ belongs to the $gh=0$ part of $\EL_N$
if and only if $c=c^\dag=A^\dag=0$ and $F(A)=0$. The $gh=0$ part of $\EL_N$ is the space flat connections. The tangent space
to $\EL_N$ at such point is naturally isomorphic to
\begin{equation}\label{tscspace-blk}
\Omega^\bullet_{d_A-closed}(N,\g)
\end{equation}
where forms are closed with respect to the differential
$d_A=d+[A, \cdot  ]$.

In this case the reduced tangent space is
\[
\ker(\hat{Q}_\cA)/Im(\hat{Q}_\cA)\simeq H^\bullet_{d_A}(N,\g)
\]


\subsubsection{The boundary BFV theory}
Boundary fields in the non-abelian Chern-Simons theory are pullbacks of the bulk fields to the
boundary.
\[
\cF_{\pa N}=\Omega^\bullet(N,\g)[1]
\]
We will use the same notation for pullbacks as for bulk fields.
This means $0,1,2$ forms will be denoted by $c, A, A^\dag$ respectively. They have ghost numbers $1,0,-1$.

The one form $\alpha_{\pa N}$, the symplectic structure, the
vector field $Q$ and the action for the boundary BFV theory are:
\[
\alpha_{\pa N}=\frac{1}{2}\int_{\pa N}tr( \cA\wedge \delta\cA)=\frac{1}{2}
\int_{\pa N}tr(A\wedge \delta A+c\wedge \delta A^\dag+A^\dag\wedge \delta c),
\]

\[
\omega_{\pa N}=\frac{1}{2}\int_{\pa N} tr(\delta\cA\wedge \delta\cA)=\int_{\pa N} tr\left(\frac{1}{2}\delta A\wedge \delta A+\delta c\wedge \delta A^\dag\right),
\]
\begin{multline*}
Q_{\pa N}=\int_{\pa N}tr\left(\left(d\cA+\frac{1}{2}[\cA, \cA]\right)\wedge \frac{\delta}{\delta \cA}\right)=\\ =\int_{\pa N}tr\left(d_Ac\wedge \frac{\delta}{\delta A}+ (F(A)+[c,A^\dag])\wedge \frac{\delta}{\delta A^\dag}+\frac{1}{2}[c,c]\wedge\frac{\delta}{\delta c}\right),
\end{multline*}
\[
S_{\pa N}=\int_{\pa N}tr\left(\frac{1}{2} \cA\wedge d\cA+\frac{1}{6}\cA\wedge [\cA, \cA]\right)=\int_{\pa N} tr\left(c\wedge F(A)+\frac{1}{2}[c,c]\wedge A^\dag\right)
\]
The boundary Euler-Lagrange space is the space of graded
flat $\g$-connections on $\pa N$:
\[
\EL_{\pa N}=\{\cA\in \cF_{\pa N}| d\cA+\frac{1}{2}[\cA, \cA]=0\}
\]
The evolution relation $\cL_N\subset \EL_{\pa N}$ is the subspace of graded flat $\g$-connections on $G\times \pa N$ which continue
to flat graded $\g$-connections on $G\times N$.

The flatness of the graded connection $\cA$ means its graded components satisfy
\[
[c,c]=0, \ \ d_Ac=0, \ \ F(A)+[c,A^\dag]=0,
\]
The tangent space $T_\cA\EL_{\pa N}\subset T_\cA\cF_{\pa N}$ is the kernel of $\hat{Q}_\cA$:
\begin{equation}\label{tspace-bd}
\ker(\hat{Q}_\cA)=\{(\gamma, \alpha, \alpha^\dag)| [\gamma,c]=0,
\ \ d_A\gamma+[\alpha,c]=0, \ \ d_A\alpha+[c,\alpha^\dag]+[\gamma, A^\dag]=0 \}
\end{equation}

In $gh=0$ part of $\EL_{\pa N}$ we have $c=A^\dag=0$ and $F(A)=0$
\begin{equation}\label{tscspace-bd}
\ker(\hat{Q}_\cA)=\{(\gamma, \alpha, \alpha^\dag)|
\ \ d_A\gamma=0, \ \ d_A\alpha=0 \}\simeq \Omega^\bullet_{d_A-closed}(\pa N,\g)[1]
\end{equation}

\[
Im(\hat{Q}_\cA)=\{(0, d_A\tilde{\gamma}, d_A\tilde{\alpha})\}
\]
Thus, in this case the reduced tangent space is
\[
\ker(\hat{Q}_\cA)/Im(\hat{Q}_\cA)\simeq H^\bullet_{d_A}(\pa N,\g)[1]
\]

The smooth locus of the Euler-Lagrange space $\EL_{\pa N}$ is
the vector bundle over the space of flat connections on $\pa N$  with fibers
$\oplus_{i\neq 1}\Omega^i_{d_A-closed}(\pa N)$. The smooth locus of the boundary $\EL$-moduli space for the non-abelian Chern-Simons theory is isomorphic to a vector bundle over the smooth locus of the representation variety $\mr{Hom}(\pi_1(\pa N), G)/G$ with fiber $H^0_{d_A}(\pa N,\g)\oplus
H^2_{d_A}(\pa N,\g)$ with the symplectic structure naturally
extending the Atiyah-Bott symplectic form on the representation
variety.

Regularity for non-abelian Chern-Simons theory can be proven using homological perturbation theory around the abelian Chern-Simons theory. This argument extends to general AKSZ theories.

\subsubsection{Reduced BV-BFV theory}
\label{non-ab CS reduced}
Reduced Euler-Lagrange spaces for the bulk and for the boundary are the spaces of leaves of the foliation $Vect_Q$ on $\EL_N$
and $\EL_{\pa N}$ respectively. We will write
\[
\EL_N/Q_N=\{\cA\in \cF_N| d\cA+\frac{1}{2}[\cA, \cA]=0\}/\{\cA\mapsto \cA+ d\lambda+ [\cA,\lambda]|\lambda\in \Omega^\bullet(N)\}
\]
and similarly for $\pa N$. These notations indicate that leaves of $Vect_Q$ should be considered as gauge orbits on $\EL_N$.

The projection $\pi: \cF_N\to \cF_{\pa N}$ (the pullback to the boundary) defines
the mapping
\[
\pi_*: \EL_N/Q_N\to \EL_{\pa N}/Q_{\pa N}
\]
Its image $\cL_N/Q\subset \EL_{\pa N}/Q$ is Lagrangian.


Tangent spaces to $\EL_N$ and $\EL_{\pa N}$ are kernels of
corresponding operators $\hat{Q}$ and are described in components
in (\ref{tspace-blk}), (\ref{tspace-bd}). The image of $\hat{Q}$ is
also easy to compute. For the bulk we have
\[
Im(\hat{Q}_\cA)=\{[\tilde{\gamma},c], d_A\tilde{\gamma}+[\tilde{\alpha},c], d_A\tilde{\alpha}+[c,\tilde{\alpha}^\dag]+[\tilde{\gamma}, A^\dag], d_A\tilde{\alpha}^\dag +[c,\gamma^\dag]+[A^\dag, \tilde{\alpha}]+[c^\dag,\gamma] \}
\]
for the boundary:
\[
Im(\hat{Q}_\cA)=\{([\tilde{\gamma},c], d_A\tilde{\gamma}+[\tilde{\alpha},c], d_A\tilde{\alpha}+[c,\tilde{\alpha}^\dag]+[\tilde{\gamma}, A^\dag])\}
\]
The reduced spaces are quotient spaces $\ker(\hat{Q})/Im(\hat{Q})$.
In other words these are cohomology spaces of the cochain complexes $\Omega^\bullet(N)[1]$ and $\Omega^\bullet(\pa N)[1]$ respectively
with respect to the differential $d_\cA\omega=[c,\omega]+d_A\omega+[A^\dag, \omega]$.

When $\cA$ is of degree zero, i.e. $c=c^\dag=A^\dag=0$ and $F(A)=0$ we have
\[
T_\cA\EL_N/Q_N=H^\bt_{d_A}(N,\g)[1], \ \ T_\cA\EL_{\pa N}/Q_{\pa N}=H^\bt_{d_A}(\pa N,\g)[1]
\]

\subsubsection{The $gh=0$ part of the theory}
The non-reduced $gh=0$ part of the theory has the space
of fields $F_N=\Omega^1(N,\g)$. In case of a general principal $G$-bundle, this is the space of connections. The classical action is the Chern-Simons functional:
\[
S_N^{cl}=\int_N tr\left(\frac{1}{2}A\wedge dA+\frac{1}{6}A\wedge [A, A]\right)
\]
Its critical points are flat connections: $F(A)=dA+\frac{1}{2}[A, A]=0$. The $\g$-valued zero-forms
act by infinitesimal gauge transformations $A\mapsto A+d_A\alpha$,
$\alpha\in \omega^0(N,\g)$.

The $gh=0$ part of the space of boundary fields is $F_{\pa N}=\Omega^1(N,\g)$. This is an exact symplectic space with
\[
\alpha_{\pa N}=\frac{1}{2}\int_{\pa N} tr(A\wedge \delta A),  \ \
\omega_{\pa N}=\delta\alpha_{\pa N}=\frac{1}{2}\int_{\pa N} tr(\delta A\wedge \delta A)
\]
The degree zero part $C_{\pa N}$ of the boundary
Euler-Lagrange space is the space of flat connections
in a trivial $G$-bundle over $\pa N$. It is coisotropic
in the space of all connections.

The gauge group $G_{\pa N}=\Map(\pa N, G)$ acts on
$F_{\pa N}$ by Hamiltonian transformations. Infinitesimally, the action is $A\to A+d_A\alpha$. The momentum map
$\mu: \Omega^1(N,\g)\to \Omega^2(N,\g)$ is the curvature.
The symplectic reduction of $C_{\pa N}$ coincides with
the Hamiltonian reduction. The reduced space $C_{\pa N}/G_{\pa N}$ is the moduli space of flat $G$-connections on the trivial $G$-bundle over $\pa N$.

The reduced $gh=0$ part of the Euler-Lagrange space is
the moduli space of flat connections in the trivial $G$-bundle over $N$. 
The mapping $\pi_*: \cM_N^G \to \cM_{\pa N}^G$ is
the natural restriction mapping of representations of $\pi_1(N)$ to representations of $\pi_1(\pa N)\subset \pi_1(N)$. The image of this mapping is the Lagrangian subvariety in $\cM_{\pa N}^G$, the
fibers are those representations of $\pi_1(N)$ which restrict trivially to $\pi_1(\pa N)$. For detailed
exposition of the global aspects of the classical Chern-Simons theory see \cite{Fr}. 

\subsection{Non-abelian $BF$ theory}
The target space for the non-abelian $BF$ theory is
described in section \ref{aksz-ex}.

\subsubsection{The bulk BV theory} The space of fields in
the BV-extended non-abelian $BF$ theory is
\[
\cF_N=\Omega^\bullet(N,\g)[1]\oplus \Omega^\bullet(N,\g)[n-2]
\]

\[
\omega_N=\int_N tr(\delta\cB\wedge \delta\cA)
\]
\[
Q_N=\int_Ntr\left(\left(d\cA+\frac{1}{2}[\cA,\cA]\right)\wedge \frac{\delta}{\delta \cA}+ d_\cA \cB\wedge \frac{\delta}{\delta B}\right)
\]
\[
S_N=\int_N tr\left(\cB\wedge \left(d\cA+\frac{1}{2}[\cA,\cA]\right)\right)
\]

The Euler-Lagrange space:
\[
\EL_N=\{(\cA,\cB)\in \cF_N|d\cA+\frac{1}{2}[\cA, \cA]=0, \ \ d_{\cA}\cB=0\}
\]

As in the Chern-Simons theory we are interested in the
smooth locus of the Euler-Lagrange space which is, in this case,
a vector bundle over the space of flat connections
on $G\times N$ with fiber $(\oplus_{i\neq 1}\Omega^i_{d_A-closed}(N,\g)[1])\oplus  \Omega^\bullet_{d_A-closed}(N,\g)[n-2]$ over a flat connection $A$.

\subsubsection{The boundary BFV theory}
The space of boundary fields is
\[
\cF_{\pa N}=\Omega^\bullet(\pa N, \g)[1]\oplus \Omega^\bullet(\pa N,\g)[n-2]
\]
The BFV structure on it is given by the AKSZ construction:
\[
\alpha_{\pa N}=\int_{\pa N} tr( \cB\wedge \delta\cA), \ \  \omega_{\pa N}=\delta\alpha_{\pa N}=\int_{\pa N} tr (\delta\cA\wedge \delta\cB)
\]
\[
Q_{\pa N}=\int_{\pa N} tr\left(\left( d\cA+\frac{1}{2}[\cA, \cA]\right)\wedge \frac{\delta}{\delta \cA}+d_{\cA}\cB\wedge \frac{\delta}{\delta \cB}\right)
\]
\[
S_{\pa N}=\int_{\pa N} tr\left( \cB\wedge \left(d\cA+\frac{1}{2}[\cA,\cA]\right)\right)
\]

The boundary Euler-Lagrange space is
\[
\EL_{\pa N}=\{(\cA,\cB)\in \cF_{\pa N}|d\cA+\frac{1}{2}[\cA, \cA]=0, \ \ d_{\cA}\cB=0\}
\]

\subsubsection{The $gh=0$ part of the theory}
The degree zero gauge theory has the space of fields
$F_N=\Omega^1(N,\g)\oplus \Omega^{n-2}(N,\g)$ we will denote
$1$-forms $A$ and $(n-2)$-forms $B$.
The action is
\[
S^{cl}_N=\int_N tr(B\wedge F(A))
\]
Euler-Lagrange equations are
\[
F(A)=0, \ \ d_AB=0
\]
Infinitesimal gauge transformations act as
\[
A\to A+d_A\mu, \ \ B\to B+[B,\mu]+d_A\lambda
\]

The space of boundary fields is the pullback of bulk fields.
The restriction mapping is the pullback. The symplectic structure
on the space of bulk fields is exact with $\omega_N=\delta\alpha_N$
where
\[
\alpha_N=\int_{\pa N} tr(B\wedge \delta A), \ \ \omega_N=\int_{\pa N} tr(\delta A\wedge \delta B)
\]
The 
coisotropic submanifold $C_{\pa N}$ is the $gh=0$ part of $\EL_{\pa N}$. The Lagrangian submanifold $L_N=\pi(EL_N)$ consists of pairs $(A,B)$ where $A$ is a
flat connection on $\pa N$ which continues to a flat connection on $N$ and $B\in \Omega^{n-2}(\pa N,\g)$ is horizontal with respect to the flat connection $A$ and extends to a horizontal $(n-2)$-form on $N$.

The reduced space $EL_N/G_N$ is a vector bundle over
the moduli space of flat connections on the trivial $G$-bundle on $N$ with fiber $H_A^{n-2}(N,\g)$ over a gauge class $[A]$.

The moduli space $EL_{\pa N}/G_{\pa N}$ has the same structure and can be identified with  $T^*M^G_{\pa N}$ .

\subsection{$BF+B^2$ theory}
In this theory the space is $4$-dimensional.
The target space is the same as the one for
the $BF$ theory.

\subsubsection{The bulk BV theory}
The
space of fields is the same as in the $4$-dimensional
$BV$-extended $BF$-theory with the same symplectic structure.

The action and the vector field $Q_N$ are:
\[
Q_N=tr \int_N \left(\left(d\cA+\frac{1}{2}[\cA, \cA]+\cB\right)\wedge \frac{\delta}{\delta \cA}+ d_{\cA}\cB\wedge \frac{\delta}{\delta \cA}\right)
\]
\[
S_N=tr \int_N \left(\cB\wedge \left(d\cA+\frac{1}{2}[\cA, \cA]\right)+\frac{1}{2}\cB\wedge\cB\right)
\]

The Euler-Lagrange space:
\[
\EL_N=\{(\cA,\cB)\in \F_N \;|\; d\cA+\frac{1}{2}[\cA, \cA]+\cB=0,\; d_{\cA}\cB=0 \}\simeq \Omega^\bullet(N,\g)[1]
\]
Indeed, the first condition $F(\cA)+\cB=0$ gives no restriction on
$\cA$, and the second condition $ d_{\cA}\cB=0$ follows from the
first one and from the Bianchi identity.

\subsubsection{The boundary BFV theory}
The space of boundary fields is the same as for the
$BF$ theory with the same symplectic structure of
degree $0$.

The boundary vector field $Q_{\pa N}$ and the boundary action
are
\[
Q_{\pa N}=tr\int_{\pa N} \left(\left(d\cA+\frac{1}{2}[\cA, \cA]+\cB\right)\wedge \frac{\delta}{\delta \cA}+ d_{\cA}\cB\wedge \frac{\delta}{\delta \cA}\right)
\]
\[
S_{\pa N}=tr \int_{\pa N} \left(\cB\wedge \left(d\cA+\frac{1}{2}[\cA, \cA]\right)+\frac{1}{2}\cB\wedge\cB\right)
\]

The boundary Euler-Lagrange space:
\[
\EL_{\pa N}\simeq \Omega^\bullet(\pa N)[1]
\]

\subsubsection{The $gh=0$ part of the theory}
The space of fields is the same as the space of fields
of the non-abelian $BF$ theory in four dimensions.
The action is
\[
S^{cl}_N=tr\int_N \left(B\wedge F(A)+\frac{1}{2}B\wedge B\right)
\]
Euler-Lagrange equations are
\[
B+F(A)=0, \ \ d_AB=0
\]
The second equation follows from the first and from
the Bianchi identity. Infinitesimal gauge transformations are
\[
A\mapsto A+d_A\alpha-\beta, \ \ B\mapsto B+[B,\alpha]+d_A\beta
\]

The reduced theory is trivial:
\[
\EL_N/Q=\{0\}, \ \ \EL_{\pa N}/Q=\{0\}
\]
so the $gh=0$ part of it also trivial.

\subsection{The Poisson sigma model}
The target space for the Poisson sigma model is described in
section \ref{aksz-ex}. In this section $M$ is a Poisson manifold with the Poisson tensor
$\pi(x)$. In local coordinates $\pi(x)=\sum_{ab}\pi^{ab}(x)\pa_a\pa_b$.

\subsubsection{The bulk BV theory} In case of the Poisson sigma model the spacetime is two dimensional. The space of fields
in the models is $\cF_N=\Map(T[1]N, T^*[1]M)$. We will use coordinate field $\tX^a=X^a+\eta^{+,a}+\beta^{+,a}$ for coordinates $q^a$ on $M$
and $\teta_a=\beta_a+\eta_a+X^\dag_a$ for coordinates $p_a$ in the cotangent directions. Components of fields $\tX^a$ and $\teta_a$
are forms of degree $0,1,2$ respectively. They have ghost numbers $gh(X)=0,\; gh(\eta^\dag)=-1,\; gh(\beta^\dag)=-2,\; gh(\beta)=1,\; gh(\eta)=0,\; gh(X^\dag)=-1$. The symplectic form and the action functional are
given by the AKSZ construction:
\[
\omega_N=\int_N\sum_a \delta \teta_a\wedge \delta \tX^a=\int_N \sum_a( \delta X^a\wedge \delta X^\dag_a+
\delta\eta^{+,a}\wedge \delta\eta_a+\delta\beta^{+,a}\wedge X^\dag_a)
\]
\[
S_N=\int_N \sum_a \left(\teta_a\wedge d\tX^a+\frac{1}{2}\int_N \sum_{ab} \pi^{ab}(\tX)\wedge\teta_a\wedge \teta_b\right)
\]
In  field components this action is
\begin{multline}
S_N=\int_N\left(\eta_a\wedge dX^a-\eta^{+,a}\wedge\beta_a+\frac{1}{2}
\pi^{ab}(X)\wedge\eta_a\wedge\eta_b+ \pi^{ab}(X)\wedge\beta_a\wedge X^\dag_b+\right. \\
\left.\eta^{+,a}\wedge\pa_a\pi^{bc}(X)\wedge\beta_b\wedge\eta_c+
\frac{1}{2}\beta^{+,c}\wedge\pa_c\pi^{ab}(X)\wedge\beta_b\wedge\beta_c+
\frac{1}{4}\eta^{+,c}\wedge\eta^{+,d}\wedge\pa_c\pa_d\pi^{ab}(X)\wedge\beta_a\wedge
\beta_b\right)
\end{multline}
The vector field $Q$ is
\[
Q_N=\int_N\left((d\tX^a+\pi^{ab}(\tX)\wedge\teta_a\wedge\teta_b)\wedge \frac{\delta}{\delta \tX^a}+(d\teta_a+\frac{1}{2}\pa_a\pi^{bc}(\tX)\wedge\teta_b\wedge\teta_c)
\wedge \frac{\delta}{\delta \teta_a}\right)
\]

The Euler-Lagrange space $\EL_N$ is the space of zeroes of $Q$ (the space of critical points of $S_N$), or the space of
solutions to
\[
d\tX^a+\pi^{ab}(\tX)\wedge\teta_a\wedge\teta_b=0, \ \ d\teta_a+\frac{1}{2}\pa_a\pi^{bc}(\tX)\wedge\teta_b\wedge\teta_c
\]

\subsubsection{The boundary BFV theory} We will use the same notations for coordinate fields in the space of boundary fields
$\cF_{\pa N}=\Map(T[1]\pa N, T^*[1]M)$. In terms of this coordinate fields the BFV data for the Poisson sigma model are:
\[
\alpha_{\pa N}=\int_{\pa N} \teta_a\wedge \delta\tX^a
\]
\[
\omega_{\pa N}=\int_{\pa N} \delta\teta_a\wedge \delta\tX^a
\]
\[
Q_{\pa N}=\int_{\pa N}\left((d\tX^a+\pi^{ab}(\tX)\wedge\teta_a\wedge\teta_b)\wedge \frac{\delta}{\delta \tX^a}+(d\teta_a+\frac{1}{2}\pa_a\pi^{bc}(\tX)\wedge\teta_b\wedge\teta_c)
\wedge \frac{\delta}{\delta \teta_a}\right)
\]
\[
S_{\pa N}=\int_{\pa N} \left( \sum_a \teta_a\wedge d\tX^a+\frac{1}{2} \sum_{ab} \pi^{ab}(\tX)\wedge\teta_a\wedge \teta_b\right)
\]

\appendix

\section{Coisotropic submanifolds and reduction}\label{co-isotr}

In this paper we use the definition of Lagrangian submanifolds
as submanifolds which are both isotropic and coisotropic. This is because many of the symplectic spaces we work with are
infinite dimensional. In finite dimensional symplectic manifolds,
Lagrangian submanifolds have half-dimension of the total
manifold. In the infinite dimensional case this property
becomes meaningless.

Let $W$ be a presymplectic space (possible infinite-dimensional) and $K\subset W$ be the kernel of its presymplectic form. Then $W'=W/K$ is symplectic. Denote by $p: W\to W'$ the natural projection map.
For a subspace $L\subset W$ define $L'\subset W'$ as $L'=p(L)=L/L\cap K$.

It is clear that $p^{-1}(L)=L+K$, and it is easy to see that
$p^{-1}((L')^\perp)=L^\perp$. Here $A^\perp$ is the orthogonal
space with respect to the symplectic (presymplectic) form.

\begin{Proposition}\label{lagr} The subspaces $L$ and $L'$ have the following properties:
\begin{enumerate}
\item $L$ is isotropic if and only if $L'$ is isotropic.
\item If $L$ is Lagrangian, then $L'$ is also Lagrangian.
\item If $L'$ is Lagrangian and $K\subset L$, then $L$ is Lagrangian.
\end{enumerate}
\end{Proposition}

\begin{proof} For $\xi, \eta\in W$, denote by $[\xi], [\eta]\in W'$ the
equivalence classes of these elements. By definition of the symplectic structure on $W'$:
\[
([ \xi ], [ \eta ])_{W'}=(\xi, \eta)
\]
The first statement is obvious from this definition.

Now assume that $L$ is coisotropic, i.e. if  $(\xi, \eta)=0$ for
any $\eta\in L$, then $\xi\in L$. This implies that also $(\xi+\kappa_1,\eta+\kappa_2)=0$ for each $\kappa_i\in K$ . But this means that if $([\xi], [\eta])_{W'}=0$
for each $[\eta]\in L'$ then  $[\xi]\in L'$. This
means $L'$ is contains its symplectic orthogonal and therefore
is coisotropic. We proved the second statement.

To prove the last statement, we have to prove that if $L'$ is
coisotropic and $K\subset L$ then $L$ is coisotropic. Then
the first statement implies the third one.

Assume $L'$ is coisotropic, i.e. that if $([\xi], [\eta])_{W'}=0$
for each $[\eta]\in L'$ then $[\xi]\in L'$.
But this means that if $(\xi+\kappa_1,\eta+\kappa_2)=0$ for each
$\eta\in L$ and $\kappa_1,\kappa_2\in K$ then $\xi\in L+K$.
If $K\subset L$, this means that if $(\xi,\eta)=0$ for each
$\eta\in L$ then $\xi\in L$. That is that $L\subset W$
is coisotropic.

\end{proof}

An important particular case is when $W$ is a
coisotropic subspace in a
bigger symplectic space. In this case $K=W^\perp$.
The space $W'$ is the Hamiltonian reduction of the space $W$.

\begin{remark} If $L\subset V$, the space $V$ is symplectic and
$W\subset V$ is coisotropic but we do not assume that $L\subset W$, then the first statement holds
when $L'=p(L\cap W)$ but the second and the third statements
do not hold in general unless $V$ is finite dimensional.
\end{remark}

Now, let $W$ be a presymplectic manifold and $K\subset TW$ be
the integrable distribution which is the kernel
of the presymplectic form.
Assume the space of leaves of $K$ is a smooth manifold $W'$.
Let $L\subset W$ be a submanifold, and $L'$ be the space of
leaves of $K$ which passes through $L$. Assume it is also smooth.
Then Proposition \ref{lagr} holds but the condition in the
last statement should be replaced with $K|_L\subset TL$.

\section{Some facts on graded manifolds}\label{graded}

\subsection{Graded manifolds}
Recall that a smooth super manifold $M$ with body $M_{e}$ is a sheaf of super algebras over $M_{e}$ locally isomorphic to the tensor product of the algebra  of smooth functions on the body with an exterior algebra. Namely, there is an atlas
$\{(U_{\alpha},\phi_{\alpha})\}$ of $M_{e}$ with super algebra isomorphisms
\[
\Phi_{\alpha}\colon M|_{U_{\alpha}}\to C^{\infty}(\phi_{\alpha}(U_{\alpha}))\otimes \bigwedge V^{{*}}=: \calA_{\alpha}
\]
for a fixed vector space $V$. The coordinate map $\phi_{\alpha}$s take values in some given vector space $W$. For infinite dimensional supermanifolds one needs more structure: $W$ is assumed to be a Banach or Fr\'echet space and, if $V$ is infinite dimensional, the tensor product has to be completed (the dual has also to be defined properly).

If $V=\oplus_{k\in\mathbb{Z}}V_{k}$ and $W=\oplus_{{k\in\mathbb{Z}}}W_{k}$ are $\mathbb{Z}$-graded vector spaces (it is safer to assume they have only finitely many nontrivial, but possibly infinite dimensional, summands), then $A_{\alpha}$ gets
additional structure as it contains the $\mathbb{Z}$-graded subalgebra of polynomial functions where by definition linear functions on $W_{k}$ or $V_{k}$ have degree $-k$. To extend the notion of grading to nonpolynomial functions, we introduce the local graded Euler vector field $E$ as follows: Pick graded bases (i.e., bases adapted to the decompositions)
$\{x^{i}\}$  and $\{y^{i}\}$
of $W^{*}$ and $V^{*}$, respectively; then define
\[
E := \sum_{i} |x^{i}| x^{i} \frac{\partial}{\partial x^{i}} + \sum_{i} |y^{i}| y^{i} \frac{\partial}{\partial y^{i}},
\]
where $|x^{i}|:=-k$ for $x^{i}\in (W_{k})^{*}$ and $|y^{i}|:=-k$ for $y^{i}\in (V_{k})^{*}$. Notice that this definition is independent of the choice of graded bases. We then say that a function $f$ is of degree $k$ if it satisfies $E(f)=kf$.

If it is possible to choose an atlas of $M$ as above such that all transition functions are compatible with the local graded Euler vector fields on charts,\footnote{That is, for any two charts $U_{\alpha}$ and $U_{\beta}$ we have
$E_{\beta}(\Phi_{\beta}\circ\Phi_{\alpha}^{{-1}})^{*}f=(\Phi_{\beta}\circ\Phi_{\alpha}^{{-1}})^{*}E_{\alpha}f$ for all $f \in A_{\alpha}$.} then we say that $M$ is a graded manifold. Notice that a graded manifold is a super manifold with additional structure; namely, that of a globally well-defined graded Euler vector field, which we will keep denoting by $E$.
A global function is said to have degree $k$ if it satisfies
\[
E(f)=kf.
\]
One may let $E$ act on vector fields and differential forms by Lie derivative and define degree accordingly. Namely, a vector field $X$ has degree $k$ if satisfies $[E,X]=kX$ and a differential form $\alpha$ has degree $k$ if $L_{E}\alpha:=i_{E}d\alpha+di_{E}\alpha = k\alpha$. Notice that the graded Euler vector field has degree zero. 
We recall from \cite{Royt} some useful facts  whose proof is a straightforward computation:
\begin{enumerate}
\item Let $\omega$ be a closed form of degree $m\not=0$. Then $\omega=d\theta$ with $\theta=\frac1m\iota_E\omega$.
\item Let  $X$ be a vector field of degree $l$ and
$\omega$ a closed $X$-invariant form of degree $m$ with $m+l\not=0$. Then $\iota_X\omega=d S$
with $S=\frac1{m+l}\iota_E\iota_X\omega$.
\end{enumerate}
In particular, this implies that symplectic forms of degree different from zero are automatically exact and that a symplectic cohomological vector field is automatically Hamiltonian if the degree of the symplectic form is different from $-1$.

\begin{definition}
Morphisms of graded manifolds are morphisms of the underlying supermanifolds that respect the graded Euler vector fields.
\end{definition}

\begin{remark}
A different definition of graded manifolds commonly used in the literature is that of a sheaf $M$ of graded algebras over the degree zero body $M_{0}$ locally isomorphic to the tensor product of smooth functions  on the degree zero body with the graded symmetric algebra of a graded vector space. In this setting functions (and transition functions) are polynomial in the coordinates of degree different from zero. To distinguish graded manifolds defined this way from the ones defined above we will call them polygraded manifolds. Notice that a polygraded manifold is not a super manifold with additional structure. On the other hand, if a graded manifold happens to have an atlas for which all transition functions are polynomial in the coordinates of degree different from zero, then it can be given the structure of a polygraded manifold simply by restricting the sheaf.
In all the examples discussed in this paper we could work with polygraded manifolds, but this would be problematic for functional integral quantization where one exponentiates the action and wants to consider integration.
\end{remark}

\begin{remark} A common notion in the literature is that of an N-manifold. This is a graded manifold with no coordinates of negative degree. They are commonly used as targets of the AKSZ construction. Notice that by degree reasons
all transition functions are polynomial in the coordinates of degree different from zero.
\end{remark}

\begin{remark} A special class of (poly)graded manifolds are those in which the $\mathbb{Z}$ and $\mathbb{Z}_{2}$ gradings agree. This means that $W_{2k+1}=0=V_{2k}$ for all $k$ or, equivalently, that for all $k$ homogenous sections of degree $2k$ ($2k+1$) are even (odd) with respect to the original super algebra grading. These graded manifolds occur in the BV formalism whenever no fermionic physical fields are present. In this paper we will restrict to this case throughout.
\end{remark}

\subsection{Mapping spaces}
Let $M$ and $N$ be finite dimensional super manifolds. Then the set of morphisms $\Mor(M,N)$ can naturally be given the structure of a (usually infinite dimensional) manifold (non super) as follows:
If the target is a super vector space $Z$, then $\Mor(M,Z)$ is the
(usually infinite dimensional) vector space
\begin{equation}\label{eCe}
(C^{\infty}(M)\otimes Z)_{e}= C^{\infty}(M)_{e}\otimes Z_{e} \oplus C^{\infty}(M)_{o}\otimes Z_{o}
\end{equation}
For the general case, one applies this construction to local charts of the target to define local charts of the manifold of morphisms.

It is important to extend this construction to define the mapping space $\Map(M,N)$ as a (usually infinite dimensional) super manifold with body $\Mor(M,N)$. Again this is done in terms of local charts, so it is enough to define $\Map(M,Z)$. This is just the super space $C^{\infty}(M)\otimes Z$ (of which \eqref{eCe} is the even part). The main property of the mapping space is
\[
\Mor(X\times M,N)=\Mor(X,\Map(M,N)) \quad \forall X.
\]

Vector fields on the source and on the target can naturally be lifted to the mapping space. If $X_{M}$ ($Y_{N}$) is a vector field on $M$ ($N$), we denote by $\Check X_{M}$ ($\Hat Y_{N}$) its lift. Notice that the map $Y_{N}\mapsto\Hat Y_{N}$ is a morphism of Lie algebras, whereas the map $X_{M}\mapsto\Check X_{M}$ is an antimorphism.

If $M$ and $N$ are graded manifolds, with graded Euler vector fields $E_{M}$ and $E_{N}$, then one can give $\Map(M,N)$ the structure of a (usually infinite dimensional) graded manifold with graded Euler vector field
\[
E_{\Map(M,N)}=\Hat E_{N}-\Check E_{M}.
\]
The set of morphisms from $M$ to $N$ in the category of graded manifolds can then be regarded as the degree zero submanifold of $\Map(M,N)$.

\begin{remark} In the category of polygraded manifolds, one can also define the mapping space as a polygraded manifold. In this case the local data are given by $\Map(M,Z)$ for $Z$ a graded vector space. Here $\Map(M,Z)$ is defined as the graded tensor product
\[
C^{\infty}(M)\otimes Z = \bigoplus_{k} (C^{\infty}(M)\otimes Z)_{k}
\]
with
\[
(C^{\infty}(M)\otimes Z)_{k}=\bigoplus_{j} C^{\infty}(M)_{j}\otimes Z_{k-j}.
\]
\end{remark}

\section{On smooth points of the moduli space $\EL/Q$}\label{smooth}

Here we will discuss the notion of smooth points in the $\EL$-moduli
spaces which we use throughout the paper. We will use notations:
$\FF$ for the space of BV fields, $F$ for its $\gh=0$ part,  $\EL$ for BV Euler-Lagrange space and $EL$ for its $\gh=0$ part.
We assume $\cF$ is given together with
the cohomological vector field $Q$. Denote by $Vect_Q$ the
Lie subalgebra of the Lie algebra of vector fields on $\cF$
formed by Lie brackets with $Q$ and by $G$ the distribution
on the body $F$ of $\cF$ induced by the span of $Vect_Q$. The distribution $G$ on $F$ should be regarded as infinitesimal gauge transformations.

For a  $x\in F$ the Taylor expansion of a vector field $Q$ in the formal neighborhood of x is:
$$Q^{\mathrm{formal}}_x=Q^{(0)}_x + Q^{(1)}_x+Q^{(2)}_x+\cdots$$
where each $$Q^{(k)}_x \in  \mbox{Coder}(\hat S T_x \FF)$$
is the extension of a $k$-linear map
$\hat{Q}^{(k)}_x: S^k (T_x \FF) \rightarrow T_x \FF$
to a coderivation of $\hat ST_x\FF$  by co-Leibniz identity.
The linear maps $\hat{Q}^{(k)}_x$ arise from the
Taylor expansion of $Q$ at $x$. By the natural inclusion $\mbox{Coder}(\hat S T_x \FF) \hookrightarrow \mbox{Der}(\hat{S} T^*_x \FF)$, $Q_x^{\mr{formal}}$ acts on $\hat S T^*_x \FF$ as a derivation.

If $x\in EL\subset F$, then $\hat{Q}^{(0)}_x=0$ and $Q^{\mathrm{formal}}_x$ endows $T_x[-1] \FF$ with the structure of $L_\infty$ algebra with differential $\hat{Q}^{(1)}_x$ and higher polylinear operations $\hat{Q}^{(k)}_x$, $k\geq 2$.

We define the formal neighborhood of $[x]\in EL/G$  in $\EL/Q$ as
\be U^{\mathrm{formal}}_{[x]}(\EL/Q):=\mr{Spec}(H_{Q^{\mathrm{formal}}_x}(\hat S T^*_x \FF))
\label{S-of-Q-coh}
\ee
It can be regarded as the ``BV" (or ``stable") version of the Maurer-Cartan set for the $L_\infty$ structure on $T_x[-1]\cF$.
When operations $\hat{Q}^{(k)}_x$ for $k\geq 2$ are identically zero, the spectrum (\ref{S-of-Q-coh}) is
\[
H_{\hat{Q}^{(1)}_x}(T_x \FF)
\]
which means that the formal neighborhood of $[x]$ in $\EL/Q$ is the graded vector space $H_{\hat{Q}^{(1)}_x}(T_x \FF)$.

When operations $\hat{Q}^{(k)}_x$ do not vanish for $k\geq 2$
they induce operations $\hat{Q}^{'(k)}_x$ on the cohomology space
$V_x=H_{\hat{Q}^{(1)}_x}(T_x \FF)$. These operations endow the
graded space $V_x[-1]$ with the structure of minimal $L_\infty$ algebra. In other words the formal series
\[
Q_x'=Q'^{(0)}_x + Q'^{(1)}_x+Q'^{(2)}_x+\cdots \in\mr{Coder}(\hat{S}(V_x))
\]
is a (formal) cohomological vector field on $V_x$.
Note that first two terms vanish by our assumptions.

This construction is known as homological perturbation theory
\cite{GL} (see also \cite{GKR} for a more concise exposition),
more specifically, as the homotopy transfer of $L_\infty$ algebras \cite{KS,VL}. By the same homological perturbation theory the
formal neighborhood of $[x]\in EL/G\subset \EL/Q$ is
\[
U^{\mr{formal}}_{[x]}(\EL/Q)=\mr{Spec}(H_{Q'_x}(\hat{S} (V_x)^*))
\]
This is a singular variety unless $Q'_x$ vanishes.

This is why we define \textit{a smooth point $[x]\in EL/G$}
as a point for which all induced $L_\infty$ operations $\hat{Q}^{'(k)}_x$ vanish. The set of such points we will call
the smooth locus of the body of the $\EL$-moduli space and will denote it $(EL/G)^{\mbox{smooth}}$.

The smooth locus in the $\EL$-moduli space  is a graded vector bundle over $(EL/G)^{\mbox{smooth}}$ with fiber $H^{\neq 0}_{\hat{Q}^{(1)}_x}(T_x \FF)$ over $[x]$. Here $H^{\neq 0}$ means that we do not count the cohomology in ghost number 0 to avoid double counting of geometric ($gh=0$) directions for the tangent space.


\section{Cartan calculus of local differential forms}\label{Cart-Calc}

\subsection{Transgression map from forms on the target to forms on the mapping space}
Recall that the space of fields in the AKSZ theory is
$\F_N=\Map(T[1]N, \M)$ with $\M$ a Hamiltonian dg manifold.

Consider natural projections
\[
\begin{array} {ccc}\F_N \times T[1]N& \stackrel{ev}{\rightarrow} & \M \\
p\downarrow &  & \\
\F_N & &
\end{array}
\]
Here $ev(f, x)=f(x)$ and $p(f,x)=f$.
The pullback $ev^*$ of a from $\M$ gives a
form on $\F_N \times T[1]N$. Let $\Omega$ be a form on
$\F_N \times T[1]N$ of type $(k,0)$. The pushforward $p_*$ of a form on $\F_N \times T[1]N$  gives a form on $\F_N$.
\[
p_*(\Omega)(X)=\int_{T[1]N} i_X^*(\Omega)
\]
Here  $i_X$ is the
composition of the natural isomorphism of $T[1]N$ with the fiber
$(X, T[1]N)=p^{-1}(X)$ and the embedding of this
fiber into $\F_N\times T[1]N$. The integration over $T[1]N$ is defined as usual
taking into account natural isomorphism $\phi: C^\infty(T[1]N)\simeq \Omega^\bullet(N)$. If $f\in C^\infty(T[1]N)$
\[
\int_{T[1]N}f\stackrel{\text{def}}{=}\int_{N}\phi(f)
\]
Having in mind this formula we will write for $p_*$
\[
p_*(\Omega)(X)=\int_{N} \Omega
\]

The transgression map is defined as
\[
T_N:=p_* \ev^*:\qquad \Omega^\bt(\MM)\ra \Omega^\bt(\FF_N)
\]

Given  $\alpha\in \Omega^\bullet(\M)$ which is $\alpha=\sum_{\{a\}}\alpha_{a_1,\dots, a_k}dx^{a_1}\dots dx^{a_k}$ in local coordinates, its transgression is
\begin{equation}\label{ul-form}
(T_N\alpha)(X)=\int_N \sum_{a_1,\dots, a_k}
\alpha_{a_1,\dots, a_k}(X(u))\delta X^{a_1}(u)\dots
\delta X^{a_k}(u)
\end{equation}
where $X^a(u)$ are coordinate fields and $\delta X^a(u)$ are
their de Rham differentials (on $\cF_N$).

Define the space of  {\it ultra local forms} $\Omega_{uloc}^\bullet(\F_N)$
as the image of the transgression map.

\subsection{De Rham and lifted vector fields}\label{cart-vect}

Recall that the de Rham differential for $N$ can be regarded
as a vector field $D_N$ on $T[1]N$ of degree $1$.  The
identification of forms on $N$ with functions on $T[1]N$
identifies the action of the de Rham differential on forms
with the action of the vector field $D_N$ on functions on $T[1]N$. In
local coordinates $\{u^i\}$ on $N$ and $\xi^i=du^i$ on $T_u[1]N$
the Lie derivative of a function along the de Rham vector field on $T[1]N$  is
$D_N f=\sum_i \xi^i\frac{\pa f}{\pa u^i}$.

The {\it de Rham vector field} $\Check D_N$ is defined as the lift of de Rham vector field $\xi^i \frac{\dd}{\dd u^i}$ on  $N$ to the mapping space $\FF_N$.

Indeed, the tangent space $T_X\F_N$ is the space of
mappings $Y: T[1]N\to T\M$
such that $(\xi, u)\mapsto Y(u,\xi)\in T_{X(u,\xi)}\M$.
A vector field $V$ on $\F_N$ is a section of the tangent
bundle, i.e. it assigns a vector $V_X: \F_N\to T_X\F_N$
to each $X\in \F_N$.

For de Rham vector field $\Check D_N$ on $\F_N$ the
vector $(\Check D_N)_X\in T_X\FF_N$ is the mapping
\[
(\Check D_N)_X: T[1]N\stackrel{D_N}{\rightarrow} T(T[1]N)
\stackrel{dX}{\rightarrow} T\M
\]
Here $dX: T(T[1]N)\to T\M$ is the differential of $X$.
In local coordinates:
\[
(\Check D_N)_X(\xi,u)=\sum_a\sum_{\{i\}, j}\frac{\pa X^a_{i_1,\dots, i_k}}{\pa u^j}(u)
\xi^j\xi^{i_1}\dots \xi^{i_k}\frac{\pa}{\pa x^a}
\]
The de Rham vector field on $\FF_N$ has ghost number $+1$.

Another important class of vector fields on $\FF_N$ comes from lifting vector fields on $\MM$ to the mapping space $\FF_N$.
If $v: \MM\to T\MM$ is a vector field on $\MM$, its lifting
is the vector field $\hat{v}$ on $\FF_N$ with $\hat{v}_X\in T_X \FF_N$ given by the mapping
\[
T[1]N\stackrel{X}{\rightarrow} \MM
\stackrel{v}{\rightarrow} T\M
\]

\subsection{Local differential forms}

{\it Local differential forms}\footnote{Note that, by our definition, local $0$-forms are special cases of local functionals.} are defined as substitutions of several copies of $\Check D_N$ into an ultra-local form:
\be
T^k_N\chi:= (\iota_{\Check D_N})^k T_N\chi,\qquad k\geq 0 \label{T^k def}
\ee

The space of forms $\Omega^\bullet(\MM)$ is naturally
bi-graded with (de Rham degree of forms on $\M$, $\ZZ$-grading on $\M$).
The space $\Omega^\bullet(\F_N)$ is also naturally bi-graded with (de Rham degree of forms on $\F_N$, ghost number).
Transgression mappings $T_N^k: \Omega^\bullet(M)\to
\Omega^\bullet(\F_N)$ are bi-graded with
\[
deg(T^k_N)=(-k,-\dim(N)+k)
\]

Also note that if $\chi\in \Omega^n(\MM)$ and $k>n$, by degree
reasons we have $T^k_N\chi=0$.

Let $\{x^a\}$ be local coordinates on $\M$ and $X^a(u)$ be corresponding coordinate fields. If $\alpha=\sum_{\{a\}}\alpha_{a_1,\dots, a_k} dx^{a_1}\wedge\dots
\wedge dx^{a_k}$ is a differential form on $\M$ written in local
coordinates, the corresponding local forms on $\F_N$ can be written in coordinate fields as

\[
T^l_N(\alpha)=\int_N \sum_{\{a\}=\{a_I\}\sqcup \{a_{II}\}} 
\alpha_{\{a_I\}\sqcup \{a_{II}\}}(X(u))\wedge dX^{\{a_I\}}(u)
\wedge \delta X^{\{a_{II}\}}(u)
\]
Here the sum is taken over partitions of
the set $\{a\}$ into two subsets, so that $\{a\}$ is a shuffle
of $\{a_{I}\}$ and $\{a_{II}\}$.
We denote $dX^{\{a\}}=
dX^{a_1}\wedge\dots
\wedge dX^{a_k}$.

We will use the notation $\Omega^\bullet_{loc}(\F_N)$
for the space of such forms.

\subsection{The Cartan calculus}\label{cart-calc}

Let us denote $\pi: \FF_N\ra \FF_{\dd N}$ the restriction map (pullback by the natural inclusion $\dd N\hookrightarrow N$) and let
\be \pi^*:\qquad \Omega^\bt(\FF_{\dd N})\ra \Omega^\bt (\FF_N) \label{pi^*}\ee
be the pullback by $\pi$ for differential forms on the space of fields.

Objects we introduced satisfy the following properties:
\begin{enumerate}[(i)]
\item \label{Cartan calculus: 1st rule} De Rham differential $\delta$ on $\FF_N$ acts on local forms by
\be \delta (T_N^k\chi)= (-1)^{\dim(N)}(T_N^kd\chi-k\cdot \pi^*(T_{\dd N}^{k-1}\chi)) \label{D on local forms}\ee
where $d$ is the de Rham differential on target.

\item ``Total derivative property'', which is a special case of (\ref{D  on local forms}) for $k=\deg\chi+1$:
\be
\chi\in \Omega^n(\MM) \quad \Rightarrow \qquad T^{n+1}_N(d\chi) = (n+1)\cdot \pi^*(T^n_{\dd N}\chi)
\label{total derivative}\ee

\item Substitution of $\Check D_N$ into a local form:
\be \iota_{\Check D_N} (T^k_N \chi)= T^{k+1}_N\chi \label{subst of hat D_N}\ee
by definition  (\ref{T^k def}).

\item Substitution of a lifted vector field into a local form:
\be \iota_{\hat{v}} (T^k_N \chi)= (-1)^{(\gh(v)+1)\cdot \dim(N)} T^k_N (\iota_v \chi) \label{subst L_N v}\ee
for any vector field $v\in \Vect(\MM)$ of ghost number $\gh(v)$ on target.

\item Rules for commuting the substitution of a vector field on $\FF_N$ with the map $\pi^*$  (\ref{pi^*}):
\begin{eqnarray}
\iota_{\Check D_N} \pi^* &=& \pi^* \iota_{\Check D_{\dd N}} \label{hat D_N on bdry}\\
\iota_{\hat{v}} \pi^* &=& \pi^* \iota_{\tilde{v}} \label{L_N v on bdry}
\end{eqnarray}
Here $\tilde{v}$ is the lifting of the vector field $v$ on $\M$ to a vector field on $\F_{\pa N}$. This is a manifestation of the fact that vector fields $\Check D_N, \hat{v}\in \Vect(\FF_N)$ are projectable by the restriction map $\pi: \FF_N \ra \FF_{\dd N}$ with projections being $\Check D_{\dd N}, \tilde{v}\in \Vect(\FF_{\dd N})$ respectively.

\item The Lie derivative of a local form along $\Check D_N$:
\be L_{\Check D_N}(T^k_N\chi)= (-1)^{\dim(N)} \pi^*(T^k_{\dd N} \chi) \label{Lie derivative along hat D_N}\ee
We define the Lie derivative by Cartan formula $L_V:=[\iota_V,\delta]= \iota_V\circ \delta + (-1)^{\gh(V)} \delta\circ \iota_V$ for any vector field $V\in \Vect(\FF_N)$. With this definition the formula (\ref{Lie derivative along hat D_N}) comes as a consequence of (\ref{D  on local forms}), (\ref{subst of hat D_N}), (\ref{hat D_N on bdry}).

\item Lie derivative of a local form along a lifted vector field:
\be L_{\hat{v}}(T^k_N \chi)= (-1)^{\gh(v)\cdot \dim(N)} T^k_N(L_v\chi) \label{Lie derivative along L_N v}\ee
(this is a consequence of (\ref{D  on local forms}), (\ref{subst L_N v}), (\ref{L_N v on bdry})).

\item \label{Cartan calculus: last rule} Commutators of vector fields:
\be [\Check D_N, \Check D_N] = 0 \label{hat D_N squares to zero} \ee
This is because $d^2=0$ on the source $N$. Also
\be [\Check D_N, \hat{v}] = 0 \label{hat D_N and L_N v commute} \ee
and
\be [\hat{u}, \hat{v}] = \widehat{[u,v]} \label{L_N u and L_N v commutator} \ee
\end{enumerate}

\subsection{Applications to AKSZ theories}

\subsubsection{Invariant definition of AKSZ theories}

Recall that the target manifold in an $n$-dimensional AKSZ theory is a Hamiltonian dg manifold of degree $n-1$. In other words, it is
symplectic with an exact symplectic form $\omega=d\alpha\in\Omega^2(\MM)$ and $deg(\omega)=n-1$
and with a potential function $\Theta$ of degree $n$ such that
$\{\Theta, \Theta\}=0$. The potential function generates the
cohomological vector field $Q$ of degree $1$:
\[
\iota_{Q}\omega = d \Theta
\]
By definition of $Q$, the symplectic form $\omega$ is $Q$-invariant, i.e.
\[
L_{Q} \omega =0
\]
where $L_{Q}\eta$ is the Lie derivative of $\eta$.
The condition  $\{\Theta, \Theta\}=0$ implies that
\[
L_{Q}^2=0
\]
which we will write as $Q^2=0$.

The AKSZ action is the following local functional on
the space $\Map(T[1]N\to \MM)$
$$
S_N=S^{kin}_N+S^{int}_N
$$
where
\begin{eqnarray*}
S^{kin}_N &=& T^1_N \alpha \\
S^{int}_N &=& T^0_N \Theta
\end{eqnarray*}
are kinetic and interaction parts of the action
respectively. Here $T^k_N$ are the transgression maps
defined above.

The BV cohomological vector field $Q_N$ in this
theory is defined as
$$Q_N= Q^{kin}_N+Q^{int}_N \quad \in \Vect(\FF_N)$$
with
\begin{eqnarray*}
Q^{kin}_N &=& \Check D_N \\
Q^{int}_N &=& \hat{Q}
\end{eqnarray*}
Here $\hat{Q}$ is the lift of the vector field $Q$
on $\M$ to a vector field on $\cF_N$, see subsection \ref{cart-vect}.

The BV 2-form and its primitive 1-form:
\begin{eqnarray*}
\omega_N &=& T^0_N \omega \\
\alpha_N &=& T^0_N \alpha
\end{eqnarray*}

Note that with our conventions we have an extra sign: $\omega_N=(-1)^{\dim(N)}\delta\alpha_N $.

These formulae applied to the boundary $\pa N$ of $N$
also define the BFV action $S_{\pa N}$, the cohomological
vector field $Q_{\dd N}$, BFV 2-form $\omega_{\dd N}$ and its primitive 1-form $\alpha_{\dd N}$.

\subsubsection{Proof of Proposition \ref{aksz=bvbfv}}

Using the definition of $Q_N$ and the Cartan calculus
we can compute the contraction of kinetic and interaction
parts of $Q_N$ with $\omega_N$:
\[
\iota_{Q_N^{kin}} \omega_N = \iota_{\Check D_N} (T^0_N \omega)=T^1_N \omega =
T^1_N \delta \alpha
\]
The rule (\ref{D on local forms}) from Cartan calculus on local forms implies
that this is equal to
\begin{equation}\label{Q_N^kin Ham}
(-1)^{\dim(N)} \delta T^1_N \alpha+ \pi^* (T^0_{\dd N} \alpha) =
(-1)^{\dim (N)} \delta S_N^{kin}+ \pi^* \alpha_{\dd N}
\end{equation}
Contracting $Q^{int}_N$ with $\omega_N$ we obtain
\[
\iota_{Q_N^{int}} \omega_N = \iota_{\hat{Q}} (T^0_N \omega)= T^0_N (\iota_{Q} \omega)
\]
Here in the last equality we used (\ref{subst L_N v}). Because $Q$ is a Hamiltonian vector field generated by $\Theta$, $\iota_{Q} \omega=d\Theta$ and for the above expression we obtain
\begin{equation}\label{Q_N^int Ham}
T^0_N (d \Theta) = (-1)^{\dim (N)} \delta(T^0_N \Theta)= (-1)^{\dim(N)} \delta S_N^{int}
\end{equation}
Here the first equality follows from (\ref{D  on local forms}).

Equations (\ref{Q_N^kin Ham}), (\ref{Q_N^int Ham}) together give (\ref{Q_N Ham}).

\subsubsection{Further applications to the AKSZ formalism}

Here we will reprove the Proposition \ref{LQS} using the Cartan calculus.

\begin{Proposition}The following identity holds:
\be L_{Q_N} S_N = (-1)^{\dim(N)}\pi^* (2S_{\dd N} - \iota_{Q_{\dd N}} \alpha_{\dd N})  \label{CMEaksz}\ee
\end{Proposition}
\begin{proof}
The proof is computational:
\begin{multline}
L_{Q_N^{kin}} S_N^{kin} = L_{\Check D_N} T^1_N \alpha = (-1)^{\dim(N)} \pi^* (T^1_{\dd N}\alpha) = \\ =(-1)^{\dim(N)} \pi^* (\iota_{Q_{\dd N}^{kin}} \alpha_{\dd N})= (-1)^{\dim (N)} \pi^* (2 S_{\dd N}^{kin}-\iota_{Q_{\dd N}^{kin}} \alpha_{\dd N})
\label{CME kin-kin}
\end{multline}
Here we used $\iota_{Q_{\dd N}^{kin}}\alpha_{\pa N}=S_{\dd N}^{kin}$.
\begin{multline}
L_{Q_N^{kin}} S_N^{int} = L_{\Check D_N} (T^0_N \Theta) = (-1)^{\dim (N)} \pi^*(T^0_{\dd N} \Theta ) = \\
= (-1)^{\dim (N)} \pi^* (S_{\dd N}^{int})
\label{CME kin-int}
\end{multline}

\begin{multline}
L_{Q_N^{int}} S_N^{kin} = L_{\hat Q} (T^1_N \alpha) = (-1)^{\dim (N)} T^1_N (L_{Q} \alpha) \underbrace{=}_{\mbox{Cartan formula}} \\
=(-1)^{\dim (N)} T^1_N (\iota_{Q}{d \alpha} - d \iota_{Q} \alpha)= (-1)^{\dim (N)} T^1_N d (\Theta - \iota_{Q}\alpha)
\end{multline}
Here we used the Cartan formula and the exactness of the symplectic form on the target, $\omega=d\alpha$.
Now applying (\ref{total derivative}) we obtain
\begin{equation}\label{CME int-kin}
=(-1)^{\dim (N)} \pi^* T^0_{\dd N} (\Theta - \iota_{Q}\alpha) =
(-1)^{\dim(N)}\pi^* (S_{\dd N}^{int}-\iota_{Q_{\dd N}^{int}}\alpha_{\dd N})
\end{equation}
\be
L_{Q_N^{int}} S_N^{int} = L_{\hat Q} T^0_N \Theta = (-1)^{\dim(N)} T^0_N (L_{Q} \Theta) = 0 \label{CME int-int}
\ee
Collecting (\ref{CME kin-kin})--(\ref{CME int-int}), we obtain (\ref{CMEaksz}).

\end{proof}

The following is a corollary of Proposition \ref{aksz=bvbfv} but here we give an independent proof.

\begin{Proposition}
The following holds:
\be L_{Q_N} \omega_N = (-1)^{\dim (N)} \pi^* \omega_{\dd N} \label{Q_N symp}\ee
\end{Proposition}
\begin{proof}
Indeed:
\be
L_{Q_N^{kin}}\omega_N = L_{\Check D_N} (T^0_N \omega)= (-1)^{\dim(N)}\pi^* \underbrace{T^0_{\dd N}\omega}_{\omega_{\dd N}} \label{Q_N^kin symp}
\ee
where we used the rule (\ref{Lie derivative along hat D_N}). Also,
\be
L_{Q_N^{int}}\omega_N = L_{\hat Q}(T^0_N \omega)= (-1)^{\dim(N)}T^0_N (L_{Q} \omega)=0 \label{Q_N^int symp}
\ee
where we used the rule (\ref{Lie derivative along L_N v}) and the fact that the target cohomological vector field is symplectic. Putting (\ref{Q_N^kin symp}) and (\ref{Q_N^int symp}) together we get (\ref{Q_N symp}).
\end{proof}

\end{document}